\documentclass[12pt]{amsart}
\usepackage{a4wide,enumerate,color,graphicx}
\usepackage{amsmath}
\usepackage{amssymb}
\usepackage{amsfonts}
\usepackage{amsthm}
\usepackage{enumerate}
\usepackage{mathrsfs}
\usepackage{mhchem}     
\newtheorem{lemma}{Lemma}[section]
\newtheorem{prop}[lemma]{Proposition}

\newtheorem{defin}[lemma]{Definition}
\newtheorem{rem}[lemma]{Remark}

\DeclareMathOperator{\divv }{div}

\allowdisplaybreaks

\begin{document}

\title[Multicomponent incompressible fluids]{Multicomponent incompressible fluids--An asymptotic study}

\author[D.~Bothe]{Dieter Bothe}
\address{Mathematische Modellierung und Analysis, Technische Universit\"at Darmstadt, Alarich-Weiss-Str. 10, 64287 Darmstadt, Germany}
\email{bothe@mma.tu-darmstadt.de}

\author[W.~Dreyer]{Wolfgang Dreyer}
\address{Weierstrass Institute, Mohrenstr. 39, 10117 Berlin, Germany}
\email{Wolfgang.Dreyer@wias-berlin.de}

\author[P.-E.~Druet]{Pierre-Etienne Druet}
\address{Weierstrass Institute, Mohrenstr. 39, 10117 Berlin, Germany}
\email{pierre-etienne.druet@wias-berlin.de}

\date{\today}


\thanks{Pierre-\'{E}tienne Druet acknowledges the grant D1117/1-1 of the
German Science Foundation (DFG).}
\maketitle


\begin{abstract}
This paper investigates the asymptotic behavior of the Helmholtz free energy of mixtures at small compressibility. We start from a general representation for the local free energy that is valid in stable subregions of the phase diagram. On the basis of this representation we classify the admissible data to construct a thermodynamically consistent constitutive model. We then analyze the incompressible limit, where the molar volume becomes independent of pressure. Here we are confronted with two problems:

(i) Our study shows that the physical system at hand cannot remain incompressible for arbitrary large deviations from a reference pressure unless its volume is linear in the composition.

(ii) As a consequence of the 2$^\textrm{nd}$  law of thermodynamics, the incompressible limit implies that the molar volume becomes independent of temperature as well. Most applications, however, reveal the non-appropriateness of this property.

According to our mathematical treatment, the free energy as a function of temperature and partial masses tends to a limit in the sense of epi-- or Gamma--convergence. In the context of the first problem, we study the mixing of two fluids to compare the linearity with experimental observations. The second problem will be treated by considering the asymptotic behavior of both a general inequality relating thermal expansion and compressibility and a PDE-system relying on the equations of balance for partial masses, momentum and the internal energy.
\end{abstract}
%
%
%
%
%

\section{Introduction}



As a rule, real world fluids--solutions, electrolytes, fluid mixtures etc., but also supposedly pure substances--are composed of several constituents and are termed \emph{multicomponent fluids}.
Recently, it has turned out that new energy technologies require a consistent coupling of flow of matter, diffusion, chemical reactions and further purely mechanical phenomena. For fluid mixtures, the resulting equations raised new interest also in the PDE community. In fact, only few results are available for  reaction--diffusion systems coupled to fluid dynamical equations. Moreover, these results almost all concern so-called ideal mixtures. Non-ideal fluid mixtures are treated for example in \cite{bothedruet}. There the analysis of a general class--one model was presented, without other restrictions on the thermodynamical potential -- in this case the Helmholtz free energy -- than being a function of Legendre--type.\footnote{A function $f$ defined in an open convex set $\mathcal{D} \subseteq \mathbb{R}^N$ is called \emph{of Legendre--type} if it is continuously differentiable and strictly convex in $\mathcal{D}$ and the gradient of $f$ blows up at every point of the boundary of $\mathcal{D}$ (see \cite{rockafellar}, Section 26).} In \cite{bothedruetincompress}, the results of \cite{bothedruet} for the isothermal case  are extended to incompressible mixtures. 

For a physical body, the molar volume $\upsilon$ expresses how much volume is locally available per mole. In a multicomponent system, the amount of matter (how many moles) which is overall necessary to fill a certain volume depends on temperature and pressure, but also on the composition of the mixture. This relationship, called the {\em thermal equation of state}, is an algebraic equation of the form
\begin{align}\label{EOS}
 \upsilon = \hat{\upsilon}(T, \, p, \, x_1, \ldots,x_N) \, ,
\end{align}
in which the variables $T$ (absolute temperature), $p$ (pressure) and $x = (x_1,\ldots,x_N)$ (mole fractions, expressing the composition) describe the state of the body, and   $\hat{\upsilon}$ is a constitutive function.

We define \textit{incompressibility} of a multicomponent fluid as zero compressibility, i.e.
\begin{align}\label{incom}
 \partial_p \hat{\upsilon}(T, \, p, \, x_1, \ldots,x_N) = 0 \, ,
\end{align}
meaning the pressure--invariance of the molar volume at fixed composition and temperature.

This concept still allows for local changes of the molar volume and the total mass density due to variations of temperature and/or composition. However, in the present paper we prove by rigorous asymptotics that the thermodynamical stability of the incompressible phase imposes (in a neighborhood of some reference state) the conditions
\begin{align}\label{incomCons}
 \partial_T \hat{\upsilon} = 0 \quad \text{ and } \quad D^2_{x,x}\hat{\upsilon} = 0 \, .
\end{align}
In other words, a thermodynamically consistent \textit{equation of state} for an incompressible fluid mixture does not depend on temperature and it is linear in the composition,
\begin{align}\label{EOSinfty}
 \upsilon = \hat{\upsilon}(T^0,p^0, \, x) = \sum_{i=1}^N \upsilon^{00}_i \, x_i \, ,
\end{align}
with constants $\upsilon^{00}_1, \ldots, \upsilon^{00}_N$ \footnote{ Depending on the context we shall denote the reference values of temperature and pressure by $(T^0, \, p^0)$, respectively $(T^{\rm R}, \, p^{\rm R})$. We use the superscript $00$ (resp. ${\rm R}$) to denote the values of thermodynamic functions in the reference state.}. Similar volume constraints have been used to express incompressibility in \cite{bothedreyer}, \cite{dreyerguhlkemueller}, \cite{dreyerguhlkelandstorfer} and \cite{mills,josef,donevetal}.

As a consequence, the asymptotic model forbids thermal expansion and nonlinear volume effects in the composition variable, two conclusions that seem counter intuitive. Indeed, in Section \ref{thermexp} we consider more refined scaling arguments such that these conclusions can be weakened. To this end we study the full compressible PDE-system of mixtures of fluids. There are flow regimes, where it is necessary to distinguish between the conclusions of thermodynamically consistent models for compressibility zero, and the behavior of the solutions to the PDE-system that describe multicomponent fluids at low Mach--number.

It is well known that the chemical potentials for incompressible multicomponent fluids and for compressible mixtures in the low Mach-number regime are linear in the pressure. In fact, for both cases the chemical potentials of the $N$ constituents indexed by $i=1,...,N$ can be written as
\begin{align}\label{incomCHEMPOTa}
\mu_i =  \hat{\mu}_i(T, \, p, \, x_1, \ldots,x_N) = \upsilon_i^{00} \, p + \mu_i^{0}(T,\, x_1, \ldots,x_N)~.
\end{align} 
Note that the pressure term in the chemical potentials induces the coupling between transport of matter, diffusion and chemical reactions. That term leads to a considerable complication of the mathematical analysis. Some analysis of weak solutions of the same kind of model was also performed in \cite{druetmixtureincompweak}.

So far we already introduced two constitutive functions, viz. \eqref{EOS} and \eqref{incomCHEMPOTa}. However, the PDE-system of a fluid mixture contains further functions of that kind. For example, diffusion fluxes, reaction rates, the heat flux and the specific enthalpy are likewise given by constitutive functions. In this context there arises the problem of thermodynamic consistency of the PDE-system because the data for these functions often result from quite different sources. In this paper we derive all these functions from a potential, the Helmholtz free energy density, with special features such that thermodynamic consistency is guaranteed. But when we ask for the origin of the free energy function, we are confronted with new problems because free energy models are not directly determined in experiments. Exclusively some derivatives of the free energy are measured. For this reason we shall construct some representation theorems of the free energy function in terms of independent measurable quantities.

%

\section{Thermodynamics of multicomponent fluids}

In this section, a brief introduction of thermodynamic modelling for multicomponent mixtures is given. The model relies on the treatment that is described in the seminal handbook article by Meixner, \cite{MR59}. Further details are found in \cite{dGM84}, \cite{mueller} and \cite{bothedreyer}.

\vspace{0.2cm} {\bf Variables.} Consider a multicomponent system consisting of $N \in \mathbb{N}$ chemical species $\ce{A}_1,\ldots,\ce{A}_N$ which are assumed to constitute a fluid phase.
Locally, we characterize the mixture by the absolute temperature $T$, the partial mass densities $(\rho_i )_{i=1,2,\ldots,N}$ and the barycentric velocity $\boldsymbol v$. These quantities are the basic variables of the model.

Total mass density, total mole density and the molar volume are defined by
\begin{align}\label{thermo1}
\varrho = \sum_{i=1}^N \rho_i\, , \quad n = \sum_{i=1}^N \frac{\rho_i}{M_i} \, , \quad  \quad \upsilon = \frac{1}{n} \, ,
\end{align}
respectively. The positive constants $M_i$ denote the molar masses of the constituents. For brevity in notation we often use the notation $\rho$ for the vector $(\rho_1, \, \ldots, \rho_N)$, while $\varrho$ is the total mass density.
The molar volume is related to the mass density by
\begin{align}\label{vequiv}
 \upsilon = \frac{M}{\varrho} \, ,
\end{align}
where $M$ denotes the mean molar mass, i.e.\
\begin{align}\label{MM}
 M = M(x) := \sum_{i=1}^N M_i \, x_i \, ,
\end{align}
where the mole fractions are defined as $x_i := n_i/n = \rho_i/(M_i \, n)$.

In single-component fluids the molar volume is defined as the specific volume, namely $\upsilon=1/\varrho$. In fluid mixtures the latter definition is not appropriate. Here the volume is defined as in \eqref{vequiv} and is called molar volume, because its physical unit is [m$^3$/~mole].

%
%
%

\vspace{0.2cm} {\bf Equations of balance.} In continuum thermodynamics, the relevant equations to determine the variables are balance equations for the partial masses, the momentum and the internal energy. These balance equations read
\begin{align}\begin{split}\label{thermo2}
&\partial_t \rho_i + \divv(\rho_i \, \boldsymbol v + \boldsymbol J_i) =  r_i \,  ,\\
&\partial_t (\varrho \boldsymbol v)+\divv(\varrho \boldsymbol v \otimes \boldsymbol v + p\,\boldsymbol 1 - \boldsymbol S)=\varrho\boldsymbol b \, ,\\
&\partial_t (\varrho u) + \divv(\varrho u \, \boldsymbol v + \boldsymbol q) =(-p\boldsymbol{1}+  \boldsymbol S )\, : \, \nabla \boldsymbol v ~ .\end{split}
\end{align}
Here further quantities do occur that are not in the list of variables: The diffusion fluxes ${\boldsymbol J}_1,\ldots,\boldsymbol{J}_N$ and the reaction rates $r_1,\ldots,r_N$. Moreover, we have the stress which here is decomposed into pressure $p$ and its irreversible part $\boldsymbol{S}$. The quantity $\boldsymbol b$ is the external body force acting on the fluid. Finally, there are the specific internal energy $u$ and the heat flux $\boldsymbol{q}$. Except for the body force, which is assumed to be given, these objects are called constitutive quantities because they describe the material behavior of the special fluid mixture at hand.

The diffusion fluxes and the reaction rates are subjected to the constraints
\begin{align}\label{thermo3}
\sum_{i=1}^N \boldsymbol J_i =0\quad\textrm{and}\quad\sum_{i=1}^N r_i =0~,
\end{align}
such that the balance of total mass becomes a conservation law, the continuity equation
\begin{align}\label{thermo4}
\partial_t \varrho+ \divv(\varrho \, \boldsymbol v ) = 0~.
\end{align}

Note that balance laws are primarily stated for the partial mass densities $\rho_1, \ldots,\rho_N$. Due to the simple connection
$\rho_i = M_i \, n_i $, they imply balance laws for the mole densities $n_1, \ldots, n_N$ as well. However, unlike the total mass $\varrho$, the total mole density $n$ is not a conserved quantity if chemical reactions are involved.

\vspace{0.2cm} {\bf Constitutive equations.} The constitutive quantities are related to the variables by constitutive equations in a material dependent manner.

We assume the Newtonian viscosity model, where the irreversible part of the stress is linear in the symmetric velocity gradient, i.e.\
\begin{equation}\label{thermo5}
    \boldsymbol S=\lambda\, (\divv \boldsymbol v)\, \boldsymbol 1+ 2\,\eta\, (\nabla \boldsymbol v)_\textrm{sym} \quad\textrm{with}\quad
    \lambda\geq 0,\quad \lambda+\frac{2}{3}\eta\geq 0~.
\end{equation}
The constitutive equations for the diffusion fluxes and the heat flux rely on the laws of Fick, Onsager and Fourier, whereupon these quantities are linearly related to the gradients of inverse temperature, $1/T$, and of chemical potentials, $\mu_i$, more specifically
\begin{align}
\begin{split}\label{thermo6}
 \boldsymbol J_i  = & - \sum_{j=1}^N M_{i,j} \, \nabla\frac{\mu_j}{T} + l_i \, \nabla \frac{1}{T} \, , \qquad \sum_{i=1}^N M_{i,j} = 0\,  \text{ for all } j \quad\textrm{and}\quad \sum_{i=1}^N l_{i} = 0 \, ,\\
 \boldsymbol q = & \sum_{j=1}^N \, l_j \,  \nabla\frac{\mu_j}{T}+\kappa \, \nabla \frac{1}{T}~.
\end{split}
 \end{align}
In \eqref{thermo6}, $\{M_{i,j}\}$ is a symmetric, positive semi-definite matrix. For its relationship with the \emph{Onsager operator}, see for instance \cite{mueller}. The coefficients $l_j$ are related to thermo-diffusion (the Soret and the Dufour effect), while $\mu_1,\ldots,\mu_N$ are the (mass-based) chemical potentials. For details concerning the modern way of deriving appropriate representations of the fluxes, see \cite{bothedreyer}, \cite{bothedruetMS} and the references given there.

The thermodynamically consistent modelling of the reaction rates assumes them exponentially proportional to reaction affinities of the form $\sum_{i=1}^N \gamma^k_i \, M_i \, \mu_i/T$, in which $\gamma^k$ is the stoichiometric vector associated with the $k^{\rm th}$ reaction. If there are $N_{\rm R} \in \mathbb{N}_0$ independent chemical reactions, then a thermodynamically consistent choice is
\begin{align}\label{thermo7}
 r_i = - \sum_{k=1}^{N_{\rm R}} R_k \, \Big(1 - \exp(\alpha_k \, \sum_{j=1}^N \gamma^k_j \, M_j \, \mu_j/T)\Big) \, \gamma^k_i \, M_i \, ,
\end{align}
in which $R_k$ and $\alpha_k$ are some positive functions of the local thermodynamic state. Each reaction conserves the mass, i.e. we have $\sum_{i=1}^N\gamma^k_i~M_i=0$.

The crucial constitutive quantity relating the chemical potentials to the main variables is the (Helmholtz) free energy density. It is defined by the combination $\varrho \psi := \varrho u - T \, \varrho s$ with the internal energy density $\varrho u$ and the entropy density $\varrho s$. In the current context, the most general constitutive function for the free energy density has the form
\begin{align}\label{thermo8a}
\varrho \psi = \varrho \psi(T, \, \rho_1 , \, \rho_2, \ldots, \rho_N ) \, ,
\end{align}
and then the (mass--based) chemical potentials are given by
\begin{align}\label{thermo8b}
 \mu_i = \frac{\partial \varrho\psi}{\partial \rho_{i}}  \quad \text{ for } i  = 1, 2, \ldots , N \, .
 \end{align}
If the free energy function \eqref{thermo8a} were given, the second law of thermodynamics implies rules whereupon the other constitutive quantities can be calculated, viz.\
\begin{align}\label{thermo8c}
p = -\varrho\psi +
 \sum_{i=1}^N \rho_i \, \mu_i  \,  , \quad \varrho u = - T^2 \, \partial_T(  \varrho\psi/T) \ , \quad  \varrho s = - \partial_T \varrho \psi \, .
 \end{align}
A proof of this proposition may be found for instance in \cite{bothedreyer}.

The existence of a free energy function $\varrho \psi(T, \, \rho_1,\ldots,\rho_N)$ allowing to apply \eqref{thermo8b} and \eqref{thermo8c} is a necessary condition to formulate constitutive equations for the fluxes and the reaction rates of multicomponent systems. Moreover, the existence of a free energy function is the guarantor for thermodynamic consistency of these models. But thermodynamic consistency requires more than the equations \eqref{thermo8b} and \eqref{thermo8c}: Additionally, there are some inequalities. Two of them have already been stated for the Newtonian viscosities, namely \eqref{thermo5}$_{2,3}$. Further inequalities concern the chemical potentials and the specific internal energy. They read
\begin{equation}\label{thermo8d}
   \left\{\frac{\partial \mu_i}{\partial \rho_j}\right\}_{i,j=1,\ldots,N} \quad \text{is symmetric, positive definite}\quad\text{and}\quad\frac{\partial u}{\partial T}>0~.
\end{equation}

There is a special case for which the system of PDEs resulting from the balance equations \eqref{thermo2} is well studied.  This concerns the class of \textit{ideal mixtures}. Here, the chemical potentials obey
\begin{align}\label{muiideal}
 \mu_i = \hat{\mu}_i(T, \, p, \, x_i) = g_i(T, \, p) + \frac{R \, T}{M_i} \, \ln x_i \, ,
\end{align}
where the functions $g_1, \ldots, g_N$ are the specific Gibbs energies of the pure constituents, $R$ is the universal gas  constant, and $x_i = n_i/n$ are the mole fractions\footnote{The construction of $\varrho\psi$ for chemical potentials obeying the additive splitting $\hat{\mu}_i = g(T,p) + a_i(T,x_i)$ is recalled in the Appendix  \ref{idmix}. It turns out that the choice of the functions $a_i$ and $g_i$ is subject to mathematical restrictions in order to guarantee the compatibility with \eqref{thermo8b} and the thermodynamic consistency.}. Usually, the specific Gibbs energies are not explicitly given, rather they must be read off from data tables.

However, there is one exceptional special case. For \textit{mixtures of ideal gases}, the specific Gibbs energies may be explicitly calculated from statistical mechanics, resulting in
\begin{equation}\label{thermo7id}
   g_i(T, \, p)=h_i^\textrm{R} - T \, s_i^{\rm R} +\frac{R\,T}{M_i}\,\left[ \, (z_i+1)\big(1-\frac{T^0}{T}\big) -
   \ln\Big(\big(\frac{T}{T^0}\big)^{z_i+1} \, \frac{p^0}{p}\Big) \right] ~.
\end{equation}
Herein, the constants $z_i$ indicate one-, two- and more-atomic gases according to the numbers $z_i=(3/2, 5/2, 3)$, respectively.
The constants $h_i^\textrm{R}$ and $s_i^\textrm{R}$ denote the specific enthalpies and specific entropies, respectively. For many pure substances, the entropy constants are determined and tabulated. This does not apply for the enthalpy constants. For given chemical reactions, the relevant combinations $(\sum_{i=1}^N \gamma_i^k M_i h_i^\textrm{R})_{k=1,..,N^\textrm{R}}$ only may be determined and can also be found in tables. In mixtures without both chemical reactions and phase transitions the enthalpy constants are not relevant.

%
%

\vspace{0.2cm} {\bf General representation for the free energy, Part 1: Preliminaries.} The special choices \eqref{muiideal} and \eqref{thermo7id} are not well suited for the non-ideal case. Thus if we want to study the mathematical properties of systems that are non-ideal, we are faced with the problem that empirical data do not primarily provide us with the free energy density as a function of the basic variables: rather they yield expressions for pressure, specific internal energy and chemical potentials. Consequently, thermodynamical consistency, expressed by the equations \eqref{thermo8b}, \eqref{thermo8c} and \eqref{thermo8d}, is no longer guaranteed. This observation leads us to the following problem: Find a general representation of the free energy density in terms of quantities resulting from a combination of theoretical and measured data.

A preliminary problem arises from the fact that empirical data usually do not use $T, \, \rho_1, \ldots ,\rho_N$ as basic variables. The reason is that in experiments,  temperature $T$, pressure $p$ and the composition vector, i.e.\ the mole fractions $x = (x_1, \ldots,x_N)$, are controlled.

While the pressure is related to the basic variables via \eqref{thermo8c}$_1$, the number fractions obey
\begin{align}\label{thermo8}
x_i = \frac{n_i}{n} = \upsilon \, n_i \quad \text{ subject to }  \quad \sum_{i=1}^N x_i = 1 \, .
\end{align}
In order to obtain a representation formula for the free energy, we adopt the change of variables
\begin{align}\label{thermo9}
(T, \, n_1, \ldots,n_N) \longleftrightarrow (T, \, p, \, x_1, \ldots, x_N) \, .
\end{align}
Much easier than \eqref{thermo9} is a transformation $(T, \, n_1, \ldots,n_N) \longleftrightarrow (T, \, \upsilon, \, x_1, \ldots, x_N)$ because the molar volume $\upsilon$ is introduced by the simple relation $n_i=\upsilon^{-1}x_i$ for $i=1,\ldots,N$. Therefore we perform the transformation \eqref{thermo9} by two steps according to
\begin{equation}\label{thermo10}
    (T, \, n_1, \ldots, n_N)\quad\longleftrightarrow\quad(T,\upsilon \, ,x_1,\ldots,x_{N})\quad\longleftrightarrow\quad(T,p,x_1,\ldots,x_{N})~.
\end{equation}
Due to the constraint $\sum_{i=1}^N x_i = 1$, the number of independent variables does not change here.

For a local thermodynamic quantity $f$, given by a constitutive function $f(T, \, \rho_1, \ldots,\rho_N)$ of the basic variables, we denote by $\bar{f}(T, \, \upsilon, \, x_1, \ldots,x_N)$ its representation in the variables $(T, \, \upsilon, \, x)$, and by $\hat{f}(T, \, p, \, x_1, \ldots, x_N)$ its representation in the variables $(T, \, p, \, x)$.

The change of variables \eqref{thermo10}$_2$ needs an additional constitutive equation. This is the so-called \textit{thermal equation of state}, relating the pressure to the chosen variables. The most general constitutive law we can expect in the current context is given by
\begin{equation}\label{thermo10a}
    p=\bar p(T, \, \upsilon, \, x)\quad \longleftrightarrow\quad \upsilon=\hat{\upsilon}(T, \, p, \, x)~.
\end{equation}
The invertibility of \eqref{thermo10a} is guaranteed by the convexity of the free energy function. A proof of this proposition will follow below. More comments on the rules of the combined transformations \eqref{thermo10} are to be found in the Appendix \ref{pieces}.

Throughout the present paper, we assume that the molar volume $\upsilon$ can be calculated or measured in the variables $T$ (temperature), $p$ (pressure) and $x$ (composition). Hence, there is a known function $\hat{\upsilon}$ such that
\begin{align}\label{Volumeconstraint}
 \upsilon = \frac{1}{n} =  \hat{\upsilon}(T, \, p, \, x_1, \ldots, x_N) \, .
\end{align}
If this constitutive equation is given, we can now introduce all thermodynamic quantities as functions of the variables $T$, $p$ and $x$. For example, the partial mass densities obey
\begin{align}\label{hatrho}
\rho_i  = n \, M_i \, x_i = \frac{M_i \, x_i}{\hat{\upsilon}(T, \, p, \, x)} =: \hat{\rho}_i(T, \, p, \, x)\, .
\end{align}

A further important application of the transformation rules concerns the \textit{specific heat} (at constant pressure). Its definition requires the knowledge of both the internal energy function $u=\hat u(T,p,x)$ and the combination $\hat h(T,p,x)\equiv\hat u(T,p,x)+p \hat\upsilon(T,p,x)/M(x)$, which is called specific enthalpy. Then, the specific heat is defined by
\begin{equation}\label{thermo12}
   \hat c_p(T,p,x)= \partial_T \hat h(T,p,x)~.
\end{equation}

\vspace{0.2cm} {\bf General representation for the free energy, Part 2: Available data.}\label{data} Depending on the chosen data, several representation theorems for $\varrho\psi$ can be derived. Unfortunately, the structure of the empirical data--base is somehow confusing. Only the use of the variables $(T,p,x)$ can be taken for granted. In terms of these variables, the experimental literature often lists the molar volume, the specific heat and the chemical potentials via the so--called activity coefficients. These are related to the chemical potentials from \eqref{thermo8b} by
\begin{equation}\label{thermo13}
   \mu_i=\mu_i^0+\frac{R\,T}{M_i}~\ln(x_i\gamma_i)~.
\end{equation}
Sometimes $\mu_i^0$ is considered as a constant, then the complete $(T,p,x)$ dependence is contained in the activity coefficients $\gamma_i$. However, one also finds representations of \eqref{thermo13} with $\mu_i^0(T,p)$, then $\gamma_i$ exclusively describes the deviation from ideal mixtures with respect to the composition (see \cite{PLG99}).

From the theoretical point of view we must be especially careful, because here we observe a dependence between molar volume, specific heat and the chemical potentials. In fact, the pressure dependence of the molar volume already uniquely determines the pressure dependence of the specific heat and of the chemical potentials. Thus only the use of experimental data of the kind
\begin{equation}\label{thermo12a}
   \upsilon=\hat\upsilon(T,p,x),\quad c_p=\hat c_p(T,p^0 ,x)\quad\textrm{and}\quad \mu_i=\hat\mu_i(T,p^0 ,x)
\end{equation}
would lead to thermodynamically consistent constitutive laws.

Our representation theorems for the free energy rely on the following data:
\begin{enumerate}[(a)]
 \item\label{referencevalues} The reference pressure $p^0$ and the reference absolute temperature $T^0 > 0$,
 \item\label{volume} the molar volume function $\upsilon=\hat\upsilon(T,p,x)$,
 \item\label{HC} the specific heat $c_p = \hat{c}_p(T, \, p^0, \, x)$ at a single reference pressure as a function of temperature and composition,
 \item\label{EH} the purely compositional dependence of specific entropy $s = \hat{s}(T^0, \, p^0, \, x)$ and specific enthalpy $\hat{h}(T^0, \, p^0, \, x)$ at reference temperature and pressure.
\end{enumerate}
In this paper, we prefer to use the specific entropy and the specific enthalpy instead of the chemical potentials as given data. Then the representation theorems for the free energy function assume their simplest form. The functional relations between the two data sets can be read off from the formulas of Appendix \ref{pieces}.
%

\vspace{0.2cm} {\bf General representations for the free energy, Part 3: Final results.} Next, we represent the free energy with respect to the above observations. The proof of the two following propositions is found in the appendix, Section \ref{pieces}.

At first we give the free energy density with respect to the variables $T, \, \rho$, which is needed for the mathematical treatment starting in Chapter \ref{freesol}:
\begin{align}\label{FE}
&  f(T, \, \rho) := \varrho \psi(T, \, \rho) =  n(\rho) \, \int_{p^0}^{p(T, \, \rho)} \hat{\upsilon}(T, \, p^{\prime}, \, x(\rho)) \, dp^{\prime} - p(T, \, \rho) \nonumber\\
 & \quad - \varrho \, \left(\int_{T^0}^T \int_{T^0}^{\theta} \frac{\hat{c}_p(\theta^{\prime}, \, p^0, \, x(\rho))}{\theta^{\prime}} \, d\theta^{\prime}d\theta + T \, \hat{s}(T^0, \, p^0, \, x(\rho)) - \hat{h}(T^0, \, p^0, \, x(\rho)) \right)\, .
\end{align}
In \eqref{FE}, the notations $n(\rho) = \sum_i (\rho_i/M_i)$ and $x(\rho) = \rho_i/(M_i n(\rho))$ are just needed to switch between $(n_1,\ldots,n_N)$ and $(\rho_1,\ldots,\rho_N)$, while $ p(T, \, \rho)$ is the representation of the pressure in the main variables following the definition \eqref{thermo10a}. If the function $\hat{\upsilon}$ is given, the knowledge of the latter function and of its derivatives results from \eqref{Volumeconstraint}, since $\pi = p(T, \, \rho) $ is the implicit solution to
\begin{align}\label{implicitpress}
 n(\rho) \, \hat{\upsilon}(T, \, \pi, \, x(\rho)) = 1 \, .
\end{align}
The representation formula allows to characterize all derivatives of the free energy in terms of the data. For the convenience of the reader, some related identities are recalled in the Appendix \ref{pieces}.

The following alternative representation is needed to perform the incompressible limit. It gives the free energy density with respect to the variables $(T, \, p, \, x)$. Note that in the $(T,p,x)-$setting, the Gibbs energy $M(x)\hat g(T,p,x)=M(x)\hat\psi(T,p,x)+p\hat\upsilon(T,p,x)$ is the relevant potential,
\begin{align}\label{FEbis}
 M(x)\hat g(T, \, p, \, x) = & \int_{p^0}^{p} \hat{\upsilon}(T, \, p^{\prime}, \, x) \, dp^{\prime} \\
 & - M(x)\, \left(\int_{T^0}^T \int_{T^0}^{\theta} \frac{\hat{c}_p(\theta^{\prime}, \, p^0, \, x)}{\theta^{\prime}} \, d\theta^{\prime}d\theta + T \, \hat{s}(T^0, \, p^0, \, x) - \hat{h}(T^0, \, p^0, \, x) \right)\, .\nonumber
\end{align}
Finally the balance of internal energy needs a representation of the specific enthalpy. It reads
\begin{align}\begin{split}\label{FEbisH}
M(x)\hat h(T, \, p, \, x) = & \int_{p^0}^{p}(\hat{\upsilon}-T\partial_T \hat{\upsilon})(T, \, p^{\prime}, \, x) \, dp^{\prime}\\
& +M(x)\, \Big(\int_{T^0}^T\hat{c}_p(\theta, \, p^0, \, x) \, d\theta + \hat{h}(T^0, \, p^0, \, x) \Big)\, .
\end{split}
\end{align}
In using the formulae \eqref{FE}, \eqref{FEbis} and \eqref{FEbisH} one must remain aware of the fact that the free energy might not be globally smooth on the state space, but only piecewise smooth. Phase transitions can occur upon temperature, pressure and composition.

Hence, the representation \eqref{FE} might not be valid globally, but rather in a stable subregion of the state space. By state space we denote, for the basic variables, the domain
\begin{align}\label{SSTrho}
\mathscr{D}_{T,\rho} = \{(T, \, \rho) \, : \, T > 0, \quad \rho_1, \ldots, \rho_N > 0\} \, .
\end{align}
In stable subregions, thermodynamics requires the strict concavity of the entropy function which, for the free energy, is equivalent to the two following conditions:
\begin{itemize}
 \item $\rho \mapsto f(T, \, \rho)$ is strictly convex;
 \item $T \mapsto f(T, \, \rho)$ is strictly concave.
\end{itemize}

\section{Definition of incompressibility and resulting problems.}\label{defandprob}

We now consider the asymptotic behavior of  the free energy density and the resulting PDE-system for an \emph{incompressible} multicomponent fluid.

\vspace{0.2cm} {\bf Definition of incompressibility.} A multicomponent fluid is called incompressible if its molar volume $\hat{\upsilon}(T, \, p, \, x)$ exhibits only a negligible dependence on pressure; that is, the derivative $\partial_p\hat{\upsilon}$ is small compared to some empirical characteristic value. The incompressible limiting case is defined by
\begin{equation}\label{incom1}
   \partial_p\hat{\upsilon}(T,p,x)\rightarrow 0~.
\end{equation}
Essentially equivalent definitions are exposed in \cite{bothedreyer}, Section 16 and \cite{dreyerguhlkemueller}.

In the mathematical literature, the term ''incompressible'' is a synonym for the relation $\divv {\bf v} = 0$, i.e.\ ${\bf v}$ is solenoidal. Due to the continuity equation, this is equivalent to $\dot{\varrho} =0$ (material derivative).

In conjunction with this definition of incompressibility two different problems will be met. The first problem already appears for single-component fluids, while the second problem is a feature of fluid mixtures: (i) The famous \textit{Boussinesq-Approximation} embodies $\divv {\bf v} = 0$ and nevertheless allows thermal expansion, i.e.\ $\dot{\varrho} \neq0$. (ii) In a mixture, as a rule, the volume is not conserved. The standard example adds the same volume $V$ of water and ethanol under normal conditions of temperature and pressure, to obtain a mixture which occupies a volume significantly smaller than $2V$.

As we will see, this phenomenon likewise leads to $\dot{\varrho} \neq0$.
%

\vspace{0.2cm} {\bf Incompressibility forbids thermal expansion.}

The incompressible limiting case was carefully discussed by I.~M\"uller in his book \cite{mueller}. He observes at first that the limit \eqref{incom1} may generate  problems with thermodynamic consistency. The subject is controversially discussed even in the case of a single--component fluid. Here we start the discussion with the fundamental inequalities \eqref{thermo8d}. In Appendix \ref{pieces} we show that the inequalities \eqref{thermo8d} imply the further inequality
\begin{equation}\label{incom2}
(\partial_T \hat \upsilon)^2\leq - \frac{\hat c_p \, M}{T} \, \partial_p\hat \upsilon \, ,
\quad\text{and from \eqref{Athermo15b}$_1$ we have}\quad
\partial_p \hat u = -T~\partial_T\hat \upsilon-p~\partial_p \hat \upsilon~.
\end{equation}
For a single--component fluid, in which case $M$ is independent of $x$, both the inequality \eqref{incom2}$_1$ and \eqref{incom2}$_2$ form the basis of the treatment by I.~M\"{u}ller in \cite{mueller}. I.~M\"uller argues as follows. If $\partial_p\hat{\upsilon}=0$, then the inequality \eqref{incom2}$_1$ implies $\partial_T\hat{\upsilon}=0$ as well. Then one obtains from \eqref{incom2}$_2$ that $\partial_p\hat{u}=0$.

Particularly the condition $\partial_T\hat{\upsilon}=0$ in an incompressible fluid has become known as the so-called \textit{M\"uller Paradox}: It seems to inevitably forbid the famous \textit{Boussinesq-Approximation}, whereupon an incompressible fluid may exhibit thermal expansion, i.e.\ $\partial_T\hat{\upsilon}\neq0$. However, at least in liquids the \textit{Boussinesq-Approximation} is optimally grounded on experimental evidence.

Due to this apparent paradoxical result, H.~Gouin, A.~Muracchini \&  T.~Ruggeri again picked up the topic.
The paper \cite{gouin} weakens the definition of incompressibility and introduces the notion of a \textit{quasi-thermal incompressible} body. Indeed, in \cite{gouin} the authors exclusively discuss the energy equation \eqref{incom2}$_2$ but not the inequality \eqref{incom2}$_1$. Nevertheless they provide data for two interesting examples showing that the pressure--dependence of the internal energy can not be neglected, even if $\partial_p\hat{\upsilon}\rightarrow0$ is assumed.

In \cite{bothedreyer} the important role of the inequality \eqref{incom2}$_1$ is also discussed. The authors likewise consider a single--component fluid and resolve the \textit{M\"uller Paradox} on the basis of the PDE-system of motion of an incompressible fluid. In the current paper we proceed with this discussion in Section \ref{thermexp}, but for a multicomponent fluid.

\vspace{0.2cm} {\bf Incompressibility implies a linear dependence of the molar volume function on the composition.}
In this paragraph we derive a further surprising consequence of the incompressibility definition \eqref{incom1}. If we ask for convexity of $\rho \mapsto f(T, \, \rho)$ also for large deviations of the reference pressure, then the function $\hat{\upsilon}$ must be linear in the composition variable $x$. This can be seen as follows.

%

By means of \eqref{FE}, we compute the Hessian as a sum of three matrices
 \begin{align}\label{D2ffull}
  \partial^2_{\rho_i,\rho_j}f(T, \, \rho) = A_{ij} + B_{ij} + C_{ij} \, .
 \end{align}
In the variables $(T, \, p,\, x)$, the expressions for the matrices $A$ and $B$ are\footnote{Remark: For a function $\phi$ defined on the surface $\{(x_1,\ldots,x_N) \, : \, x_i \geq 0, \, \sum_{i=1}^N x_i = 1\}$, the expressions $D_x \phi \cdot [e^i - x] = \sum_{j=1}^N (\delta_{ij} - x_j) \, \partial_{x_j} \phi$ define a tangential differential operator.}
 \begin{align*}
  \hat{A}_{ij} = & \frac{\hat{\upsilon}}{M_i\,M_j} \,\int_{p^0}^{p} (D^2_{x,x}\hat{\upsilon}(T,p^{\prime},x)\, [e^i \, -x] \cdot  [e^j-x]) \, dp^{\prime} \, ,\\
  \hat{B}_{ij} = & -\frac{\hat{\upsilon}}{M_i\,M_j} \,\frac{(\hat{\upsilon} + D_x\hat{\upsilon} \cdot [e^i-x]) \, (\hat{\upsilon} + D_x\hat{\upsilon} \cdot [e^j-x])}{\partial_p\hat{\upsilon}} \, .
  \end{align*}

With the auxiliary function $F(T,x) := -\int_{T^0}^T\int_{T^0}^{\theta} (\hat{c}_p^0(\theta^{\prime},x)/\theta^{\prime}) \, d\theta \, d\theta^{\prime} - T \, \hat{s}^{00}(x) + \hat{h}^{00}(x)$, which is independent of pressure, we obtain for the matrix $C$ the expression
\begin{align*}
  \hat{C}_{ij} = & \frac{M \, \hat{\upsilon}}{M_i \, M_j} \, D^2_{x,x}F [e^i-x]\cdot[e^j-x] \\
  & + \frac{\hat{\upsilon}}{M_i} \, \Big(1-\frac{M}{M_j}\Big) \, D_xF\cdot[e^i-x] + \frac{\hat{\upsilon}}{M_j} \, \Big(1-\frac{M}{M_i}\Big) \, D_xF\cdot[e^j-x] \, .
 \end{align*}
We note that $\hat{B} = \lambda \,  \xi  \otimes \xi$ is a rank--one matrix, where
\begin{align*}
\lambda = -\frac{\hat{\upsilon}}{\partial_p\hat{\upsilon}}, \qquad \xi_i := \frac{\hat{\upsilon} + D_x\hat{\upsilon} \cdot [e^i-x]}{M_i} \, \text{ for } i = 1,\ldots,N \, .
 \end{align*}
 For $\partial_p \hat{\upsilon} \rightarrow 0$ we have $\lambda = + \infty$, and we moreover observe that
\begin{align*}
  \hat{A}_{ij} =  \frac{\hat{\upsilon}(T, \, p^0,x)}{M_i\,M_j} \, (p-p^0) \, D^2_{x,x}\hat{\upsilon}^{\infty}(T,p^0,x)  [e^i \, -x] \cdot  [e^j-x] \, .
\end{align*}
Hence the $(T, \, p, \, x)$ representation $\widehat{D^2_{\rho,\rho}f}$ of the Hessian matrix exists as an operator on the subspace $\{\xi\}^{\perp}$, and there it is given as
\begin{align}\begin{split}\label{D2fpartial}
\widehat{D^2_{\rho,\rho}f}(T,\, p,\, x)  = &
 \frac{\hat{\upsilon}(T, \, p^0, \,  x)}{M_i\,M_j} \,  (p-p^0)\,  D^2_{x,x}\hat{\upsilon}(T, \, p^0, \, x) [e^i \, -x] \cdot [e^j-x]\\ & + \frac{M(x) \, \hat{\upsilon}(T,p^0,x)}{M_i \, M_j} \, D^2_{x,x}F(T,x) [e^i-x]\cdot[e^j-x] \\
  & + \frac{\hat{\upsilon}(T,p^0,x)}{M_i} \, \Big(1-\frac{M(x)}{M_j}\Big) \, D_xF(T,x)\cdot[e^i-x] \\
  & + \frac{\hat{\upsilon}(T,p^0,x)}{M_j} \, \Big(1-\frac{M(x)}{M_i}\Big) \, D_xF(T,x)\cdot[e^j-x]  \, .
\end{split}
  \end{align}
Due to the condition that $D^2_{\rho,\rho}f$ is positive definite, the latter matrix must generate a positive operator on $\{\xi\}^{\perp}$. But in the latter representation, the only contribution varying with pressure is
\begin{align}\label{singterm}
(p-p^0)\,  D^2_{x,x}\hat{\upsilon}(T, \, p^0, \, x) [e^i \, -x] \cdot [e^j-x] \, .
\end{align}
We conclude that unless $D^2_{x,x}\hat{\upsilon} = 0$,
we always can find a contradiction to positive definiteness at finite pressures.

Hence, full exploitation of the inequalities \eqref{thermo8d} has the following two consequences for the equation of state of an incompressible fluid:
\begin{align}\label{conseqs}
\partial_T \hat{\upsilon} = 0 \quad \text{ and } \quad D^2_{x,x} \hat{\upsilon} = 0 \, .
\end{align}
In other words, as a direct mathematical consequence of the definition \eqref{incom1}, the equation of state for an incompressible fluid would assume the form:
\begin{align}\label{basicstate}
 \upsilon = \hat{\upsilon}(T^0, \, p^0, \, x) = \sum_{i=1}^N \upsilon^{00}_i \, x_i \quad \text{ with constants } \upsilon^{00}_1, \ldots, \upsilon^{00}_N \, .
\end{align}
Note that, for a mixture allowing independent vanishing of the constituents, the positivity of volume moreover implies that the constants $\upsilon^{00}$ are all positive.

%
%
%

\vspace{0.2cm}

\section{Thermal expansion and non--solenoidal effects in incompressible mixtures}\label{thermexp}

In this chapter we straighten three items. We show: (i) There is a scaling regime where thermal expansion is possible in the incompressible limit. (ii) In single-component fluids and in dilute solutions we have $\divv {\bf v} = 0$ in that scaling regime. This is not true in concentrated solutions. (iii) A linear dependence of the molar volume on the composition vector, as it is indicated in \eqref{basicstate}, is nevertheless capable to describe the observed non-linear composition dependence of the excess volume during mixing of different fluids.

\vspace{0.2cm} {\bf The compressible fluid equations.} In a first step we rewrite the balance equations \eqref{thermo2} with respect to the variables $(T,p,x)$ which are best suited to perform the incompressible limit. To this end the partial mass balances are decomposed into $N-1$ diffusion equations and the balance equation for the total mass of the mixture. Recall the definitions \eqref{thermo1}, \eqref{MM} and \eqref{thermo8}
to obtain
\begin{align}\begin{split}\label{incom3}
\frac{M_i\, x_i}{\upsilon}\Big(\ln\big(\frac{M_i\, x_i}{M(x)}\big)\Big)^{\bullet}+\divv \boldsymbol J_i &=r_i\,,\\
\dot\upsilon -\upsilon \divv \boldsymbol v  &= \frac{\upsilon}{M(x)}\dot{M}(x) \, , \\
\frac{M(x)}{\upsilon}\dot{\boldsymbol v}+\nabla p-\divv \boldsymbol S &=\frac{M(x)}{\upsilon}\boldsymbol b \, ,\\
\frac{M(x)}{\upsilon}\dot u + \divv \boldsymbol q &=-p\divv\boldsymbol v+ \boldsymbol S\, : \, \nabla \boldsymbol v ~ .
\end{split}
\end{align}
Here we used the material time derivative which is defined by $\dot\psi=\partial_t\psi+\boldsymbol{v}\cdot\nabla\psi$. Next, we introduce the general constitutive laws $\upsilon=\hat\upsilon(T,p,x)$ and $u=\hat u(T,p,x)=\hat h(T,p,x)-p~\hat\upsilon(T,p,x)/M(x)$. The derivatives of $\hat h$ with respect to $T$ and $p$ are substituted by the equations \eqref{thermo14b} and \eqref{Athermo15b}$_1$ of Appendix, Section \ref{pieces}. Then we obtain
\begin{align}\begin{split}\label{incom4}
\frac{M_i\,x_i}{\upsilon}\Big(\ln\big(\frac{M_i\, x_i}{M}\big)\Big)^{\bullet} +\divv \boldsymbol J_i&=r_i \, ,\\
\partial_T\upsilon~\dot T +\partial_p\upsilon~\dot p-\upsilon \divv\boldsymbol v &= \frac{\upsilon}{M}\dot{M}-\sum_{j=1}^N\partial_{x_j}\upsilon~\dot x_j \, , \\
\frac{M}{\upsilon}\dot{\boldsymbol v}+\nabla p-\divv\boldsymbol S&=\frac{M}{\upsilon}\boldsymbol b \, ,\\
\frac{M}{\upsilon}c_p~\dot T-\frac{T}{\upsilon}\partial_T\upsilon~\dot p +\frac{M}{\upsilon}\sum_{j=1}^N\partial_{x_j}h~\dot x_j +\divv\boldsymbol q &= \boldsymbol S\, : \, \nabla \boldsymbol v ~ .
\end{split}
\end{align}
In order to obtain a PDE-system for the variables $(T,p,x, \boldsymbol v)$, we must insert here the various constitutive laws from above: 1. Diffusion fluxes, reaction rates, the irreversible part of stress and the heat flux according to \eqref{thermo5} and \eqref{thermo6}. 2. Molar volume, specific heat and the specific enthalpy according to the paragraph Available Data.

To study the incompressible limit of these equations, the crucial object is the constitutive function for the molar volume $\upsilon$.

\vspace{0.2cm} {\bf A simple constitutive equation.} We assume that the constitutive equation $\upsilon=\hat\upsilon(T,p,x)$ has the form
\begin{equation}\label{incom5}
\frac{\sum_{i=1}^N\upsilon_i^{\rm R} \, x_i}{\hat \upsilon(T,p,x)}  = 1-\beta T^{\rm R}\big(\frac{T}{T^{\rm R}}-1\big)+\frac{p^{\rm R}}{K} \, \big(\frac{p}{p^{\rm R}}-1\big)~,
\end{equation}
describing changes of the molar volume due to thermal expansion, elastic compression and changes of the composition. We have $N+2$ constants: The quantities $\upsilon_i^{\rm R}>0$ are the molar volumes in a reference state with $(T^{\rm R}, p^{\rm R})$, $\beta$ denotes the thermal expansion coefficient and $K>0$ is the compression modulus.
%

The simple constitutive law \eqref{incom5} allows a detailed and sufficient discussion of the incompressible limit.

\vspace{0.2cm} {\bf The incompressible limit equations.} In this paragraph we discuss the incompressible limit of both the inequality \eqref{incom2}$_1$ and the PDE-system \eqref{incom4}. We follow the discussion of \cite{bothedreyer} that is extended here to liquid mixtures.

We consider a mixture with water as the solvent, which is indicated by the lower index ${\rm S}$. Then we have $\upsilon_\textrm{S}^{\rm R}=1/55.4$ L/mol for $T^{\rm R}=293$ K, $p^{\rm R}= 10^5$ Pa (for $x_\textrm{S}^\textrm{R}= 1$). In the neighborhood of this state we have a thermal expansion coefficient of about $\beta=2.07\cdot 10^{-4}$ 1/K, a bulk modulus of $K=2.18\cdot 10^9$
Pa and a specific heat of $c_p=4.18\cdot 10^3$ J/kg/K. Moreover, we need the reference values of viscosity and heat conduction, which are chosen as $\eta^\textrm{R}=10^{-3} \mbox{Pa} \, \mbox{s}$ and $\kappa^\textrm{R}=0.6$ W/K/m.

As in \cite{bothedreyer}, we introduce a small parameter $\varepsilon\ll 1$ and two parameters $\beta_0$, $\alpha_0$ of order one such that $\beta \, T^{\rm R}=\beta_0\sqrt{\varepsilon}$ and $p^{\rm R}/K=\alpha_0\varepsilon$. The choice $\varepsilon=10^{-4}$ leads to $\beta_0=6.07$ and $\alpha_0=0.46$.

Then, up to terms of order $\varepsilon^2$, we obtain from \eqref{incom5}
\begin{equation}\label{incom6}
\hat \upsilon(T,p,x) =\sum_{i=1}^N\upsilon_i^{\rm R} \, x_i \, \Big( 1+\beta_0 \,  \big(\frac{T}{T^{\rm R}}-1\big) \, \sqrt{\varepsilon} + \big[\beta_0^2 \, (\frac{T}{T^\textrm{R}}-1)^2-\alpha_0 \, \big(\frac{p}{p^{\rm R}}-1\big)\big] \, \varepsilon\Big)~.
\end{equation}

The constitutive law \eqref{incom6} is now used to exploit the inequality \eqref{incom2}$_1$. In this chapter it suffices to simplify the subsequent discussion by assuming hat the specific heat function $\hat c_p(T,p,x)$ is given by a positive constant. Then \eqref{incom2}$_1$ may be calculated in the highest order as
\begin{equation}\label{incom8}
\beta_0^2 \, \varepsilon < \frac{c_p~ T^{\rm R} M_\textrm{S}}{p^{\rm R}\upsilon^{\rm R}_\textrm{S}}\alpha_0\frac{T^{\rm R}}{T}
\frac{1+\sum_{i=1}^{N-1}(\frac{M_i}{M_\textrm{S}}-1)x_i}{1+\sum_{i=1}^{N-1}(\frac{\upsilon_i^{\rm R}}{\upsilon_\textrm{S}^{\rm R}}-1)x_i}
\varepsilon~.
\end{equation}
%

%
Note that the smallness parameter $\epsilon$ from \eqref{incom6} drops out here.

%

Next we prepare the PDE-system \eqref{incom4}. At first we rewrite the equations \eqref{incom4} in a non--dimensional form. To this end we introduce dimensionless quantities indicated by a prime: (i) time and space as $t=t_0 t'$, $x=L_0 x'$, (ii) the variables except the dimensionless mole fractions as $T=T^{\rm R} \, T'$, $p=p^{\rm R} \, p'$, $\boldsymbol v=v_0 \,  \boldsymbol v'$ with $v_0:=L_0 / t_0$, (iii) reaction rates, molar volume and specific enthalpy as $r_i=M_S/(\upsilon^R_S \, t_0) \, r_i^{\prime}$, $\upsilon=\upsilon^{\rm R}_{\rm S}\upsilon'$ and $h=p^{\rm R} \, \upsilon^\textrm{R}_{\rm S}/(M_{\rm S} \, T^\textrm{R}) \, h'$, (iv) the fluxes $\boldsymbol J_i=L_0/t_0 \, M_\textrm{S}/\upsilon^{\rm R}_{\rm S} \boldsymbol J'_i$, $\boldsymbol q=\kappa^\textrm{R}T^{\rm R}/L_0\boldsymbol q'$, (v) the irreversible part of the stress
$\boldsymbol S =\eta^\textrm{R}/t_0 \boldsymbol S^{'}$.

The phenomena that are included in the current model are diffusion, chemical reactions, convective and non-convective flow of mass and heat conduction under the force of gravity. These phenomena are weighted by dimensionless characteristic numbers. To perform the incompressible limit, not all possible characteristic numbers are needed. The relevant numbers in the context of this paper are Mach number, Reynolds number, Froude number and Fourier number, respectively. They are defined by
\begin{equation}\label{incom9}
    \textrm{Ma}^2=\frac{L_0^2\, M_\textrm{S}}{p^{\rm R}\upsilon^{\rm R}_\textrm{S}\, t_0^2},\quad
    \textrm{Re}=\frac{M_\textrm{S} \, L_0^2 }{\eta^\textrm{R}\, \upsilon_\textrm{S}^{\rm R} \, t_0},\quad
    \textrm{Fr}^2=\frac{v_0^2}{b \, L_0},\quad
    \textrm{Fo}=\frac{\kappa^\textrm{R} \, \upsilon_\textrm{S}^{\rm R}\, t_0}{M_\textrm{S} \, c_p^\textrm{R} L_0^2}.
\end{equation}

%
Written in dimensionless quantities and with the characteristic numbers from \eqref{incom9}, the PDE-system \eqref{incom4} reads
\begin{align}\begin{split}\label{incom10}
\frac{M_i\, x_i}{\upsilon} \, \Big(\ln\big(\frac{M_ix_i}{M}\big)\Big)^{\bullet}+\divv \boldsymbol J_i &=r_i \, ,\\
\partial_T\upsilon~\dot T +\partial_p\upsilon~\dot p-\upsilon \divv \boldsymbol v  &= \frac{\upsilon}{M}\dot{M}-\sum_{j=1}^N\partial_{x_j}\upsilon~\dot x_j \,  ,\\
\frac{M}{\upsilon}\dot{\boldsymbol v}+\frac{1}{\textrm{Ma}^2}~\nabla p-\frac{1}{\textrm{Re}}~\divv\boldsymbol S &=\frac{M}{\upsilon}\frac{1}{\textrm{Fr}^2}~\boldsymbol e \, ,\\
\frac{M}{\upsilon}c_p~\dot T-\frac{T}{\upsilon}\partial_T\upsilon~\dot p +\frac{M}{\upsilon}\sum_{j=1}^N\partial_{x_j}h~\dot x_j +\textrm{Fo}~\divv \boldsymbol q &=\frac{\textrm{Ma}^2}{\textrm{Re}} \boldsymbol S\, : \, \nabla \boldsymbol v ~ .
\end{split}
\end{align}

The primes, indicating the non-dimensional quantities, are dropped here. The unit vector $\boldsymbol e$ points into the direction of the gravitational force whose magnitude is $b=9.81$ m/s$^2$.

In order to study the low Mach number limit Ma$~=\sqrt{\varepsilon}$ with finite values of Re and Fo and a fixed ratio $\textrm{Fr}^2/\textrm{Ma}$, we formally expand the variables according to
\begin{align}
\begin{split}\label{incom12}
    T=T_0+T_1 \textrm{Ma}+\ldots,\quad & p=p_0+p_1 \textrm{Ma}+p_2 \textrm{Ma}^2+ \ldots,\quad x=x_0+x_1 \textrm{Ma}+ \ldots,\\
    \quad & \boldsymbol v=\boldsymbol v_0+ \boldsymbol v_1 \textrm{Ma}+\ldots,\quad
    \end{split}
\end{align}
and correspondingly the further constitutive functions in the relevant equations \eqref{incom10}$_{2,3,4}$. For example, the right-hand side of \eqref{incom10}$_{2}$ yields in the highest order
\begin{equation}\label{incom11}
 \frac{\upsilon}{M}\dot{M}-\sum_{j=1}^N\partial_{x_j}\upsilon~\dot x_j =\upsilon_\textrm{S}^\textrm{R}
 \Big(\frac{1+\sum_{j=1}^{N-1}(\frac{\upsilon_j^\textrm{R}}{\upsilon_\textrm{S}^\textrm{R}}-1)x_j}
 {1+\sum_{j=1}^{N-1}(\frac{M_j}{M_\textrm{S}}-1)x_j}\sum_{j=1}^{N-1}
 (\frac{M_j}{M_\textrm{S}}-1)\dot x_j
 -\sum_{j=1}^{N-1}
 (\frac{\upsilon_j^\textrm{R}}{\upsilon_\textrm{S}^\textrm{R}}-1)\dot x_j\Big)~.
 \end{equation}
%
%
%

\vspace{0.2cm} {\bf Discussion of the incompressible limit equations.} We start the discussion with inequality \eqref{incom8}.  For a single-component fluid we have
\begin{equation}\label{incom8a}
\beta_0^2  < \frac{c_p~ T^{\rm R} M_\textrm{S}}{p^{\rm R}\upsilon^{\rm R}_\textrm{S}}\alpha_0\frac{T^{\rm R}}{T},\quad\textrm{implying for water}\quad
37.0  < 5618.0~\frac{T^{\rm R}}{T}~.
\end{equation}
We conclude that in the highest order the inequality \eqref{incom8} is satisfied even for large deviation of the temperature from the reference temperature.

Here a short remark on the role of the inequality \eqref{incom2}$_1$ is in order. In the context of thermodynamically consistent constitutive equations, the inequality \eqref{incom2}$_1$ restricts the class of admissible constitutive functions. In the current study we proposed by \eqref{incom5} an explicit constitutive function for the molar volume. Thus in this case, the inequality restricts the temperature domain, where the constitutive function \eqref{incom5} is thermodynamically consistent.

For a multicomponent fluid we obtain from \eqref{incom8} the inequality
\begin{equation}\label{incom8b}
37.0  < 5618.0~\frac{T^{\rm R}}{T}~
\frac{1+\sum_{i=1}^{N-1}(\frac{M_i}{M_\textrm{S}}-1)x_i}{1+\sum_{i=1}^{N-1}(\frac{\upsilon_i^{\rm R}}{\upsilon_\textrm{S}^{\rm R}}-1)x_i}~.
\end{equation}
We conclude that the highest order of the inequality \eqref{incom2}$_1$ is satisfied even for large deviation of the temperature from the reference temperature and additionally for large deviation of the composition from the pure solvent.

Thus the limit $\varepsilon\rightarrow 0$ allows thermal expansion. Merely at first glance, I.~M\"{u}llers proposition, \textit{incompressibility forbids thermal expansion} is well founded \cite{mueller}. The inequality \eqref{incom8b} yields the domain for $T$ and $x$, where the \textit{M\"uller Paradox} is removed.

Next, we study the incompressible limit of the PDE-system \eqref{incom10}. Particularly we ask whether the condition $\divv \boldsymbol v=0$ is compatible with thermal expansion and/or with variations of the composition. At first we exploit the total mass balance \eqref{incom10}$_{2}$ in the highest order. Inserting here the molar volume function \eqref{incom6} and the expansions \eqref{incom12}, we obtain
\begin{equation}\label{incom10a}
-\divv \boldsymbol v_0=\frac{\sum_{j=1}^{N-1}
 (\frac{M_j}{M_\textrm{S}}-1)\dot x_{0,j}}{1+\sum_{j=1}^{N-1}(\frac{M_\textrm{j}}{M_\textrm{S}}-1)x_{0,j}}
  -\frac{\sum_{j=1}^{N-1}
 (\frac{\upsilon_j^\textrm{R}}{\upsilon_\textrm{S}^\textrm{R}}-1)\dot x_{0,j}}{1+\sum_{j=1}^{N-1}(\frac{\upsilon_j^\textrm{R}}{\upsilon_\textrm{S}^\textrm{R}}-1)x_{0,j}}~.
\end{equation}
%

Thus in the incompressible limit for mixtures of fluids we have $\divv \boldsymbol v=0$ exclusively for dilute mixtures, where the mole fractions of the solute are small, i.e.  for $(x_i)_{i=1,\ldots,N-1}=(x_{1,i})_{i=1,\ldots,N-1}\sqrt{\varepsilon}$. Note that thermal expansion is allowed in any case.

Finally we represent the momentum and the internal energy equations for dilute mixtures in the incompressible limit,
\begin{equation}\label{incom13}
 \dot {\boldsymbol v}_0+\nabla p_2-
\frac{1}{\textrm{Re}}~\divv \boldsymbol S_0=-\beta_0 \, (T_0-1),\quad
\frac{M}{\upsilon}c_p\dot T_0+\textrm{Fo}~\divv{\boldsymbol q_0}=0.
\end{equation}
Observe that the higher--order pressure $p_2$ becomes the Lagrange
multiplier that guarantees the constraint $\divv \boldsymbol v=0$ for dilute mixtures and \eqref{incom10a} for concentrated mixtures, respectively. The PDE-system \eqref{incom13} together with $\divv \boldsymbol v=0$ constitute the \textit{Boussinesq-Approximation}.

Thus we have established a thermodynamically consistent limit
describing incompressible behavior of a dilute mixture with diffusion and thermal expansion.

\vspace{0.2cm}

\section{Non-linear volume changes during mixing of water and ethanol}\label{WaEt}

It was already noted that the mixing of 1 liter water and 1 liter ethanol under normal conditions, i.e.\ $T^{\rm R}=298$K, $p^{\rm R}=10^5$Pa, leads to a fluid mixture whose volume is significantly smaller than 2 liter. In the experiment the mole fractions of the two fluids are controlled. The measured data show a strong non-linear dependence of the molar excess volume with respect to the ethanol mole fraction
(see the graphics \cite{EXP}). In this section we show that our linear constitutive law \eqref{incom5} is capable to explain the non-linear phenomenon. To this end we study equilibria in a homogeneous mixture.

\vspace{0.2cm} {\bf A simple mixing model.} According with experimental observations (cf.\ \cite{QR9662000001}) we assume that the dissolved ethanol molecules form ethanol clusters consisting of $\kappa_\textrm{E}$ ethanol units via O-H bonding. Moreover, a cluster is hydrated by $\kappa_\textrm{S}$ water molecules. Thus the chemical reaction reads $\kappa_\textrm{A} \, W + \kappa_\textrm{S} \, E 	 \ce{<=>} \textrm{C}$. However, even the simplest choice of $\kappa_{\rm A} = \kappa_{\rm S} = 1$ already embodies the principle phenomenon. 
The three constituents W, E, C are characterized by the mole numbers $N_{\rm W}, \, N_{\rm E}, \, N_{\rm C}$, and $N=N_{\rm W}+N_{\rm E}+N_{\rm C}$ is the total mole number. $N_{\rm W}, \, N_{\rm E}, \, N_{\rm C}$ are the variables. Initially $N_{\rm W}^0=(1-x) \, N^0$, $N_{\rm E}^0=x~N^0$ and $N_{\rm C}^0=0$ are given. The quantity $x \in [0,1]$ denotes the initial mole fraction of ethanol.

We assume an incompressible simple mixture and use the linear constitutive law \eqref{incom5} at fixed temperature. Then the total volume may be written according to \eqref{incom5} as
\begin{align}\label{mix1}
V = \upsilon_{\rm W}~ N_{\rm W}  + \upsilon_{\rm E}~ N_{\rm E}+\upsilon_{\rm C}~ N_{\rm C} \quad \text{ with constants } \upsilon_{\rm W}, \, \upsilon_{\rm E}, \,\upsilon_{\rm C} > 0 \, .
\end{align}
In order to have in the end a negative excess volume, we must have
\begin{align}\label{mix2}
 \Delta\upsilon := \upsilon_{\rm W} + \upsilon_{\rm E} - \upsilon_{\rm C}  > 0 \, .
\end{align}

\vspace{0.2cm} {\bf Mass balance for a homogeneous mixture and mass action law.} The three variables are determined by the homogeneous mass balances and by the mass action law, which is for a single reaction the limiting case $R\rightarrow \infty$ of \eqref{thermo7},
\begin{align}\label{mix3}
  N_{\rm W}+N_{\rm C}=N_{\rm W}^0~,\qquad N_{\rm E}+N_{\rm C}=N_{\rm E}^0~,\qquad\mu_{\rm W}+\mu_{\rm E}=\mu_{\rm C}\, .
\end{align}
The solution of the mass balances \eqref{mix3}$_{1,2}$ can be represented by
\begin{align}\label{mix4}
& N_{\rm W} = (1-x-\gamma) \, N^0, \quad N_{\rm E} = (x-\gamma) \, N^0 \, .
\end{align}
The newly introduced variable $\gamma\equiv N_{\rm C}/N^0$ will be determined by the mass action law \eqref{mix3}$_{3}$ for an incompressible simple mixture, where the (mole based) chemical potentials are given by
\begin{align}\label{mix5}
\mu_i=g_i^{\rm R}+\upsilon_i(p-p^{\rm R})+RT \, \ln(\frac{N_i}{N}) \quad \text{ with constants }\quad g_i^{\rm R} \, .
\end{align}
Then $\gamma$ follows from the algebraic equation $(x-\gamma) \, (1-x-\gamma) = K \, \gamma \, (1-\gamma)$ where $K$ is defined as
\begin{align}\label{mix6}
 K(p) := \exp\Big( - \frac{\Delta g^{\rm R} + \Delta \upsilon \, (p-p^{\rm R})}{RT} \Big)\, , \quad \Delta g^{\rm R} :=  g_{\rm W}^{\rm R} + g_{\rm E}^{\rm R} - g_{\rm C}^{\rm R} \, .
\end{align}
Since $N_{\rm C} \leq \min\{N_{\rm W}^0, \, N_{\rm E}^0\}$ implies $\gamma \leq \min\{x, \, 1-x\} \leq 1/2$, we can solve this equation to the result
\begin{align}\label{mix7}
 \gamma = \gamma(x, \, p) = \frac{1}{2} \, \Big(1- \sqrt{1-4 \, \frac{x\, (1-x)}{1+K(p)}}\Big) \, .
\end{align}
Computing $\partial_p \gamma$ we easily show that $p \mapsto \gamma$ is increasing in $p$ because $\Delta\upsilon>0$. Moreover, $x \mapsto \gamma$ has a single maximum in $[0,1]$.

\hspace{0.2cm} 

{\bf The molar excess volume.} We define the molar excess volume as
\begin{align}\label{mix8}
  \upsilon^\textrm{E}:=V^\textrm{E}/N^0:=V/N^0-\sum_{i=1}^2\upsilon_i N_i^0/N^0\quad\textrm{implying}\quad
   \upsilon^\textrm{E} (x, \, p) =-\Delta\upsilon~\gamma(x, \, p)\, .
\end{align}
The example allows four important conclusions: (i) As a consequence of the linear constitutive law \eqref{incom5} we would have
$\upsilon^\textrm{E}=0$ if the mixing were not accompanied by a chemical reaction. (ii) If there is a chemical reaction according to our simple model from above, the molar excess volume shows qualitatively already the experimental observations. (iii) Moreover, the absolute value of the  molar excess volume increases with $p$ which is expected. (iv) An incompressible mixture of fluids may exhibit changes of its volume.

\vspace{0.2cm}

\section{Free energy and chemical potentials for an incompressible fluid}

The chemical potentials defined in \eqref{thermo8b} are the quantities directly driving diffusion and reaction mechanisms. At this level too, the incompressible limit has a heavy impact on the limit PDE models.
Indeed, the form of the chemical potentials for an incompressible system are significantly affected by the two consequences $\partial_T \hat{\upsilon}=0$ and $D^2_{x,x} \hat{\upsilon}=0$ exhibited previously (see \eqref{conseqs}).

\vspace{0.2cm} {\bf Properties of the free energy density in the incompressible case.}

Due to the definition \eqref{incom1}, recall that we have $\hat{\upsilon}(T,p,x) = \hat{\upsilon}(T,p^0,x)$.
Considering the main variables, we have
\begin{align*}
 \rho_i = \hat{\rho}_i(T,p,x) = \frac{M_i \, x_i}{\hat{\upsilon}(T, p,x)} = \frac{M_i \, x_i}{\hat{\upsilon}(T, p^0,x)} \, .
\end{align*}
Thus we see that the state $(T,\rho)$ in an incompressible system is subject to the constraint
\begin{align}\label{INCCONS}
n(\rho) \, \hat{\upsilon}(T, \, p^0, \, x(\rho)) = 1 \, .
\end{align}

A fundamental question is next: In which subset of the state space $\mathscr{D}_{T,\rho}$ can we reasonably assume the pressure--invariance of the volume? In general this will be the case only for a certain neighborhood $\mathscr{B}(T^0, \, \rho^0)$ of the reference state -- the one occurring in the formulae \eqref{FE}, \eqref{FEbis}.

In this paper, we will moreover postulate that \emph{incompressibility occurs only inside of regions of stability of the free energy}. Hence, $\partial_p \hat{\upsilon} \rightarrow 0$ in a certain region for $(T,p,x)$ presupposes that
\begin{align*}
 D^2_{\rho,\rho}f \geq 0, \qquad \partial^2_{T,T}f \leq 0
\end{align*}
in the corresponding region for $(T,\rho)$. In other words, where ever the convexity/concavity conditions on the free energy are violated, it is not meaningful to perform the limit $\partial_p\hat{\upsilon} \rightarrow 0$.

Due to this assumption, either the data in the material laws are so particular as to guarantee global stability, or the \emph{range of temperatures, pressures and compositions} for which we can perform the incompressible limit must be subject to restrictions. We will discuss both cases in our analysis.

\vspace{0.2cm}
{\bf Globally stable incompressible phase.}

The first case is assuming that, inside of a certain range of temperatures $]T^0-\delta, \, T^0+\delta[$, incompressibility is valid for arbitrary compositions and pressures. As discussed in section \ref{defandprob}, this type of limit can occur only if the constitutive function $\hat{\upsilon}$ describing the molar volume is linear in $x$ (see \eqref{conseqs}).

Since the inequality \eqref{incom2} requires that the volume of an incompressible system is independent on temperature, the asymptotic volume $\hat{\upsilon}$ depends only on composition, and this linearly. Hence, this type of incompressible system is restricted by the constraint (cp.\ \eqref{basicstate})
\begin{align}\label{VOLVOL}
\sum_{i=1}^N \frac{\upsilon^{00}_i}{M_i}  \, \rho_i = 1 \qquad \text{ with } \upsilon^{00}_1, \ldots, \upsilon^{00}_N  \, \text{ (positive) constants.}
\end{align}
Let us first describe in more details the mathematical procedure.
We consider a sequence of free energies $\{f^m\}$ obeying the representation \eqref{FE} with data $\hat{\upsilon}^m$, $\hat{c}_p^m(p^0, \, \cdot)$, $\hat{s}^m(T^0, \, p^0, \cdot)$, $\hat{h}^m(T^0, \, p^0, \cdot)$ indexed by the large parameter. We assume that these data converge --say uniformly on compact subsets-- while the isothermal compressibility $\hat{\beta}_T^m(T, \, p, \, x) = -\partial_p\hat{\upsilon}_m(T,p,x)/\hat{\upsilon}_m(T,p,x)$ tends to zero for all $(T, \, p, \, x)$. Hence, the asymptotic volume $\hat{\upsilon}^{\infty}$ is independent on $p$.

Suppose now that $\rho \mapsto f^m(T, \, \rho)$ is convex, and $T\mapsto f^m(T, \, \rho)$ is concave. We show below that the limit of $f^m$ in the sense of epi-convergence and the representation of the free energy in the main variables is given by the singular convex function
\begin{align}\label{hinfty}
 \varrho\psi = f^{\infty}(T, \, \rho_1, \ldots, \rho_N) := \begin{cases}
                                                      \varrho\psi^{\infty}(T,\, \rho_1, \ldots, \rho_N) & \text{ for }  \sum_{i=1}^N \frac{\upsilon^{00}_i}{M_i}  \, \rho_i = 1 \, ,\\
                                                      + \infty & \text{ otherwise.}
                                                     \end{cases}
\end{align}
Here $\psi^{\infty}(T, \, \rho) = \lim_{m\rightarrow \infty} \hat{\psi}^m(T, \, p^0,\, \rho)$ is the limit of reference free energies corresponding to the isobaric system at $p^0$. Hence, if the data $\hat{c}_p^m(p^0, \cdot)$, $\hat{s}^m(T^0, \, p^0, \cdot)$, $\hat{h}^m(T^0, \, p^0, \cdot)$ at reference pressure/temperature are available, and they all converge in the classical sense to corresponding limits, we will have
\begin{align}\label{psiinfty}
 \psi^{\infty}(T, \, \rho) := & -\frac{p^0}{M(x(\rho))} \, \hat{\upsilon}^{\infty}(T^0,p^0,x(\rho)) \, \\
 & - \left(\int_{T^0}^T \int_{T^0}^{\theta} \frac{\hat{c}_p(\theta^{\prime}, \, p^0, \, x(\rho))}{\theta^{\prime}} \, d\theta^{\prime}d\theta + T \, \hat{s}(T^0, \, p^0, \, x(\rho)) - \hat{h}(T^0, \, p^0, \, x(\rho)) \right) \, . \nonumber
\end{align}
For an incompressible system, the relations \eqref{thermo8b}, $\eqref{thermo8c}_1$ find natural generalizations in $$\mu \in \partial_{\rho} f^{\infty}(T, \, \rho) \, ,$$ with the subdifferential operator $\partial $ of convex analysis or, equivalently\footnote{The function $\psi^{\infty}$ of \eqref{psiinfty} is defined for all $\rho_1, \ldots, \rho_N >0$ and differentiable in $\rho$.} ,
\begin{align}\label{MU1}
\mu_i = p \, \frac{\upsilon^{00}_i}{M_i} + \partial_{\rho_i} (\varrho\psi^{\infty})(T, \, \rho) \, .
\end{align}
As shown in \eqref{D2fpartial}, if the pressure is allowed to exhibit large deviations above and below the reference value $p^0$, the Hessians $D^2_{\rho,\rho} f^m$ cannot remain positive semi--definite on all states unless $D^2_{x,x}\hat{\upsilon}^{\infty} = 0$.  Hence, the linearity of $\hat{\upsilon}^{\infty} $ in the composition follows, in the last instance, from the assumption that the incompressible phase remains globally stable.

This problem can only be avoided by a treatment of the incompressible limit restricted to local stable subregions of the state space.
For instance, it is sensible that a nontrivial matrix $$\{D^2_{x,x}\hat{\upsilon}^{\infty}[e^i \, -x] \cdot [e^j-x]\}$$ in \eqref{D2fpartial} with moderate eigenvalues is compatible with moderate oscillations of $p$ around the reference value $p^0$ and with the requirement of convexity of $\rho \mapsto f^m(T, \, \rho)$ for the free energy.
For this reason, it is desirable to also investigate, as a more realistic concept, the more complex case of incompressibility being assumed only locally in the state space, for a certain neighborhood of the reference state.

\vspace{0.2cm}
{\bf Locally stable incompressible phase.}

We do not assume that the incompressible phase is global, but it is located around a reference state $(T^0, \, \rho^0)$. The smoothness and the convexity/concavity assumptions on the free energy are valid beyond this neighborhood.

Let us explain what this means. With $\widehat{D^2_{\rho,\rho}f}(T, \, p, \, x)$ denoting the representation in the variables $(T, \, p, \, x)$ for the Hessian of $\rho \mapsto f(T, \, \rho)$, stability means that
\begin{align}\label{bolle}
&  \widehat{D^2_{\rho,\rho}f}(T, \, p, \, x) \text{ positive definite,} \qquad \widehat{\partial^2_{T,T}f}(T, \, p, \, x) < 0 \, .
\end{align}
We assume that there is a fixed region $\widehat{\mathscr{B}}$ of the reference state $(T^0, \, p^0, \, x^0)$ such that the two conditions \eqref{bolle} are valid and that there are two thresholds $\pi_{\inf} < p^0 < \pi_{\sup}$ such that
$$\widehat{\mathscr{B}} \supset \{(T, \, p, \, x) \, : \, T \in ]T^0-\delta,T^0+\delta[, \, p \in ]\pi_{\inf}(T,x), \, \pi_{\sup}(T,x)[ \} \, .$$
Consider a sequence of free energies $\{f^m\}$ obeying the representation \eqref{FE}.
We now study the asymptotic limit $\partial_p\hat{\upsilon}^m(\cdot,p,\cdot) \rightarrow 0$ only for $p \in ]\pi_{\inf},\, \pi_{\sup}[$, while assuming that $\partial_p \hat{\upsilon}^{\infty} < 0$ for $p < \pi_{\inf}$ and $p > \pi_{\sup}$.

As noted before, we must also have $\partial_T \hat{\upsilon}^m \rightarrow 0$ for all $p \in ]\pi_{\inf}, \, \pi_{\sup}[$ and $T \in ]T^0-\delta, \, T^0+\delta[$. Hence the incompressible phase is characterized by the constraint
\begin{align}\label{ic}
n(\rho) \, \hat{\upsilon}^{\infty}(T^0, \, p^0, \,x(\rho))  = 1 \, .
\end{align}
For states not satisfying this constraint, the free energy will depend on pressure again. Consider $\rho$ such that
$n(\rho) \, \hat{\upsilon}^{\infty}(T^0, \, p^0, \,x(\rho)) > 1$. For all $T \in ]T^0-\delta, \, T^0+\delta[$, we must have $$n(\rho) \, \hat{\upsilon}^{\infty}(T, \, p^0, \,x(\rho)) = n(\rho) \, \hat{\upsilon}^{\infty}(T^0, \, p^0, \,x(\rho)) > 1\, .$$ Since for $p > \pi_{\sup}$ the asymptotic volume is decreasing, we might find a solution $\pi = p(T, \, \rho) > \pi_{\sup}$ to
$$n(\rho) \, \hat{\upsilon}^{\infty}(T, \, \pi, \,x(\rho)) = 1  \, .$$
A similar consideration applies to states such that $n(\rho) \, \hat{\upsilon}^{\infty}(T^0, \, p^0, \,x(\rho)) < 1$.
For this type of system, the limit free energy in the sense of epi--convergence (and even of uniform convergence) can thus be identified as
\begin{align*}
& f^{\infty}(T, \, \rho) = \\
& \begin{cases}
n(\rho) \, \int_{p^0}^{p(T, \, \rho)} \hat{\upsilon}(T, \, p^{\prime},  \, x(\rho)) \, dp^{\prime} - p(T, \, \rho) + \varrho \psi^{\infty}(T, \, \rho) & \text{ for } n(\rho) \, \hat{\upsilon}(T^0, \, p^0, \,x(\rho)) > 1\\
\varrho\psi^{\infty}(T, \, \rho) & \text{ for } n(\rho) \, \hat{\upsilon}(T^0, \, p^0, \,x(\rho)) = 1\\
n(\rho)\,    \int_{p^0}^{p(T, \, \rho)} \hat{\upsilon}(T, \, p^{\prime},  \, x(\rho)) \, dp^{\prime} - p(T, \, \rho) + \varrho \psi^{\infty}(T, \, \rho) & \text{ for } n(\rho) \, \hat{\upsilon}(T^0, \, p^0, \,x(\rho))  < 1
                         \end{cases}
\end{align*}
As we already stressed, this description is valid in some neighborhood of a stable incompressible phase characterized by \eqref{ic}, but not globally in the state space. It is to note that $f^{\infty}$ is continuous, and even continuously differentiable except in points $\rho$ subject to \eqref{ic}. There, it is still sub--differentiable, and $\mu \in \partial_{\rho}f^{\infty}(T,\rho)$ if and only if there exists $p \in [\pi_{\inf}, \, \pi_{\sup}]$ such that, for $i = 1,\ldots,N$,
\begin{align}\label{MU2}
\mu_i = p \, \frac{1}{M_i} \, \Big(\hat{\upsilon}^{\infty}(T^0,p^0,x(\rho)) + D_x\hat{\upsilon}^{\infty}(T^0,p^0,x(\rho))\cdot[e^i-x(\rho)]\Big)  + \partial_{\rho_i} (\varrho\psi^{\infty})(T,\rho) \, .
\end{align}
Comparing \eqref{MU2} with \eqref{MU1}, the question whether $D^2_{x,x} \hat{\upsilon}^{\infty} = 0$ is satisfied or not has an essential impact on the structure of diffusion and chemical reactions in the limit PDEs describing an incompressible fluid.

\subsubsection*{Main findings and structure of the paper}

Our analysis yields several important conclusions:
\begin{itemize}
\item Incompressible mixtures defined by $\partial_p \hat{\upsilon} = 0$ allow for non--solenoidal effects and thermal expansion in certain regimes;
 \item Performing the low Mach--number limit while assuming globally the convexity of the free energy leads, for multicomponent systems, to the conclusion that incompressibility implies $D^2_{x,x} \hat{\upsilon} = 0$. In this case, the asymptotic free energy is a singular function of the main variables, and approaching free energy functions epi-- or Gamma--converge;
 \item A more careful treatment is based on the assumption that the incompressible phase is (strictly) included in a region of stability of the phase diagram, characterized by finite pressure thresholds. The thermal equation of state is singular only inside of these thresholds. Then the asymptotic free energy is a continuous function, and approaching free energy functions converge uniformly.
 \end{itemize}

The remainder of the paper is devoted to rigorously proving the validity of these limits.

We shall, however, start with some illustrating examples in the next section.

In the section \ref{notationen} we then introduce somewhat more convenient notations for the proofs, the latter occupying sections \ref{freesol}--\ref{incompressible}.

\section{Two examples}\label{twoex}

\subsection{The volume additive case}

The thermal equation of state proposed in the papers \cite{dreyerguhlkemueller}, \cite{dreyerguhlkelandstorfer}, \cite{dreyerguhlkemueller19} is $$p = p^0 + K \, (\sum_{i=1}^N \upsilon^{00}_i \, n_i - 1)\, ,$$ describing ideal electrolyte mixtures with finite volume effects of the ions. Here $K$ is a positive constant (the compression module), $p^0 \in \mathbb{R}$ is the reference pressure, and $\upsilon^{00} \in \mathbb{R}^N$ is a fixed vector. Since the model was used in an isothermal context, the temperature--dependence of the data is not precised further.

This model, which is slightly different from \eqref{incom5}, is equivalently rephrased as \eqref{Volumeconstraint} with
\begin{align*}
\hat{\upsilon}(T, \, p, \, x) = \frac{K}{p-p^0+K} \, \sum_{i=1}^N \upsilon_i^{00} \,  x_i \, .
\end{align*}
The constitutive function $\hat{\upsilon}$ is defined for all $T > 0$, $p > p_{\inf} := p^0 - K$ and $x_1, \ldots,x_N > 0$. Moreover, it is linear in the composition variable.

Note that for $K \rightarrow + \infty$ (incompressible limit), we have the properties
\begin{align*}
p_{\inf}(K)  \rightarrow & - \infty\, , \\
\hat{\upsilon}(T, \, p, \, x) \rightarrow & \sum_{i=1}^N\upsilon^{00}_i \, x_i\, ,  \quad \partial_p \hat{\upsilon}(T, \, p, \, x) \rightarrow 0 \text{ pointwise.}
\end{align*}
We can put this example into the framework of ideal mixtures.
If we propose the globally convex free energy
\begin{align*}
 f(T, \, \rho) = & K \, \sum_{i=1}^N \frac{\upsilon^{00}_i}{M_i} \, \rho_i \, \ln \Big(\sum_{i=1}^N \frac{\upsilon^{00}_i}{M_i} \, \rho_i\Big) + (p^0-K) \, \Big(\sum_{i=1}^N \frac{\upsilon^{00}_i}{M_i} \, \rho_i-1\Big) \\
 & + R \, T \, n(\rho) \, x(\rho) \cdot \ln x(\rho) \, ,
\end{align*}
we reach \eqref{muiideal} with the choice of $g_i(T, \, p) := (p^0-p_{\inf}) \,  \frac{\upsilon^{00}_i}{M_i} \, \ln (p-p_{\inf})$.

In the limit of $K \rightarrow \infty$, the functions $\rho \mapsto f^K(T, \, \rho)$ epi--converge at fixed $T$ to the limit
\begin{align*}
 f^{\infty}(T, \, \rho) = \begin{cases}
                            R \, T \, n(\rho) \, x(\rho) \cdot \ln x(\rho)  & \text{ if } \sum_{i=1}^N \frac{\upsilon^{00}_i}{M_i} \, \rho_i = 1 \, ,\\
                           + \infty & \text{ otherwise.}
                          \end{cases}
\end{align*}

\subsection{The non--ideal case}

We revisit the example proposed in \cite{bothedreyer}, Sec.\ 16. In this case, the thermal equation of state is
\begin{align*}
 p = p^0 + K \, \Big(\frac{n}{n^0} - a(T, \, x)\Big) \, ,
\end{align*}
in which $a$ is a given function, $K$ the constant compression module and $p^0$, $n^0$ reference values. This is equivalent with choosing $\hat{\upsilon}$ nonlinear in $x$ via
\begin{align}\label{section16}
 \hat{\upsilon}(T, \, p, \, x) = \frac{1}{n^0 \, (a(T, \, x) + \frac{p-p^0}{K}  )} \, .
\end{align}
The representation formula \eqref{FE} yields the free energy $f(T, \, \rho) =  f_{\rm mech}(T, \, \rho) + f_{\rm therm}(T, \, \rho)$, in which
\begin{align*}
 f_{\rm mech}(T, \, \rho) := & n \, \int_{p^0}^{p} \hat{\upsilon}(T, \, p^{\prime}, \, x) \, dp^{\prime} - p\\
 = & K \, \frac{n}{n^0} \, \ln \frac{n}{n^0} - K \, \frac{n}{n^0} \, (1 + \ln a(T, \, x)) - p^0 + K \, a(T, \, x) \,
 \end{align*}
 with $n = n(\rho)$, $p=p(T, \, \rho)$ and $x = x(\rho)$. The terms in $f_{\rm therm}$ are given as
 \begin{align*}
 f_{\rm therm}(T, \, \rho) = - \varrho \, \left(\int_{T^0}^T \int_{T^0}^{\theta} \frac{\hat{c}_p^0(\theta^{\prime}, \, x)}{\theta^{\prime}} \, d\theta^{\prime}d\theta + T \, \hat{s}^{00}(x) - \hat{h}^{00}(x) \right) \,
 \end{align*}
with data $\hat{c}_p^0 = \hat{c}_p(p^0, \, \cdot)$, $\hat{s}^{00} = \hat{s}(T^0, \, p^0, \cdot)$ and $\hat{h}^{00}= \hat{h}(T^0, \, p^0, \cdot)$ that we do not specify further.

By elementary means we compute the Hessian $\{D^2_{\rho_i,\rho_j} f_{\rm mech}\}$ (see \eqref{D2ffull}). We next want to discuss the convexity/concavity of this free energy.
For $\eta \in \mathbb{R}^n$ arbitrary, we define
\begin{align*}
 \xi_i = \eta_i - \frac{x_i}{\hat{\upsilon}} \,  \eta \cdot \{1^N \, (\hat{\upsilon} - D_x\hat{\upsilon} \cdot x) + D_x\hat{\upsilon}\} \, \text{ for } i = 1,\ldots,N \, .
\end{align*}
One can verify that
\begin{align*}
\sum_{i,j} M_i \, M_j \, \xi_i \, \xi_j \, \frac{\partial^2 f }{\partial \rho_i \partial \rho_j} = \hat{\upsilon} \, \sum_{i,j} \eta_i \, \eta_j \, \int_{p^0}^{p} (D^2_{x,x}\hat{\upsilon} [e^i \, -x] \cdot & [e^j-x]) \, dp^{\prime} \\
+ \sum_{i,j} M_i \, M_j \, \eta_i \, \eta_j \, \frac{\partial^2 f_{\rm therm}}{\partial \rho_i \, \partial \rho_j} \, .
\end{align*}
Assume that $f$ is convex in $\rho$, which is necessary for a stable phase. For $K \rightarrow + \infty$, the latter implies that
\begin{align*}
0 \leq & \sum_{i,j} M_i \, M_j \, \eta_i \, \eta_j \, \frac{\partial^2 f_{\rm therm}}{\partial \rho_i \, \partial \rho_j} +  \frac{p-p^0}{a^3(T, \, x) \, (n^0)^2}  \times\\
& \times \sum_{i,j}  \eta_i \, \eta_j \Big[\frac{2}{a} \, a_x(T, \, x) [e^i-x] \, a_x(T,x)[e^j - x] - a_{x,x}(T,x) [e^i-x]  [e^j-x]\Big]\, .
\end{align*}
Since $\eta \in \mathbb{R}^N$ was arbitrary, we obtain from the requirement of stability that the matrix
\begin{align*}
 & M_i \, M_j \, \frac{\partial^2 f_{\rm therm}}{\partial \rho_i \, \partial \rho_j} \\
 & + \frac{p-p^0}{a^3(T, \, x) \, (n^0)^2} \, \Big[\frac{2}{a} \, a_x(T, \, x) [e^i-x] \, a_x(T,x)[e^j - x] - a_{x,x}(T,x) [e^i-x]  [e^j-x]\Big] \,
\end{align*}
must be positive definite. Clearly, if $a$ is an arbitrary function, this condition cannot be valid for all temperatures, compositions and pressures.

For the stability of the incompressible phase, it must further hold that
\begin{align*}
 0 \geq \partial_{T,T}^2 f = n(\rho) \, \int_{p^0}^{p(T, \, \rho)} \partial_{T,T}^2\hat{\upsilon}(T,p^{\prime},x(\rho)) \, dp^{\prime} - n(\rho) \, \frac{(\partial_T\hat{\upsilon})^2}{\partial_p \hat{\upsilon}} - \varrho \, \hat{c}^0_p(T, \, x(\rho)) \, ,
\end{align*}
which is equivalent to
\begin{align*}
 \frac{n}{n^0} \, (\partial_Ta(T,x))^2 \leq & a^2(T,x) \, \Big(\frac{\varrho \, \hat{c}_p^0(T,x)}{K} - \partial^2_{T,T}a(T,x) \, (1- \frac{n}{n^0\, a})\Big) \\
 = & \frac{a^2(T,x)}{K} \, \Big(\varrho \, \hat{c}_p^0(T,x)- \frac{\partial^2_{T,T}a(T,x)}{a(T,x)} \, (p-p^0)\Big) \, .
\end{align*}
Hence, letting $K \rightarrow +\infty$ at fixed $(T,p,x)$, we see that $\partial_Ta = 0$ is a necessary requirement.

Overall, for a neighborhood of stability around the reference point $(T^0, \, p^0, \, x^0)$ we get the conditions
\begin{align*}
& a(T, \, x) = a(T^0, \, x) =: a^0(x)\, ,\\
& M_i \, M_j \, \widehat{\frac{\partial^2 f_{\rm therm}}{\partial \rho_i \, \partial \rho_j}}(T,p,x) + \frac{p-p^0}{a^0(x)^3 \, (n^0)^2}  \times \\
 & \qquad \times  \Big[\frac{2}{a^0(x)} \, a_{x}^0(x) [e^i-x] \, a_{x}^0(x)[e^j - x] - a_{x,x}^0(x) [e^i-x]  [e^j-x]\Big] \quad - \text{ positive definite.}
\end{align*}
We now assume that the Hessians of the thermal part are positive semi--definite at $(T^0, \, p^0, \, x_0)$.

Then there is a neighborhood $B_{T^0,\, x^0} $ of $(T^0, \, x^0)$, such that for every $(T, \, x) \in B_{T^0, \, x^0}$, we can find thresholds $\pi_{\inf}(T,x)$ and $\pi_{\sup}(T,x)$ for which all states $(T, \, p, \, x)$ subject to
\begin{align*}
(T, \, x) \in B_{T^0,\, x^0} \quad \text{ and } \quad \pi_{\inf}(T,\, x) \leq p \leq \pi_{\sup}(T,\, x) \,
\end{align*}
belong to a stable incompressible phase. Inside of this domain, the limit free energy is of the form
\begin{align}\label{result}
f^{\infty}(T, \, \rho) = f_{\rm therm}(T, \, \rho) - p^0 \, .
\end{align}
In order to obtain an epi--convergence result for this case, we can start from a slight modification of the constitutive equation \eqref{section16}. Instead of a constant $K$, we choose $1/K = \hat{\beta}_m(T,p,x)$ such
\begin{align*}
 \lim_{m\rightarrow \infty} \hat{\beta}_m(T,p,x) \,\, = \begin{cases} 0 & \text{ if } \pi_{\inf}(T,x) \leq p \leq \pi_{\sup}(T,x) \, ,\\
       \hat{\beta}^{\infty}(T,p,x)                                     > 0 & \text{ otherwise.}
                                           \end{cases}
\end{align*}
Assuming that the free energies $f_m$ can be constructed globally convex, we find a non-singular epi--limit $f^{\infty}$ matching the desired function \eqref{result} in the stable incompressible region.

\section{Additive decomposition of the free energy}\label{notationen}

In order to study the incompressible limit and the convexity properties, we introduce a decomposition of the solution formula \eqref{FE}, that we judge useful in order to extract the contributions that are essentially affected by the mechanical properties of the system.

\subsubsection*{The volume and the Gibbs potential} Making use of the identities $n_i = \rho_i/M_i$ and $x_i(\rho) = \rho_i/(M_i \, \sum_{j=1}^N (\rho_j/M_j))$, $n(\rho) = \sum_{j=1}^N n_j$, we define a further function
\begin{align}\label{VolumeTot}
 V(T, \, \pi, \, \rho) := n(\rho) \, \hat{\upsilon}(T, \, \pi, \, x(\rho))\,
\end{align}
where $\pi$ is a free variable. The function $V$ is dimensionless. It is readily checked that the map $\rho \mapsto V(T, \, \pi, \, \rho)$ is positively homogeneous for all $(T, \, \pi)$. Moreover, the definition \eqref{Volumeconstraint} is simply equivalent to
\begin{align}\label{arbeitvolume}
V(T, \, p, \, \rho) = 1 \, .
\end{align}
We may compute that
\begin{align*}
 \partial_{\rho_i} V(T, \, \pi, \, \rho) =& \frac{1}{M_i} \,  \hat{\upsilon}(T, \, \pi, \, x(\rho)) + n(\rho) \, \sum_{k=1}^N\partial_{\rho_i} x_k(\rho) \, \partial_{x_k} \hat{\upsilon}(T, \, \pi, \, x(\rho)) \, .
 \end{align*}
Recall that the composition vector $x$ is subject to the constraint $\sum_{i=1}^N x_i = 1$ and that the function $\hat{\upsilon}$ is measured only for such physical states. Hence, the functions $\hat{\upsilon}$ possess only tangential derivatives in $x$.
For a function $\phi$ depending on $x$, natural tangential derivatives are given by
\begin{align}\label{difftau}
\partial^{\tau}_{x_i} \phi(x)  := \phi_{x_i}(x) - \sum_{j=1}^N x_j \, \phi_{x_j}(x) = D_x \phi(x) \cdot [e^i-x] \, .
\end{align}
This makes sense with arbitrarily chosen differentiable extension outside of the surface $\sum_{i=1}^N x_i = 1$, and is independent of the extension. Then
 \begin{align}\label{Vinter}
  \partial_{\rho_i} V(T, \, \pi, \, \rho)  = & \frac{1}{M_i} \, (\hat{\upsilon}(T, \, \pi, \, x(\rho)) + \partial^{\tau}_{x_i} \hat{\upsilon}(T, \, \pi, \, x(\rho)) \, .
\end{align}
As an illustration, for an ideal mixture characterized by the potentials $\mu_i = g_i(T, \, p) + R \, T/M_i \, \ln x_i$ (cf.\ \eqref{muiideal}), we have
\begin{align}\label{Videal}
 V(T, \, \pi, \, \rho) := \sum_{i=1}^N \partial_{\pi}g_i(T, \, \pi) \, \rho_i, \quad   \partial_{\rho_i} V(T, \, \pi, \, \rho) =\partial_{\pi} g_i(T,\, \pi) \, .
\end{align}

\subsubsection*{Interpretation of the function $V$}

With $\psi$ being the specific free energy and $\hat{\psi}$ its representation in the variables $(T,p,x)$, the Gibbs density $g(T,p,x)$ is introduced via \eqref{FEbis}. The connection with the free energy results from  \eqref{hatrho} and \eqref{FEbis}, and is directly given as
\begin{align}\label{superdirect}
f(T, \, \rho_1,\ldots,\rho_N) = \varrho \psi(T, \, \rho_1,\ldots,\rho_N) = \varrho \, g(T, \, p(T, \, \rho), \, x(\rho)) - p(T, \, \rho) \, .
\end{align}
Together with $\eqref{thermo8c}_1$, this also shows that
\begin{align}\label{superdirect2}
g(T, \, p(T, \, \rho), \, x(\rho)) = \frac{1}{\varrho} \, \sum_{i=1}^N \rho_i \, \mu_i\, ,
\end{align}
or equivalently, with $\hat{\mu}_i(T, \, p, \, x) = \partial_{\rho_i}f\big(T, \, \hat{\rho}(T, \, p, \, x)\big)$,
\begin{align}\label{superdirect3}
g(T, \, p, \, x) = \frac{1}{M(x)} \, \sum_{i=1}^N x_i \, M_i \, \hat{\mu}_i(T, \, p, \, x)\, .
\end{align}
We introduce $M_i \, \hat{\mu}_i =: \hat{\mu}_i^{\text{mol}}$, the molar--based chemical potential. With the help of \eqref{FEbis}, we find that
\begin{equation}\label{thermo13b}
\partial_p(M\, g) = \hat{\upsilon}, \qquad \partial_{x_i}^{\tau} (M\, g) = \hat{\mu}_i^{{\rm mol}}- x \cdot \hat{\mu}^{{\rm mol}} \, .
\end{equation}
Using \eqref{superdirect2}, we see that $x \cdot \hat{\mu}^{{\rm mol}} = M \, g$ and might also rephrase \eqref{thermo13b}$_2$ as
$$ \partial_{x_i}^{\tau} (M\, g) + M \, g = \hat{\mu}_i^{{\rm mol}} \, .$$
Hence, building the $p$ derivative yields $\partial_p \hat \mu_i^{\text{mol}} =\partial_{x_i}^{\tau} \hat{\upsilon} + \hat{\upsilon}$.
We thus see with the help of \eqref{Vinter} that
\begin{align}\label{shoulde}
  \partial_{\rho_i} V(T, \, \pi, \, \rho) = \frac{1}{M_i} \, \partial_p\hat{\mu}_i^{\text{mol}}(T, \, \pi, \, x(\rho)) = \partial_p\hat{\mu}_i(T, \, \pi, \, x(\rho))\, .
\end{align}
We might call $V$ the volume potential, because for $\pi = p$, \eqref{shoulde} shows that the derivatives $\partial_{\rho_i}V(T, \, p, \, \rho)$ are the partial volumes per unit mass of the species $\ce{A}_i$.

In all subsequent considerations, the primitive of $V$ in the pressure variable, denoted by $\bar{V}$, plays an important role. Normalized at the reference pressure $p^0$, this function has the expression
\begin{align*}
\bar{V}(T, \, p, \, \rho) := \int_{p^0}^p V(T, \, p^{\prime}, \, \rho) \, dp^{\prime} \, .
\end{align*}
Clearly, $\eqref{thermo13b}_1$ implies that
\begin{align}\label{integragips1}
g(T, \, p, \, x) = \frac{1}{M(x)} \, \int_{p^0}^p \hat{\upsilon}(T, \, p^{\prime}, \, x) \, dp^{\prime} + g(T, \, p^0, \, x) \, .
\end{align}
Now, \eqref{integragips1} shows that $\bar{V}(T, \, p, \, \rho) = \varrho \, (g(T, \, p,  \, x(\rho)) - g(T, \, p^0,  \, x(\rho)))$. Invoking \eqref{shoulde}, we in particular have
\begin{align*}
\partial_p\bar{V}(T, \, p, \, \rho) = V(T, \, p, \, \rho), \quad   \partial_{\rho_i}\bar{V}(T, \, p, \, \rho) = (\hat{\mu}_i(T, \, p, \, x(\rho)) - \hat{\mu}_i(T, \, p^0, \, x(\rho))) \, .
\end{align*}

\subsubsection*{The mechanically neutral part of the free energy.}

Using the function $g$, we introduce
\begin{align}\label{kgrund}
k(T, \, \rho) := & \varrho \, \hat{\psi}(T, \, p^0, \, x(\rho)) =  \varrho \, g(T, \, p^0, \, x(\rho)) - p^0 \, n(\rho) \, \hat{\upsilon}(T, \, p^0, \, x(\rho)) \, .
\end{align}
If all data needed in the formula \eqref{FE} are directly available, then\begin{align}\label{KP0}
 k(T, \, \rho) = & - \varrho \, \left(\int_{T^0}^T \int_{T^0}^{\theta} \frac{\hat{c}_p(\theta^{\prime}, \, p^0, \, x(\rho))}{\theta^{\prime}} \, d\theta^{\prime}d\theta + T \, \hat{s}(T^0, \, p^0, \, x(\rho)) - \hat{h}(T^0, \, p^0, \, x(\rho)) \right)\nonumber\\
 & - p^0 \, n(\rho) \, \hat{\upsilon}(T, \, p^0, \, x(\rho)) \, .
\end{align}
 The function $\rho \mapsto k(T, \, \rho)$ is positively homogeneous, which implies that $- k + \rho \cdot \nabla_{\rho} k = 0$. Hence $k$ contributes to the free energy without contributing to the pressure (see \eqref{thermo8c}$_1$). For this reason we might call this function {\it mechanically neutral}.
As an illustration, in the typical example of an ideal mixture, we obtain the expression
\begin{align}\label{Boltzmannbb}
 k(T, \ \rho) = \sum_{i=1}^N \mu^0_i(T) \, \rho_i + R \, T \, n(\rho) \, \sum_{i=1}^N x_i(\rho) \, \ln x_i(\rho), \qquad \mu_i^0(T) :=  - p^{0} \,  \partial_pg_{i}(T, \, p^0) \, .
\end{align}
In the case that all data needed in the formula \eqref{FE} are available, the function $k$ is just an abbreviation for \eqref{KP0}. It is however possible that the free energy is constructed from other data. In this case it might be useful not to further specify $k$.

\subsubsection*{The additive splitting}

We can now rewrite \eqref{FE} as
\begin{align}\label{FESimple}
f(T, \, \rho) =  k(T, \,\rho) + p^0 \, V(T, \, p^0, \, \rho) + \underbrace{\bar{V}(T, \, p(T, \,\rho), \, \rho) - p(T, \,\rho)}_{=: f_{\rm mech}} \, .
\end{align}
We note that $f$ is the sum of $k(T, \, \rho)+p^0 \, V(T, \, p^0, \, \rho) $ --which is positively homogeneous in $\rho$-- and of the mechanical part $f_{\text{mech}}$ involving the pressure.

In the introduction of the present paper, the free energy was constructed by means of the integration method of classical thermodynamics to obtain the formula \eqref{FE}. From now on, we shall adopt a slightly different viewpoint in order to study asymptotic limits: We assume that the function $V$ of \eqref{VolumeTot} and the function $k$ of \eqref{kgrund} are the relevant objects. Recall that, in fact, these new ''data'' are constructed from the thermodynamic functions $\hat{\upsilon}$, $\hat{c}_p(p^0)$, $\hat{h}(T^0,p^0)$ and $\hat{s}(T^0,p^0)$.

\section{The free energy function as a mathematical object}\label{freesol}

\subsection{Some preliminary remarks}

So far we mainly addressed the meaning of the free energy function within classical thermodynamics. From this viewpoint, the formula \eqref{FE} and its reformulation \eqref{FESimple} are meaningful for all temperatures and densities around the reference state $(T^0, \, \rho^0)$ of a certain physical system. The original data, or the inferred ''data'' $V$ and $k$ in \eqref{FESimple}, can be evaluated at $(T, p(T,\rho), \, x(\rho))$ for all $(T, \rho)$ in this neighborhood. If this region is moreover stable in the phase--diagram, the function $f$ in \eqref{FESimple} has to be postulated convex in $\rho$ and concave in $T$, for otherwise the data must be rejected as incompatible with the second law of thermodynamics.

Considering now the PDEs \eqref{thermo2}$_1$ of mass transport in multicomponent systems, a description of the free energy which is only local leads to state--constraints. These PDEs would loose the parabolic structure if the convexity/concavity behavior of the free energy is violated in a neighborhood of the solution. A local description of the free energy generates rather severe obstacles to solution, approximation or convergence analysis with available methods.

From this viewpoint, the best case is a description of the free energy for the entire state space ($T > 0$ and $\rho_1, \ldots,\rho_N > 0$ arbitrary)! Moreover, as long as we do not expect phase--transitions of the material, this description would be a stable one --by which we mean that $f$ is convex in $\rho$ and concave in $T$ for all states.

Such global free energy models can be constructed. The simplest procedure is to assume that the data in the formula \eqref{FESimple} are extrapolated in such a way that they can be evaluated at arbitrary $(T, \, \rho)$.

We shall next implement this idea into some mathematical formalism. Let us at first introduce the notations
\begin{align*}
\overline{\mathbb{R}}^N_{+} := & \{ \rho \in \mathbb{R}^N \, : \, \rho_i \geq 0 \text{ for } i = 1,\ldots,N\} \, ,\\
 S^1_+ := & \{ y \in \mathbb{R}^N_+ \, : \, |y|_1 = 1\} \, ,\\
 S^2_+ := & \{ \nu \in \mathbb{R}^N_+ \, : \, |\nu|_2 = 1\} \, ,
\end{align*}
where $| x |_{r} := (\sum_{i=1}^N |x_i|^r)^{1/r}$ is the $r-$norm (here $r=1,2$).

The function $\hat{\upsilon}$ occurring in \eqref{Volumeconstraint} is a positive function of the variables $T, \, p, \, x$ defined in a set $$\mathcal{D}_{\hat{\upsilon}} \subseteq \mathbb{R}_+ \times \mathbb{R} \times  S^1_+ \, .$$

The proper range of pressures and compositions for which the volume of a mixture exhibits a smooth dependence --or can be measured at all-- is obviously a delicate topic: We think of limited possibility of the measurement apparatus, of phenomena like phase transitions, or of complex mixtures, for instance  electrolytes, where the charged constituents cannot vanish independently of each other. For the reasons mentioned just before, we now assume that the volume function $\hat{\upsilon}$ possesses a domain of the form
\begin{align}\label{DOM}
\mathcal{D}_{\hat{\upsilon}} := ]T_{\inf}, \, T_{\sup}[\times ]p_{\inf}, \, p_{\sup}[ \times S^1_+\, .
\end{align}
The latter means that $\hat{\upsilon}$ is well-defined for \emph{all possible compositions} vectors $x \in S^1_+$ inside of some thresholds $p_{\inf} < p_{\sup} \in \mathbb{R} \cup\{+\infty\}$ of pressure, and  $0 < T_{\inf} \leq T_{\sup} \leq + \infty$ of temperature \emph{which will be assumed independent of each other}.

Minimal regularity assumptions are that the function $(T, \, p) \mapsto \hat{\upsilon}(T, \, p, \, x)$ is continuous for all $x \in S^1_+$ fixed, and that the map $x \mapsto \hat{\upsilon}(T, \, p, \, x)$ possesses continuous extensions up to the boundary of $S_1^+$ (class $C(\overline{S}^1_+)$) for all $p \in ]p_{\inf}, \, p_{\sup}[$. For short, we write\begin{align}\label{conti}
 \hat{\upsilon} \in  C(]T_{\inf}, \, T_{\sup}[ \times ]p_{\inf}, \, p_{\sup}[ \times \overline{S}^1_+) \, .
\end{align}

The positively homogeneous volume-potential $V$ introduced in \eqref{VolumeTot} is then defined in $$\mathcal{D}_V :=  ]T_{\inf}, \, T_{\sup}[ \times ]p_{\inf}, \, p_{\sup}[ \times \mathbb{R}^N_{+} \, .$$ Under the assumption \eqref{conti}, we have $V \in C(]T_{\inf}, \, T_{\sup}[ \times ]p_{\inf}, \, p_{\sup}[ \times \overline{\mathbb{R}}^N_+) $.

If $\hat{\upsilon}$ is strictly positive, which is an obvious physical requirement, then
\begin{align}\label{posV}
V(T, \, p, \, \rho) > 0 \text{ for all } \rho \in \overline{\mathbb{R}}^N_+, \,\rho \neq 0 \text{ and all } (T, \, p) \in ]T_{\inf}, \, T_{\sup}[ \times ]p_{\inf}, \, p_{\sup}[ \, .
\end{align}
Next we want to evaluate the free energy formula for all states. This implies that we must find solutions to the equation \eqref{arbeitvolume} for states $\rho$ with arbitrary small or large norm. Hence it is necessary that the asymptotic values of $V$ obey
\begin{align}\label{consistency}
 \lim_{p \rightarrow p_{\inf}} V(T, \, p, \, \rho) = + \infty, \quad  \lim_{p \rightarrow p_{\sup}} V(T, \, p, \, \rho) = 0 \, \text{ for all } \rho \in \mathbb{R}^N_+ \, .
\end{align}
Since we consider the system in the range where it is stable, then surely\begin{align}\label{decreaseV}
& \hat{\upsilon}(T, \, p_2, \, x) < \hat{\upsilon}(T, \, p_1, \, x) \quad \text{ and } \quad V(T, \, p_2, \, \rho) < V(T, \, p_1, \, \rho) \nonumber\\
\text{ for all } & x \in S^1_+, \, \rho \in \mathbb{R}^N_+ \text{ and } p_{\inf} < p_1 < p_2 < p_{\sup}  \, .
\end{align}

To see that this condition is necessary in a smooth setting, we recall that the mass densities as functions of the variables $(T, \, p, \, x)$ are defined via \eqref{hatrho}.
We rephrase $\eqref{thermo8c}_1$ as $p = -f(T, \, \hat{\rho}) + \sum_{i=1}^N \hat{\rho}_i \, \hat{\mu}_i$, with $f = \varrho\psi$, $\hat{\rho}$ defined via \eqref{hatrho}, and $\hat{\mu} := \nabla_{\rho}f(T, \, \hat{\rho})$ depending on $T$, $p$ and $x$. Differentiation with respect to $p$ yields
\begin{align}\label{vp<0}
1 = \hat{\rho} \cdot D^2_{\rho,\rho} f(T, \, \hat{\rho}) \, \partial_p \hat{\rho} = - \frac{\hat{\upsilon}_p}{\hat{\upsilon}} \,  \hat{\rho} \cdot D^2_{\rho,\rho} f(T, \, \hat{\rho}) \, \hat{\rho} \, ,
\end{align}
proving that $\hat{\upsilon}$ must be strictly decreasing in $p$ for convex $f$.

For a strictly decreasing volume function, we can define the concept of a thermal equation of state.
\begin{defin}[Thermal equation of state for the pressure]\label{Def1}
Assume \eqref{conti},
and that the function $\pi \mapsto V(T, \, \pi, \, \rho)$ satisfies the conditions \eqref{consistency} and \eqref{decreaseV}.
Then there exists $p \in C(]T_{\inf}, \, T_{\sup} [ \times \overline{\mathbb{R}}^N_+)$ such that \eqref{arbeitvolume} is valid if and only if $p = p(T, \, \rho)$. Note that $p_{\inf} < p(T, \, \rho) < p_{\sup}$ by definition.
\end{defin}
\begin{rem}[Threshold of pressure]\label{p+infty}
In consequence of \eqref{conti}, \eqref{posV} and \eqref{consistency}, we note that for a finite upper pressure threshold $p_{\sup} < + \infty$, the molar volume satisfies $\hat{\upsilon}(T, \, p_{\sup}, \, x) = 0$. Of course, it is perfectly meaningful to formulate a free energy model under the restriction that its validity does not apply to pressures exceeding some number $p_{\sup}$. This is in the present context the source of conceptional difficulties though, because it seems to imply that we can have volume zero at finite pressures. Mathematically, we will see that it is not possible to construct a free energy function of Legendre--type on $\mathbb{R}^N_+$ for models having a finite pressure threshold from above. Thus, it is certainly reasonable to assume that
\begin{align*}
p_{\sup} = + \infty \, .
\end{align*}
\end{rem}

\subsection{A PDE for the free energy}

For $T$ fixed, the reference isobar $S_0 = S_0(T)$ at given $p^0$ is defined as
\begin{align*}
S_0 = \{\rho \in \mathbb{R}^N_+ \, : \, V(T, \, p^0, \, \rho) = 1 \} \, .
\end{align*}
This is a well--defined hyper--surface in $\mathbb{R}^N_+$: For $x \in S^1_+$ fixed, the line $\{t \, x \, : \, t > 0\}$ intersects $S_0(T)$ in exactly one point $\rho = \hat{\rho}(T, \, p^0, \, x)$ with $\hat{\rho}$ defined in \eqref{hatrho}.

We can then directly verify that $f$ defined by \eqref{FESimple} solves the following first order \emph{linear} PDE boundary-value-problem (bvp), in which the temperature is only a parameter:
\begin{alignat}{2}\label{probPbulk}
-f(T, \, \rho) + \rho \cdot \nabla_{\rho} f(T, \, \rho) = & p(T, \, \rho)  & &  \text{ for } \rho \in \mathbb{R}^N_+ \, ,\\
\label{probPbounda} f(T, \, \rho) = & k(T, \, \rho) & & \text{ for } \rho \in S_0(T) \, .
\end{alignat}
If the ''data'' $V$ and $k$ are not smooth -- for instance merely continuous -- the formula \eqref{FESimple} does not provide a continuously differentiable solution to the equation \eqref{probPbulk} anymore. Still, in the case that $f$ is convex in the variables $\rho_1,\ldots,\rho_N$, we can prove (see the Appendix) that it is sub--differentiable and satisfies
\begin{align}\label{bulkweak}
-f(T, \, \rho) + \rho \cdot \mu = p(T, \, \rho) \text{ for all } \rho \in \mathbb{R}^N_+ \text{ such that } \partial_{\rho} f(T,\rho) \neq \emptyset, \,  \mu \in \partial_{\rho} f(T,\rho) \, ,
\end{align}
where $\partial_{\rho} f$ is the subdifferential of $\rho \mapsto f(T, \, \rho)$.

The free energy of an incompressible system cannot, however, be interpreted as solution to \eqref{probPbulk} or even \eqref{bulkweak}. The reason is simple: the incompressibility is defined as $\partial_p V(T, \, p, \, \rho) = 0$, so that the definition of a function $p(T,\rho)$ as implicit solution to $V(T,\pi,\rho ) =1$ is ill posed.

In order to study singular limits, we shall therefore rely on an equivalent formulation of \eqref{probPbulk}, \eqref{probPbounda} which avoids introducing the function $p(T, \, \rho)$ explicitly. Due to \eqref{arbeitvolume} and the definition \eqref{kgrund}, the function $f$ also solves the first order \emph{nonlinear} bvp
\begin{alignat}{2}\label{probPbulkneu}
 V\big(T, \, -f(T, \, \rho) + \rho \cdot \nabla_{\rho} f(T, \, \rho), \, \rho\big) = & 1  & &  \text{ for } \rho \in \mathbb{R}^N_+ \, ,\\
\label{probPboundaneu} f(T, \, \rho) = & k(T, \, \rho) & & \text{ for } \rho \in S_0(T) \, .
\end{alignat}

The concept of solution for the problem \eqref{probPbulkneu}, \eqref{probPboundaneu} can as well be generalized in the framework of convex analysis as follows: Find a free energy function $f: \, ]0, \, +\infty[ \times \mathbb{R}^N_+ \rightarrow \mathbb{R} \cup \{+ \infty\}$ such that
\begin{align}\label{probPprimebulk}
 V\big(T, \, -f(T, \, \rho) + \rho \cdot \mu, \, \rho\big) =   1   \text{ for all } \rho \in \mathbb{R}^N_+ \text{ s.\ t.\ } \partial_{\rho} f(T, \rho) \neq \emptyset, \,  \mu \in \partial_{\rho} f(T,\rho) \, ,\\
\label{probPprimebounda} f(T, \, \rho) =  k(T, \, \rho)   \text{ if } \rho \in S_0(T) \, .
 \end{align}
Such weaker solutions to the problem \eqref{probPbulk} are precisely the relevant class to study mathematically the incompressible limit.
We next define the concept of a solution to the Gibbs-Duhem/Euler equation\footnote{Note that for the integration of \eqref{probPbulk}, the temperature can be treated as a parameter.}.
\begin{defin}\label{DEF}
Assume that $V \in C(\{T\} \times ]p_{\inf}, \, p_{\sup}[ \times \overline{\mathbb{R}}^N_+)$ is given, and that $p \in C(\{T\} \times\overline{\mathbb{R}}^N_+)$ is the corresponding thermal equation of state for the pressure.
\begin{enumerate}[(1)]
\item We call $f = f(T, \, \cdot) \in C^1(\mathbb{R}^N_+)$  a classical solution to the Gibbs-Duhem equation if the identity $-f + \rho \cdot \nabla_{\rho} f = p$ is valid in $\mathbb{R}^N_+$ {\rm (}cf.\ \eqref{probPbulk}{\rm )}. We call $f(T, \, \cdot)$ a classical solution \emph{of Legendre--type} if the two following additional conditions are satisfied:
\begin{itemize}
 \item $f(T, \, \cdot)$ is strictly convex;
 \item $f(T, \, \cdot)$ is essentially smooth {\rm (}see \cite{rockafellar}, page 251{\rm )}: $\lim_{m\rightarrow \infty} |\nabla_{\rho} f(T,\rho^m)| = +\infty$ for all $\{\rho^m\}_{m\in \mathbb{N}} \subset \mathbb{R}^N_+$ approaching a boundary point of $\mathbb{R}^N_+$.\footnote{Essential smoothness for a strictly convex function implies that the classical Legendre transform and the operation of convex conjugation are equivalent.}
 \end{itemize}
The solution is moreover called co-finite if $\nabla_{\rho} f(T) \big(\mathbb{R}^N_+\big) = \mathbb{R}^N$ {\rm (}$\nabla_{\rho} f(T, \, \cdot)$ is surjective from $\mathbb{R}^N_+$ onto $\mathbb{R}^N${\rm )}.

\item A (strictly) convex function $f = f(T, \, \cdot): \, \mathbb{R}^N_+ \rightarrow ]-\infty, \, + \infty]$ is called weak solution to the Gibbs-Duhem equation if $f$ is not identically $+\infty$, and for every $\rho \in \mathbb{R}^N_+$ and $\mu \in \partial_{\rho} f(T, \, \rho)$, the identity $V(T, \, -f + \mu \cdot \rho, \, \rho) = 1$ is valid {\rm (}cf.\ \eqref{probPprimebulk}{\rm )}. If $\partial_{\rho} f(T,\rho) = \emptyset$, the condition is assumed to hold vacuously. We call $f(T, \, \cdot)$ co-finite if the subdifferential $\partial_{\rho} f(T, \, \cdot)$ maps $\mathbb{R}^N_+$ onto $\mathbb{R}^N$.
\end{enumerate}
\end{defin}
\begin{rem}
For a co-finite, strictly convex function $f$ with domain in $\mathbb{R}^N_+$, the convex conjugate $f^*(x) := \sup_{\rho \in \mathbb{R}^N_+} (x \cdot \rho - f(\rho))$ is continuously differentiable in the whole of $\mathbb{R}^N$. From the point of view of mathematical analysis of models for multicomponent mass transport in fluids, this property is very important, as shown in \cite{dredrugagu20}, \cite{druetmixtureincompweak}, \cite{bothedruet}, \cite{bothedruetincompress}.
\end{rem}

\subsection{Mathematical properties of the solution formula} \label{mathprop}

The equations \eqref{probPbulk}, \eqref{probPbounda}, in which $T$ is only a parameter, determine the free energy function uniquely. \emph{A-posteriori} it will of course be necessary to verify the concavity in the variable $T$, which imposes compatibility restrictions on the data $V$ and $k$. As far as only the integration of the PDE is concerned we shall, for the sake of simplicity, not include the temperature in the notation.
\begin{prop}\label{formula}
Suppose that $V = V(T) \in C^{1,2}(]p_{\inf}, \, p_{\sup}[ \times \mathbb{R}^N_+)\cap C(]p_{\inf}, \, p_{\sup}[ \times \overline{\mathbb{R}}^N_+) $ satisfies $V > 0$, $\partial_p V < 0$ at all $(p, \, \rho) \in \mathcal{D}_V := ]p_{\inf}, \, p_{\sup}[ \times \mathbb{R}^N_+$. Suppose that $p^0 \in ]p_{\inf}, \, p_{\sup}[$ is a fixed reference value. For $(p, \, \rho) \in \mathcal{D}_V$ we introduce the primitive
\begin{align}\label{primitive}
\bar{V}(p, \, \rho) = \int_{p^0}^p V(p^{\prime}, \, \rho) \, dp^{\prime} \, ,
\end{align}
and we denote by $p \in C^1( \mathbb{R}^N_+) \cap C(\overline{\mathbb{R}}^N_+)$ the associated equation of state according to Def.\ \ref{Def1}. Assume that $k \in C^2(\mathbb{R}^N_+)$ is a positively homogeneous function. Then the formula
\begin{align}\label{solform}
f(\rho) :=  k(\rho) +p^0 \, V(p^0, \, \rho) + \bar{V}(p(\rho), \, \rho) - p(\rho) \,
\end{align}
provides the unique classical solution of class $C^2(\mathbb{R}^N_+)$ to the bvp $-f + \rho \cdot \nabla_{\rho} f = p$ in $\mathbb{R}^N$ with $f(\rho) = k(\rho)$ if $p(\rho) = p^0$ {\rm (}cf.\ \eqref{probPbulk}, \eqref{probPbounda} or, equivalently, \eqref{probPbulkneu}, \eqref{probPboundaneu}{\rm )}. The derivatives of the solution satisfy
\begin{align}\label{Gradienth}
\partial_{\rho_i} f(\rho) = & p^0 \, V_{\rho_i}(p^0, \, \rho) + k_{\rho_i}(\rho) + \bar{V}_{\rho_i}(p(\rho), \, \rho) \, ,\\
\label{Hessianh}
 \partial^2_{\rho_i,\rho_j} f(\rho)  = & p^0 \, V_{\rho_i,\rho_j}(p^0, \, \rho) + k_{\rho_i,\rho_j}(\rho) + \bar{V}_{\rho_i,\rho_j}(p(\rho), \, \rho) \nonumber\\
 & - \frac{1}{V_{p}(p(\rho), \, \rho)} \, V_{\rho_i}(p(\rho), \, \rho) \, V_{\rho_j}(p(\rho), \, \rho) \, .
\end{align}
The solution is (strictly) convex whenever the function $\rho \mapsto p^0 \, V(p^0, \, \rho) + k(\rho) + \bar{V}(p, \, \rho)$ is (strictly) convex for all $p \in ]p_{\inf}, \, p_{\sup}[$ or, equivalently, if it is (strictly) sub-additive.
\end{prop}
We next study additional properties of the solution.
\begin{prop}\label{exi}
We adopt the assumptions of Prop.\ \ref{formula} and assume that
\begin{enumerate}[(a)]
\item \label{condition-1} $p_{\sup} = + \infty$;
\item\label{condition0} $\lim_{p\rightarrow +\infty} V(p, \, \rho) = 0$, $\lim_{p\rightarrow p_{\inf}} V(p, \, \rho) = + \infty$ for all $\rho \in \mathbb{R}^N_+$;
 \item \label{condition1} $\lim_{p\rightarrow +\infty} \bar{V}(p, \, \rho) = + \infty$ and $\lim_{p\rightarrow p_{\inf}} \bar{V}(p, \, \rho) = -\infty$ for all $\rho \in \mathbb{R}^N_+$;
\end{enumerate}
 Assume further that $k \in C^2(\mathbb{R}^N_+)$ is a given positively homogeneous convex function and there exist $p^0 \in ]p_{\inf}, \, +\infty[$ such that:
\begin{enumerate}[(a)]
\addtocounter{enumi}{+3}
\item \label{condition4} The positively homogeneous function $\bar{k}(\rho) := k(\rho) + p^0 \, V(p^0, \, \rho)$ is convex and its restriction to $S^1_+$ is strictly convex and essentially smooth;
\item \label{condition4prime} For all $p_{\inf} < p_1 < p_2 < + \infty$ the quantity $\inf_{p \in ]p_1, \, p_2[} | \bar{k}_{\rho}(y) + \bar{V}_{\rho}(p, \, y)|$ tends to $+\infty$ for $y \in S^+_1$, $y \rightarrow \partial S^1_+$;
\item \label{suprem} $\sup_{y \in S^1_+} \{p^0 \, V(p^0, \, y) + k(y) + \bar{V}(p_m, \, y)\} \rightarrow -\infty \text{ for } p_m \rightarrow p_{\inf}$;
\item \label{inf2} $
\inf_{y \in S^1_+} \{p^0 \, V(p^0, \, y) + k(y) + \bar{V}(p_m, \, y)\} \rightarrow + \infty \text{ for } p_m \rightarrow +\infty$;
\item \label{condition5} For all $y \in S^1_+$ and $p > p_{\inf}$, the matrix $D^2_{\rho,\rho} \bar{k}(y) + D^2_{\rho,\rho} \bar{V}(p, \, y)$ possesses $N-1$ positive eigenvalues.
\end{enumerate}
Then $f$ defined by \eqref{solform} is the unique co-finite, classical solution of Legendre--type to $- f + \rho \cdot \nabla_{\rho}  f= p$ such that $f = k$ on the surface $\{\rho \in \mathbb{R}^N_+  \, : \, V(p^0, \, \rho) = 1\}$. Moreover, $f$ belongs to $C^2(\mathbb{R}^N_+)$ and $D^2f(\rho)$ is positive definite for all $\rho \in \mathbb{R}^N_+$.
\end{prop}
A proof for the two Propositions \ref{formula} and \ref{exi} can be found in the appendix, section \ref{detrop}.

Thus, due to the Gibbs-Duhem equation, the free energy obeys the representation \eqref{FESimple}. Now, the data in this formula are also restricted by the requirement of concavity in temperature. In order to verify the validity of this condition, we compute
\begin{align}\label{FEconcTun}
\partial^2_{T,T} f(T, \, \rho) = & \partial^2_{T,T}k(T, \, \rho) + p^0 \, \partial^2_{T,T} V(T, \, p^0, \, \rho) + \int_{p^0}^{p(T,\rho)}  \partial^2_{T,T}V(T, \, p^{\prime}, \, \rho) \, dp^{\prime} \nonumber\\
& - \frac{\partial_TV(T, \, p(T,\rho), \, \rho)}{\partial_pV(T, \, p(T,\rho), \, \rho)} \, .
\end{align}
Hence, in addition to the restrictions formulated in the Propositions \ref{formula} and \ref{exi}, the functions $V$ and $k$ are restricted by the condition that the right-hand side in \eqref{FEconcTun} is negative.

To be even more specific, employing the full formula \eqref{FE} and using this time the representation in the variables $(T, \, p, \, x)$, we obtain
\begin{align}\label{FEconcT}
\partial^2_{T,T} \hat{f}(T, \, p, \, x) = & - \frac{1}{\hat{\upsilon}(T, \, p, \, x)} \,  \frac{(\partial_T\hat{\upsilon}(T, \, p, \, x))^2}{\partial_p\hat{\upsilon}(T, \, p, \, x)} + \frac{1}{\hat{\upsilon}(T, \, p, \, x)} \, \int_{p^0}^p \partial^2_{T,T} \hat{\upsilon}(T, \, p^{\prime}, \, x) \, dp^{\prime} \nonumber\\
& - \frac{M(x)}{\hat{\upsilon}(T, \, p, \, x)} \,  \frac{\hat{c}_p(T, \, p^0, \, x)}{T} \, ,
\end{align}
and the data in \eqref{FE} are restricted, for all relevant $(T,p,x)$, by the inequality
\begin{align}\label{compatoublie}
 \int_{p^0}^p \partial^2_{T,T} \hat{\upsilon}(p^{\prime}, \, \cdot) \, dp^{\prime} - \frac{M(x)}{T} \,  \hat{c}_p(p^0, \, \cdot) <  \frac{(\partial_T\hat{\upsilon})^2}{\partial_p\hat{\upsilon}} \, .
\end{align}

\section{The convergence result}\label{mainres}

The main result in the mathematical part of the paper concerns deriving the volume constraint which defines the incompressible limit and the structure of the free energy. Let $T \in ]T_{\inf}, \, T_{\sup}[$ be fixed. We consider a sequence of volume functions
\begin{align}\label{Vregum}
                                                                                                                                                                                            \{V^m(T)\}_{m\in \mathbb{N}} \subset C^{1,2}(]p_{\inf}^m, \, + \infty[ \times \mathbb{R}^N_+) \cap C(]p_{\inf}^m, \, + \infty[ \times \overline{\mathbb{R}}^N_+ ) \, .
                                                                                                                                                                                            \end{align}
The reference value $p^0$ belongs to $]p_{\inf}^m, \, +\infty[$ for all $m$.
We define $$p_{\inf} := \limsup_{m\rightarrow \infty} p_{\inf}^m \, ,$$and we assume that $p_{\inf}$ is any real value or $\{-\infty\}$.

As to the general procedure, we \underline{assume} that for all values of the parameter $m$, a free energy function $f^{m}(T, \cdot)$ is \emph{globally available} on $\mathbb{R}^N_+$ and $\rho \mapsto f^m(T, \, \rho)$ satisfies all conditions for a smooth co-finite function of Legendre--type.

To construct this function, we consider sequences $\{V^m(T)\}$ with the regularity \eqref{Vregum} and, for all $m$,
\begin{align}\label{kregum}
\{k^m(T)\} \subset C^{2}(\mathbb{R}^N_+), \quad \text{ subject to } \quad \rho \mapsto k^m(T, \, \rho) \text{ pos.\ homog.\ and convex on } \mathbb{R}^N_+ \, .
\end{align}
For each $m$, we assume that the pair of $V^m(T)$ and $k^m(T)$ satisfies the assumptions of Prop.\ \ref{formula} and \ref{exi}. Then, there is a unique classical co-finite solution of Legendre--type $f^m(T) \in C^2(\mathbb{R}^N_+)$ to the problem \eqref{probPbulk}, \eqref{probPbounda} (equivalently \eqref{probPbulkneu}, \eqref{probPboundaneu}) with data $V^m(T), \, k^m(T)$.

We further \underline{assume} some natural convergence properties for the data $V^m$ and $k^m$: These are:
\begin{enumerate}[(a)]
\addtocounter{enumi}{+8}
 \item \label{belle1} $V^m(T)$ converges uniformly on compact subsets of $]p_{\inf}, \, +\infty[ \times \overline{\mathbb{R}}^N_+$, and $D^2_{\rho,\rho}V^m(T)$ converges uniformly on compact subsets of $]p_{\inf}, \, +\infty[ \times \mathbb{R}^N_+$.
 \item \label{belle2} There exists a function $k(T)$ such that $k^m(T,\cdot) \rightarrow k(T,\cdot)$ uniformly on compact subsets of $\overline{\mathbb{R}}^N_+$ and $D^2_{\rho,\rho}k^m(T,\cdot)$ converges uniformly to $D^2_{\rho,\rho}k(T,\cdot)$ on compact subsets of $\mathbb{R}^N_+$. Moreover, the limit $k(T)$ is essentially smooth on $S^1_+$ and $D^2_{\rho,\rho}k(T, \, y)$ possesses $N-1$ strictly positive eigenvalues for all $y \in S^1_+$.
 \end{enumerate}

Then we prove that also $f^m(T, \cdot)$ converges, and we \underline{derive} necessary structural consequences for the limit. As mentioned in the introduction, the structure of the limit depends on the range of pressures for which we have pressure-independence of the volume, and we consider two cases: Incompressibility in the entire state space or incompressibility in a bounded subregion.

In the first scenario, we let the isothermal compressibility tend to zero in the entire domain of definition. This case exhibits two striking features: First, the incompressiblity constraint is necessarily linear; second the free energies epi--converge (\cite{rockawets}, Def. 7.1), or Gamma- or Mosco-converge\footnote{Since we are in a finite-dimensional convex setting, these concepts are all equivalent.} to the limit free energy.
\begin{prop}\label{MAIN}
In addition to the assumptions \eqref{Vregum}, \eqref{kregum}, \eqref{belle1} and \eqref{belle2} for $V^m$ and $k^m$, we assume that $p_{\inf} = -\infty$, and that
\begin{enumerate}[(a)]
\addtocounter{enumi}{+10}
\item \label{belle0} $\partial_p V^m \rightarrow 0$ pointwise in $]-\infty, \, +\infty[ \times \mathbb{R}^N_+$.
 \end{enumerate}
Let $V^{\infty}(\rho) := \lim_{m\rightarrow \infty} V^m(p^0, \, \rho)$ and assume that $V^{\infty}$ is strictly positive on $\overline{\mathbb{R}}^N_+ \setminus\{0\}$.

Then there is a fixed vector $\bar{\upsilon}\in \mathbb{R}^N_+$ such that $V^{\infty}(\rho) = \sum_{i=1}^N\bar{\upsilon}_i \,  \rho_i$. The convex function
 \begin{align*}
  f^{\infty}(\rho) := \begin{cases}
                                 k(\rho) & \text{ if } V^{\infty}(\rho) = 1 \, ,\\
                                 + \infty & \text{ otherwise, }
                                \end{cases}
 \end{align*}
is the limit of $\{f^m\}_{m\in \mathbb{N}}$ in the sense of epi-convergence, and $f^{\infty}$ is a strictly convex weak solution to \eqref{probPbulkneu}, \eqref{probPboundaneu}.
Moreover, for all $(\mu, \, \rho) \in \mathbb{R}^N \times \mathbb{R}^N_+$ the following statements are equivalent:
\begin{enumerate}[(1)]
 \item $\mu \in \partial f^{\infty}(\rho)$;
 \item $\mu = p \, \bar{\upsilon} + \nabla_{\rho} k(\rho)$ with $p = (f^{\infty})^*(\mu)$, $(f^{\infty})^*$ being the convex conjugate of $f^{\infty}$.
 \end{enumerate}
The pressures $p^m(\rho) = -f^m(\rho) + \rho \cdot \nabla_{\rho}f^m(\rho)$ satisfy the set convergence
\begin{align*}
 \text{Graph}(p^m) = \{(\rho, \, p^m(\rho)) \, : \, \rho \in \mathbb{R}^N_+\} \longrightarrow & \{(\rho, \, -f^{\infty}(\rho) + \mu \cdot \rho) \, : \,  \bar{\upsilon} \cdot \rho = 1, \, \mu \in \partial f(\rho)\}\\
& =  \{(\rho, \, p) \, : \, \rho >0, \, \bar{\upsilon} \cdot \rho = 1, \, p \in \mathbb{R}\} \, ,
\end{align*}
where $\longrightarrow $ denotes the set convergence of Painlev\'e--Kuratowski in $\mathbb{R}^{N+1}$ (see \cite{rockawets}, Ch.\ 4).
\end{prop}
This type of limit describes the free energy of an incompressible phase assumed globally stable. The linearity of $V^{\infty}$ is a major drawback with respect to description of real incompressible mixtures. In the next situation, we can avoid this conclusion by suitably restricting the incompressible phase.
\begin{prop}\label{MAINNN}
Assume \eqref{Vregum} and \eqref{kregum}, \eqref{belle1} and \eqref{belle2} for $V^m$ and $k^m$. Suppose that $p_{\inf} \in \mathbb{R}$,
and, instead of \eqref{belle0}, that
\begin{enumerate}[(a')] \addtocounter{enumi}{+10}
 \item \label{belle0pr} For $x \in S_1^+$, there are real numbers $a(T, \, x) < p^0 <  b(T, \, x)$ such that $\partial_p V^m(T,\pi, \, x) \rightarrow 0$ for all $\pi \in ]a(T, \, x), \, b(T, \, x)[$.
\end{enumerate}
For the limit function $V^{\infty}$ we moreover assume that:
\begin{enumerate}[(a)]
\addtocounter{enumi}{+11}
 \item\label{bien5strich} $\partial_pV^{\infty}(T,p, \, \rho) < 0$ for all $p < a(T, \, x(\rho))$ and all $p > b(T,x(\rho))$. Moreover,  for all $x \in S^1_+$,
\begin{gather*}
\lim_{p\rightarrow + \infty} V^{\infty}(T,p, \, x) = 0 \, ,\\
\lim_{p\rightarrow p_{\inf}} \int_p^{a(T, \, x)} V^{\infty}(T, \, p^{\prime}, \, x) \, dp^{\prime} = + \infty, \quad \lim_{p\rightarrow + \infty} \int_{b(T, \, x)}^{p} V^{\infty}(T,p^{\prime}, \, x) \, dp^{\prime} = +\infty \, ;
\end{gather*}

 \item \label{bien6strich} $|\nabla_{\rho}V^{\infty}| \leq C$ in $]p_{\inf}, \, + \infty[\times \mathbb{R}^N_+$ for some constant $C>0$.
\end{enumerate}
Then the threshold functions $a$, $b$ are subject to the necessary condition
\begin{align*}
 \inf_{\pi \in [a(T, \, x), \, b(T, \, x)], \, \eta \in \mathbb{R}^N} \big(D^2_{\rho,\rho}V^{\infty}(T, \, p^0, \, x) + \pi \, D^2_{\rho,\rho}k(T, \, x)\big) \eta \cdot \eta \geq 0 \, \quad \text{ for all } x \in S^1_+ \, .\end{align*}
Moreover, for all $\rho \in \mathbb{R}^N_+$ subject to $V^{\infty}(T,p^0, \, \rho) \neq 1$, the equation $V^{\infty}(T,\pi,\rho) = 1$ possesses a unique solution $\pi = p(T,\rho)$ and
\begin{align*}
\begin{cases} p(T,\rho) > b(T,x(\rho)) & \text{ if } V^{\infty}(T,p^0, \, \rho) >1\, ,\\
 p(T,\rho) < a(T,x(\rho)) & \text{ if } V^{\infty}(T,p^0, \, \rho) < 1 \, .
\end{cases}
\end{align*}
At fixed $T$, the functions $\{f^m(T)\}_{m\in \mathbb{N}}$ converge uniformly on compact subsets of $\mathbb{R}^N_+$ to the continuous convex function
 \begin{align*}
  & f^{\infty}(T,\rho) := \\
  & \begin{cases}
                         k(T,\rho) + p^0 \, V^{\infty}(T,p^0,\, \rho) +\bar{V}(T,p(T,\rho), \,\rho) - p(T,\rho) & \text{ for } \rho  \text{ s.\ t.\ } V^{\infty}(T,p^0,\, \rho) \neq 1 \, ,\\
                         k(T,\rho) & \text{ for } \rho  \text{ s.\ t.\ } V^{\infty}(T,p^0,\, \rho) = 1
                         \end{cases}
 \end{align*}
and $f^{\infty}(T)$ is a strictly convex weak solution to \eqref{probPbulkneu}, \eqref{probPboundaneu} for $V = V^{\infty}$.
A point $\mu \in \mathbb{R}^N$ belongs to $\partial_{\rho} f^{\infty}(T,\rho)$ if and only if one of the two following conditions is valid:
\begin{itemize}
\item[\text{Either}] $V^{\infty}(T,p^0, \, \rho) \neq 1$ and $\mu = p^0 \, V^{\infty}_{\rho}(T,p^0,\, \rho) + \nabla_{\rho}k(T,\rho) + \bar{V}_{\rho}(T,p(T,\rho), \,\rho)$
\item[\text{or}] $V^{\infty}(T,p^0, \, \rho) = 1$ and $\mu = p \, V^{\infty}_{\rho}(T,p^0, \, \rho) + \nabla_{\rho}k(T,\rho)$ with a number $p$ subject to $a(T,x(\rho)) \leq p \leq b(T,x(\rho))$.
\end{itemize}
For all $\rho$ subject to $V^{\infty}(T,p^0, \, \rho) \neq 1$, the pressures converge as functions of $(T, \, \rho)$ and we have in fact $p^m(T, \rho) \rightarrow p(T,\rho)$.
\end{prop}
The condition \eqref{bien5strich} corresponds to requiring that we have a meaningful compressible model below and beyond the pressure thresholds $a, \, b$, while \eqref{bien6strich} is a simplifying assumption. Note also that, if we assume $b = +\infty$ (only a lower pressure threshold), then we obtain a mixed version of the limit. We will however not sketch this interesting case here.

%

In both cases we also recover, for variable temperature, the M\"uller--''paradox'' under the following general assumptions:
\begin{enumerate}[(a)]
\addtocounter{enumi}{+14}
\item \label{belle3} The function $T \mapsto k^m(T, \, \rho)$ is twice continuously differentiable on $]T_{\inf}, \, T_{\sup}[$ for all $\rho$, and the convergence $\partial^2_{T,T} k^m \rightarrow \partial^2_{T,T} k$ holds pointwise in $]T_{\inf}, \, T_{\sup}[ \times \mathbb{R}^N_+$;
 \item \label{belle4} The function $T \mapsto V^m(T, \, p, \, \rho)$ is twice continuously differentiable on $]T_{\inf}, \, T_{\sup}[$ for all $(p, \, \rho)$, and $\partial^2_{T,T} V^m \rightarrow \partial^2_{T,T} V^{\infty}$ pointwise in $]T_{\inf}, \, T_{\sup}[ \times ]p_{\inf}, \, + \infty[\times \mathbb{R}^N_+$;
 \item \label{belle5} For all $m$ the condition $\partial^2_{T,T} f^m < 0$ is valid on $]T_{\inf},T_{\sup}[  \times \mathbb{R}^N_+$. [In other words, the pair of $(V^m, \, k^m)$ satisfies the compatibility conditions necessary for $\partial^2_{T,T} f^m < 0$ being true (cp.\ \eqref{FEconcTun})].
 \end{enumerate}

\begin{lemma}\label{MUELLER}
If we adopt all assumptions of Proposition \ref{MAIN} and moreover \eqref{belle3}, \eqref{belle4} and \eqref{belle5}, then $\partial_T V^{\infty} = 0$.

If we adopt the assumptions of Proposition \ref{MAINNN}, and \eqref{belle3}, \eqref{belle4} and \eqref{belle5}, then $\partial_T V^{\infty}(T, \, \pi, \, \rho) = 0$ for all $(T, \, \pi, \, \rho)$ subject to $\pi \in [a(T,x(\rho)), \, b(T,x(\rho))]$.
\end{lemma}
The next sections are devoted to the proof of these statements.

\section{Passage to dual variables}\label{dualvar}

In this section we set up the main technical tool for performing asymptotic limits {\it in the main variables}. In the incompressible limit, the function $V$ is independent of pressure on parts of the interval $]p_{\inf}, \, + \infty[$.
We claim that it is useful to look at the equation \eqref{probPbulk} from the more abstract viewpoint of \eqref{probPbulkneu}, that is, we rephrase the Gibbs--Duhem equation as $V(-f + \rho \cdot \nabla_{\rho} f, \, \rho) = 1$ on $\mathbb{R}^N_+$. This formulation possesses another interpretation in terms of dual variables. Indeed, assuming that a co-finite solution $f$ of Legendre--type on $\mathbb{R}^N_+$ to the equation
\begin{align}\label{GIBBSDUHEM3}
p(\rho) = -f(\rho) + \sum_{i=1}^N \rho_i\, \partial_{\rho_i}f(\rho) \text{ in } \mathbb{R}^N_+
\end{align}
is at hand, then we are allowed to introduce on $\mathbb{R}^N$ the Legendre transform $g$ of $f$ via $ g(\nabla_{\rho} f(\rho)) := - f(\rho) + \rho \cdot \nabla_{\rho} f(\rho)$. Since $f$ is convex, we may also introduce on $\mathbb{R}^N$ the conjugated convex function $f^*$ via
\begin{align*}
 f^*(x) := \sup_{\rho \in \mathbb{R}^N_+} \{x \cdot \rho - f(\rho)\} \text{ for } x \in \mathbb{R}^N \, .
\end{align*}
Essentially smooth, strictly convex functions satisfy $g = f^*$ and $\nabla_{x} f^* = (\nabla_{\rho} f)^{-1}: \, \mathbb{R}^N \rightarrow \mathbb{R}^N_+$, which induces a helpful reformulation of the equations \eqref{probPbulkneu}.
\begin{lemma}\label{pdelemma}
Assume that $f$ is a function of Legendre--type on $\mathbb{R}^N_+$ solving \eqref{probPbulk} or, equivalently, \eqref{GIBBSDUHEM3}. Then, the function $g = f^* \in C^1(\mathbb{R}^N)$ is strictly convex on $\mathbb{R}^N$, and it solves the first order non-linear PDE
 \begin{align}\label{PDE}
 V(g, \, \nabla_x g) = 1  \text{ in } \mathbb{R}^N \, .
 \end{align}
Suppose, conversely, that $g \in  C^1(\mathbb{R}^N)$ is a strictly convex solution to \eqref{PDE} in $\mathbb{R}^N$ such that the image of $\nabla_x g$ is $\mathbb{R}^N_+$. Then $f := g^*$ is of Legendre--type on $\mathbb{R}^N_+$. Moreover, the identity $g(\nabla_{\rho} f) = \nabla_{\rho}  f \cdot \rho - f$ is valid on $\mathbb{R}^N_+$.
\end{lemma}
\begin{proof}
Under the assumptions of the Lemma, the relation \eqref{GIBBSDUHEM3} means nothing else but $p(\rho) = g(\nabla_{\rho} f(\rho))$. For $x := \nabla_{\rho} f(\rho)$, that is, $\rho = (\nabla_{\rho} f)^{-1}(x)$, we have $V(p(\rho), \, \rho) = 1$ if and only if $V(g(x), \, \nabla_{x} g(x)) = 1$. The additional properties follow from the main theorem on the Legendre transformation (\cite{rockafellar}, Theorem 26.5).
\end{proof}
In the next statement, we shall treat asymptotic limits in general, considering the limit model of converging sequences of volume functions $\{V^m\}_{m \in \mathbb{N}}$ and $\{k^m\}_{m \in \mathbb{N}}$. Afterwards, in the next section, we will consider specifically the incompressible limit.
%
\begin{lemma}\label{asymptotics}
Consider sequences
\begin{itemize}
\item $\{p_{\inf}^m\}_{m \in \mathbb{N}}$ of real numbers;
\item $\{V^m\}_{m \in \mathbb{N}}\subset C^{1,2}(]p_{\inf}^m, \, + \infty[ \times \mathbb{R}^N_+) \cap C(]p_{\inf}^m, \, + \infty[ \times \overline{\mathbb{R}}^N_+)$ of positive, in the first variable strictly decreasing functions;
\item $\{k^m\}_{m \in \mathbb{N}} \subset C^2(\mathbb{R}^N_+)$ of positively homogeneous functions.
\end{itemize}
We define $\bar{k}^m(\rho) := k^m(\rho) + p^0 \, V^m(p^0, \, \rho)$, and we denote by $H^m(p, \, y)$ the matrix $$\bar{k}_{\rho,\rho}^m(\frac{y}{V^m(p, \, y)}) + \bar{V}^m_{\rho,\rho}(p, \, \frac{y}{V^m(p, \, y)}) - \frac{V^m(p, \, y)}{V^m_p(p, \, y)} \, (V^m_{\rho}(p, \, y) \otimes V^m_{\rho}(p, \, y)) \, .$$
We assume that
\begin{enumerate}[(1)]
\item \label{pinf} $p_{\inf} := \limsup_{m\rightarrow \infty} p_{\inf}^m$ belongs to $[-\infty, \, +\infty[$;
\item \label{pointwisev} There exists the pointwise limit $V(p, \, \rho) := \lim_{m \rightarrow \infty} V^m(p, \, \rho)$ for all $(p, \, \rho) \in ]p_{\inf}, \, + \infty[ \times \mathbb{R}^N_+$;
\item \label{pointwisek} There exists the pointwise limit $k(\rho) := \lim_{m\rightarrow \infty} k^m(\rho)$ for all $\rho \in \mathbb{R}^N_+$;
\item \label{compactum} $\liminf_{m\rightarrow \infty} \inf_{y \in S^1_+} |\bar{k}^m(y) + \bar{V}^m(p, \, y)| \rightarrow + \infty$ for $p \rightarrow p_{\inf}$ and $p \rightarrow \, + \infty$.
\end{enumerate}
Moreover, for all $p_{\inf} < p_1 < p_2 < + \infty$ and $y_0 > 0$, we assume that
\begin{enumerate}[(1)]\addtocounter{enumi}{+4}
\item \label{compactmore} $\liminf_{m\rightarrow \infty} \inf_{p \in ]p_1, \, p_2[} |\bar{k}^m_{\rho}(y) + \bar{V}^m_{\rho}(p, \, y)| \rightarrow + \infty$ for $y \rightarrow \partial S^1_+$;
\item \label{Convex} $\liminf_{m \rightarrow \infty} \inf_{p \in ]p_1, \, p_2[, \, y \in S^1_+ \, : \, \inf_{i=1,\ldots, N}y_i\geq y_0} \lambda_{\min}(H^m(p, \, y)) > 0$.
\end{enumerate}
If there is a unique classical co-finite solution $f^m$ of Legendre--type to \eqref{probPbulkneu}, \eqref{probPboundaneu}, then
\begin{enumerate}[(i)]
\item \label{claimexi} For each $m \in \mathbb{N}$, there is a unique strictly convex function $g^m \in C^2(\mathbb{R}^N)$ solution to \eqref{PDE} for $V = V^m$, with boundary condition $g^m = p^0$ on
\begin{align*}
S^m_0 := & \nabla_{\rho}(p^0 \, V^m(p^0, \, \cdot) + k^m(\cdot))(\mathbb{R}^N_+) \\
 = &\{x \in \mathbb{R}^N \, : \, x = p^0 \, V^m_{\rho}(p^0, \, r) + k^m_{\rho}(r) \text{ for a } r \in \mathbb{R}^N_+\} \, ;
 \end{align*}
\item \label{claimexilimit} There is a convex function $g \in C^2(\mathbb{R}^N)$ such that $g^m \rightarrow g$ in $C^1(K)$ for all compact $K \subset \mathbb{R}^N$, and $g$ is a solution to $V(g, \, \nabla g) = 1$ in $\mathbb{R}^N$, with $g(x) = p^0$ if and only if $x \in S_0$. Here $S_0$ is the set of accumulation points of the family of manifolds $\{S_0^m\}_{m \in \mathbb{N}}$;
\item \label{claimimage} The image of $\nabla g$ satisfies
\begin{align*}
\nabla g(\mathbb{R}^N) = \bigcup_{p \in]p_{\inf}, \, +\infty[} \{ \rho \in \mathbb{R}^N_+ \, : \, V(p, \, \rho) = 1\} \, ;
\end{align*}
\item \label{claimregu} The convex conjugated function $f := g^*$ is a strictly convex function, which is finite and subdifferentiable in the set $\nabla g(\mathbb{R}^N)$ and otherwise infinite. It is a weak solution to \eqref{probPbulkneu}, \eqref{probPboundaneu};
\item \label{claimepi} $f$ is the limit of $f^m$ in the sense of epi-convergence, meaning that $$f(\rho) = \inf\{\liminf_{m\rightarrow \infty} f^m(\rho^m) \, : \, \rho^m \rightarrow \rho\} \, .$$
\end{enumerate}
\end{lemma}
\begin{proof}
\eqref{claimexi}. By assumption, there is for every $m \in \mathbb{N}$ a Legendre solution $f^m$ to the problem $V^m(-f^m + \rho \cdot \nabla f^m, \,\rho) = 1$ on $\mathbb{R}^N_+$ with $f^m = k^m$ if $V^m(p^0, \, \rho) = 1$. The convex conjugated functions $g^m$ are global solutions to $V^m(g^m,\, \nabla g^m) = 1$ (cf. Lemma \ref{pdelemma}).
Since $f^m = k^m$ for $\rho \in \Gamma_0 := \{\rho \, : \, V^m(p^0, \, \rho) = 1\}$, there must exist $\lambda = \lambda(\rho)$ such that
\begin{align*}
 \nabla_{\rho}f^m(\rho) = \nabla_{\rho}k^m(\rho) +  \lambda(\rho) \, V^m_{\rho}(p^0, \, \rho)  \text{ for all } \rho \in \Gamma_0 \, .
\end{align*}
We multiply with $\rho$, use the positive homogeneity of $k^m$ and $V^m(p^0, \cdot)$, and for all $\rho \in \Gamma_0$ it follows that
\begin{align*}
p^0 + f^m(\rho) & = \rho \cdot   \nabla_{\rho}f^m(\rho) =  k^m(\rho) +  \lambda(\rho) \, V^m(p^0, \, \rho) \\
& = k^m(\rho) +  \lambda(\rho) = f^m(\rho)+  \lambda(\rho) \, .
\end{align*}
Hence $\lambda(\rho) = p^0$ and $\nabla_{\rho}f^m(\rho) = \nabla_{\rho}k^m(\rho) +  p^0 \, V^m_{\rho}(p^0, \, \rho) = \nabla_{\rho}\bar{k}^m(\rho)$ for all $\rho \in \Gamma_0$. Since for $\rho \in \Gamma_0$, it holds that $g^m(\nabla_{\rho} f^m(\rho)) = p^0$, we have shown that $g_m = p^0$ on $S^m_0 = \nabla_{\rho}(k^m(\cdot)+  p^0 \, V^m(p^0, \cdot))(\mathbb{R}^N_+)$.

\eqref{claimexilimit}.
To prove the existence of a limit, we next show that for $K \subset \mathbb{R}^N$ compact, there is $K^{\prime} \subset \mathbb{R}^N_+$ compact and independent on $m$ such that $\nabla g^m (K) \subseteq K^{\prime}$. Consider $x \in K$ arbitrary and $m$ fixed. Due to the representation \eqref{Gradienth}, there are unique $ y = \nabla g^m(x)/|\nabla g^m(x)|_1 \in S^1_+$ and $s = p^m(\nabla g^m(x))$ such that $x = \nabla_{\rho}\bar{k}^m(y) + \bar{V}^m_{\rho}(s, \, y)$. Multiplying with $y$ we see that $x \cdot y = \bar{k}^m(y)  + \bar{V}^m(s, \, y)$, and this yields
\begin{align*}
 \inf_{j\geq m, \, y \in S^1_+} | \bar{k}^j(y)  + \bar{V}^j(s, \, y)| \leq |x|_{\infty} \leq \sup_{x \in K} |x|_{\infty} \, .
\end{align*}
We infer that $\liminf_{m \rightarrow \infty} \inf_{y \in S^1_+} | \bar{k}^m(y)  + \bar{V}^m(s, \, y)| \leq \sup_{x \in K} |x|_{\infty} $. Thus, owing to assumption \eqref{compactum}, there must exist certain constants $p_1, \, p_2$, depending only on $K$, such that $p_{\inf} < p_1 < s < p_2 < + \infty$. Since $x =  \nabla_{\rho}\bar{k}^m(y) + \bar{V}^m_{\rho}(s, \, y)$, we next infer that, for these numbers,
\begin{align*}
 \sup_{x \in K} |x| \geq \inf_{j \geq m, \, p \in ]p_1, \, p_2[}   |\nabla_{\rho}\bar{k}^j(y) + \bar{V}^j_{\rho}(p, \, y)| \, ,
\end{align*}
hence $\liminf_{m\rightarrow \infty} \inf_{p \in ]p_1, \, p_2[}   |\bar{k}^m_{\rho}(y) + \bar{V}^m_{\rho}(p, \, y)| \leq  \sup_{x \in K} |x|$. This time invoking \eqref{compactmore}, we see that $y$ remains in a compact subset of $S^1_+$ being independent on $m$. We conclude for $x \in K$ that the point $(y, \, s)$ such that $x = \nabla_{\rho}\bar{k}^m(y) + \bar{V}^m_{\rho}(s, \, y)$ is in some fixed compact subset of $S^1_+ \times ]p_{\inf}, \, + \infty[$, independently on $m$.

Consider now a point $\rho$ in the image $\nabla g^m(K) \subset \mathbb{R}^N_+$. Then $\rho = |\rho|_1 \, y$ where, as just shown, $y$ belongs to some compact subset of $S^1_+$ independently of $m$. Since $V^m(s, \, \rho) = 1$, it follows that $|\rho|_1  = 1/V^m(s, \, y)$. Hence
\begin{align*}
\frac{1}{V^m(p_1, \, y)} \leq |\rho|_1 \leq  \frac{1}{V^m(p_2, \, y)} \,
\end{align*}
where $p_1, \, p_2$ depend only on $K$. This shows that $\nabla g^m(K) \subset K^{\prime}$ for all $m$, where $K^{\prime}$ is a compact subset of $\mathbb{R}^N_+$.

Now, we invoke the assumption \eqref{Convex}. Note that for $\rho \in \mathbb{R}^N_+$, the formula \eqref{Hessianh} yields $D^2f^m(\rho) = H^m(s, \, y)$ where $y = \rho/|\rho|_1$ and $s = p^m(\rho)$. Assumption \ref{Convex} implies that there is $m_0 = m_0(K)$ such that, for all $m \geq m_0$,
\begin{align}\label{liminfconv}
\inf_{\rho \in K^{\prime}} \lambda_{\min}(D^2f^m(\rho)) \geq \inf_{j \geq m, \, p_1 \leq s \leq p_2, \, \rho \in K^{\prime}} \lambda_{\inf}(H^j(s, \, \frac{\rho}{|\rho|_1}))  > 0 \, .
\end{align}
For $\rho \in K^{\prime}$ and $x \in K$ connected via $\rho = \nabla g^m(x)$, the Hessian of $f^m$ and $g^m$ are inverse to each other: $D^2g^m(x) = [D^2f^m(\rho)]^{-1}$.
Hence \eqref{liminfconv} yields $ \sup_{x \in K, \, j \geq m} \lambda_{\max}(D^2g^j(x)) < + \infty$ for all $m \geq m_0(K)$.

We thus have shown that $\limsup_{m \rightarrow \infty} \|g^m\|_{C^2(K)} < + \infty$ for any compact subset of $\mathbb{R}^N$.
This implies the existence of a convex limit function $g \in C^2(\mathbb{R}^N)$ such that (passing to a subsequence if necessary) $g^m \rightarrow g$ and $\nabla g^m \rightarrow \nabla g$ uniformly on compact subsets of $\mathbb{R}^N$. As a corollary, the limit satisfies $V(g, \, \nabla g) =1$ in $\mathbb{R}^N$, where $V$ is the pointwise limit volume function.

To finish the proof of \eqref{claimexilimit}, it remains to discuss the boundary conditions. We consider the accumulation set $S_0 = \{ x \in \mathbb{R}^N \, : \, \exists \, m_k \rightarrow \infty, \, x^k \in S^{m_k}_0, \, x^k \rightarrow x\}$. If $x \in S_0$, then obviously $g^{m_k}(x^k) = p^0$, which implies that $g(x) = p^0$ by the uniform convergence. If now $x \in \mathbb{R}^N$ is fixed and such that $g(x) = p^0$, we can find for all $m \in \mathbb{N}$ a number $\lambda_m$ such that $g^m(x+\lambda_m \, 1^N) = p^0$. To show this, we remark that $g_m(x+ \lambda \, 1^N) - g^m(x) \geq  \lambda \, \nabla_xg^m(x) \cdot 1^N = \lambda \, |\nabla_xg^m(x)|_1$ by the convexity of $g^m$. Thus $g_m(x+ \lambda \, 1^N) \rightarrow +\infty$ for $\lambda \rightarrow +\infty$. Moreover, $g_m(x+ \lambda \, 1^N) \leq g^m(x) + \lambda \, |\nabla_xg^m(x+\lambda \, 1^N)|_1$. Hence $|\nabla_xg^m(x+\lambda \, 1^N)|_1 \rightarrow 0$ for $\lambda \rightarrow - \infty$, since otherwise $p_{\inf}^m = \inf g^m = -\infty$, a contradiction to the choice of finite $p_{\inf}^m$. Now we use the identity $V^m(g^m(x+\lambda \, 1^N), \, \nabla_x g^m(x+\lambda \, 1^N)) = 1$, which implies that
\begin{align*}
 V^m(g^m(x+\lambda \, 1^N), \, y^m_{\lambda}) = \frac{1}{ |\nabla_xg^m(x+\lambda \, 1^N)|_1}, \, \quad y^m_{\lambda} := \frac{\nabla_x g^m(x+\lambda \, 1^N)}{ |\nabla_xg^m(x+\lambda \, 1^N)|_1} \, .
\end{align*}
Passing to the limit $\lambda \rightarrow - \infty$, it follows that $\lim_{\lambda \rightarrow - \infty} V^m(g^m(x+\lambda \, 1^N), \, y^m_{\lambda}) = +\infty$, which is impossible if $g^m(x+\lambda \, 1^N)$ does not tend to $p_{\inf}^m$. Thus, $\lambda \mapsto g^m(x+ \lambda \, 1^N)$ attains every value in $]p_{\inf}^m, \, + \infty[$ and we must have $p^0 = g^m(x+ \lambda_{m} \, 1^N)$ for some $\lambda_m$. Since
\begin{align*}
|p^0 - g^m(x)| =  |\lambda_m| \, \int_{0}^1 \nabla_xg^m(x + \theta\, \lambda_m \, 1^N) \cdot 1^N \, d\theta \geq \frac{|\lambda_m|}{\sup_{p = \theta \, g^m(x) + (1-\theta) \, p^0, \, y \in S^1_+} V^m(p, \, y)} \, ,
\end{align*}
we see that $\lambda_m \rightarrow 0$. Thus $x+\lambda_m \, 1^N \in S_0^m$ converges to $x$, which was chosen an arbitrary point of the level-set $g(x) = p^0$. It follows that $S_0 = \{x \, : \, g(x) = p^0\}$.

\eqref{claimimage}. Next we characterize the image of $\nabla g$. If $ \rho \in \nabla g(\mathbb{R}^N)$, then $\rho = \nabla g(x)$ for some $x$. Since $\{x\}$ is trivially a compact set, we prove as above that $\rho \in \mathbb{R}^N_+$, and for $p = g(x)$ we have $V(p, \, \rho) = 1$. Thus $\nabla g(\mathbb{R}^N) \subseteq \bigcup_{p> p_{\inf}} \{ \rho \in \mathbb{R}^N_+ \, : \, V(p, \, \rho) = 1\}$.

We show the reverse inclusion. Since $\nabla g^m(\mathbb{R}^N) = \mathbb{R}^N_+$, the image of $\nabla g^m$ contains all vectors of the form $y/V^m(p, \, y)$ where $y \in S^1_+$ and $p \in ]p_{\inf}, \, + \infty[$ are arbitrary. In fact one can show as above that
\begin{align}\label{machin}
\frac{ y}{ V^m(p, \, y)} = \nabla g^m(x^m)\, \text{ for } x^m = \nabla_{\rho}\bar{k}^m(y) + \bar{V}^m_{\rho}(p,\, y) = \nabla_{\rho} f^m(y / V^m(p, \, y) ) \, .
\end{align}
For $i = 1,\ldots,N$, the convexity of $\bar{k}^m$ yields $t \, \bar{k}^m_{\rho_i}(y)  \leq \bar{k}^m(y +t \, e^i) - \bar{k}^m(y)$ for all $|t|$ smaller than half of the distance of $y$ to $\partial \mathbb{R}^N_+$. Since $\bar{k}^m$ converges pointwise, we see that $\{|\nabla_{\rho} \bar{k}^m(y)|\}$ must be bounded. Due to the assumption \eqref{Convex}, the Hessian of $\rho \mapsto \bar{k}^m(\rho) + \bar{V}^m(p,\, \rho)$ is positive semi-definite. Thus, this function is convex, and the same argument as just seen for bounding $|\nabla_{\rho}\bar{k}^m(y)|$ implies that $|\nabla_{\rho} \bar{k}^m(y) + \bar{V}^m_{\rho}(p,\, y)|$ is also bounded. By this argument, also the sequence $\{x^m\}$ in \eqref{machin} is bounded, hence there is some accumulation point $x = \lim x_m$. Passing in \eqref{machin} to the limit, we find that $\frac{y}{ V(p, \, y)} = \nabla g(x)$. Since $y \in S^1_+$ and $p\in]p_{\inf}, \, \infty[$ were arbitrary, $\nabla g(\mathbb{R}^N)$ contains all vectors $\rho \in \mathbb{R}^N_+$ such that $V(p, \, \rho) = 1$. This proves that
\begin{align*}
 \nabla g(\mathbb{R}^N) = \bigcup_{p > p_{\inf}} \{ \rho \in \mathbb{R}^N_+ \, : \, V(p, \, \rho) = 1\}\, .
\end{align*}

\eqref{claimregu}. We define $f = g^*$. Since $g$ is smooth, $f$ is essentially strictly convex (cf. \cite{rockafellar}, Th. 26.3). Moreover, $f$ is finite and subdifferentiable at every $\rho \in \nabla g(\mathbb{R}^N)$. A point $x$ belongs to $\partial f(\rho)$ if and only if $\rho = \nabla g(x)$ and $g(x) = -f(\rho) + x\cdot \rho$ (\cite{rockafellar}, Th. 23.5). By these means, we show that
\begin{align*}
V(-f(\rho) + \rho \cdot x, \, \rho) = 1 \text{ for all }  \rho \in \mathbb{R}^N_+, \, x \in \partial f(\rho) \, .
\end{align*}
Note that in a point $\rho \in \mathbb{R}^N_+$ such that $\rho \not\in \nabla g(\mathbb{R}^N)$, we have $\partial f(\rho) = \emptyset$ and the latter condition holds vacuously.
If $-f(\rho) + \rho \cdot x = p^0$ for some $\rho \in \mathbb{R}^N_+, \, x \in \partial f(\rho)$, then $\rho = \nabla g(x)$ and $g(x) = p^0$. which means that necessarily $x \in S_0$. Thus, we can approximate $x$ with points $x^m  \in S^m_0$. For $\rho^m := \nabla g^m(x^m)$, we have by construction $k^m(\rho^m) = f^m(\rho^m) = \rho^m \cdot x^m - p^0$. Thus, as $m \rightarrow \infty$ we find that $k(\nabla g(x)) = \lim_{m \rightarrow \infty} \rho^m \cdot x^m - p^0 = \nabla g(x) \cdot x - p^0 = f(\nabla g(x))$, where $k$ is the pointwise and uniform limit of the convex functions $k^m$. This achieves to prove that $f$ is a weak solution to the bvp \eqref{probPbulkneu}, \eqref{probPboundaneu}.
\eqref{claimepi}. Let us show that the function $\tilde{f}(\rho) := \inf\{\liminf_{m \rightarrow \infty} f^m(\rho^m) \, : \, \rho^m \rightarrow \rho\}$ is nothing else but $f$.
At first, the definition of convex conjugates yields $f^m(\rho^m) \geq \rho^m \cdot x - g^m(x)$ for all $x \in \mathbb{R}^N$, and therefore $\liminf_{m\rightarrow \infty} f^m(\rho^m) \geq g^*(\rho) = f(\rho)$ for arbitrary sequences $\rho^m \rightarrow \rho$. This shows that $\tilde{f} \geq f$. On the other hand, for every fixed $x \in \mathbb{R}^N$ and $\rho^m = \nabla g^m(x)$, the definition of the subdifferential of $g^m$ yields $f^m(\rho^m) = x \cdot \nabla g^m(x) - g^m(x) \rightarrow  x \cdot \nabla g(x) - g(x) = f(\nabla g(x))$. Thus for all $\rho \in \nabla g(\mathbb{R}^N)$, we obtain that $\inf\{\liminf_{m \rightarrow \infty} f^m(\rho^m) \, : \, \rho^m \rightarrow \rho\} \leq f(\rho)$. For $\rho \not\in  \nabla g(\mathbb{R}^N)$, we have $f(\rho) = + \infty$ and $\tilde{f} \leq f$ is obviously true. Thus $\tilde{f} = f$.
\end{proof}

\section{Incompressible limits}\label{incompressible}
We now consider the case of a small isothermal compressibility, which we mathematically express by choosing a sequence of functions $\{V^m\}_{\in \mathbb{N}}$ such that $\partial_p V^m \rightarrow 0$ as $m \rightarrow \infty$. In which precise meaning we assume this convergence is stated immediately below.

For a given positively homogeneous convex function $k^m$, we solve the problem
\begin{align}\label{GDm}
 V^m(-f^m + \rho \cdot \nabla f^m, \, \rho) = & 1 \text{ in } \mathbb{R}^N_+\, , \quad  f^m =  k^m  \text{ for } V^m(p^0, \, \rho) = 1 \, ,
\end{align}
and shall exhibit the mathematical consequences of the assumptions that 1) the problems \eqref{GDm} possesses a classical solution $f^m$ of Legendre--type in the sense of Definition \ref{DEF} and 2) $\partial_p V^m \rightarrow 0$.

\subsection{Proof of Proposition \ref{MAIN}}

Under the assumption \eqref{belle0}, we can readily verify that the limit function $V^{\infty}$ is globally independent of $p$ and belongs to $C^{\infty,2}(\mathbb{R} \times \mathbb{R}^N_+) \cap C^{\infty,0}(\mathbb{R} \times \overline{\mathbb{R}}^N_+)$.

Adopting the assumptions of Proposition \ref{MAIN}, we find unique classical co-finite solutions of Legendre--type to \eqref{GDm}. We go for applying Lemma \ref{asymptotics} and want to check the assumptions \eqref{compactum}, \eqref{compactmore} and \eqref{Convex}.

In order to first verify \eqref{compactum}, we recall the definition of $\bar{k}^m(\cdot) = k^m(\cdot) + p^0 \, V^m(p^0,\cdot)$ in order to bound
\begin{align}\label{eigenschaft0}
| \bar{k}^m(y) + \bar{V}^m(p, \, y)|  & \geq |\bar{V}^m(p, \, y)| - |k^m(y)| - |p^0| \, |V^m(p^0, \, y)| \\
 & \geq \int_{p^0}^p V^m(t, \, y) \, dt - |k^m(y)| -|p^0| \, |V^m(p^0, \, y)|  \, .\nonumber
\end{align}
Therefore, exploiting the uniform convergence \eqref{belle1} and \eqref{belle2},
\begin{align}\label{eigenschaft}
& \liminf_{m \rightarrow \infty} \inf_{y \in S^1_+} |\bar{k}^m(y) + \bar{V}^m(p, \, y)| \nonumber\\
& \quad \geq  \liminf_{m \rightarrow \infty} \left|\int_{p^0}^p \inf_{y \in S^1_+} V^m(t, \, y) \, dt\right| - \limsup_{m \rightarrow \infty} \sup_{y \in S^1_+} (|k^m(y)| + |p^0| \, V^m(p^0, \, y))  \nonumber\\
& \quad =  |p-p^0| \,  \inf_{y \in S^1_+} V^{\infty}(y) - \sup_{y \in S^1_+} (|k(y)|  + |p^0| \, V^{\infty}(p^0, \, y)) \, .
\end{align}
Since we assume \eqref{belle1}, we obtain $\inf_{y \in S^1_+} V^{\infty}(y) = \min_{y \in \overline{S}^1_+} V^{\infty}(y) > 0$ because $V^{\infty}$ is assumed strictly positive. This verifies the assumption \eqref{compactum} of Lemma \ref{asymptotics}.

We next check \ref{asymptotics}, \eqref{Convex}, by proving a bound from below for $ \liminf_{m \rightarrow \infty} \lambda_{\min}(H^m(p, \, y))$ on compact subsets of $]-\infty, \, +\infty[ \times S^1_+$ with $p \in [p_1, \, p_2]$ and $\inf_{i=1,\ldots,N} y_i \geq y_0 >0$. Recall that
\begin{align*}
 H^m(p, \, y) = D^2\bar{k}^m(\frac{y}{V^m(p, \, y)}) + \bar{V}^m_{\rho,\rho}(p, \, \frac{y}{V^m(p, \, y)}) - \frac{V^m(p, \, y)}{V^m_p(p, \, y)} \, V^m_{\rho}(p, \, y) \otimes V^m_{\rho}(p, \, y) \, .
\end{align*}
Using the positive homogeneity of $\bar{V}^m$ in the second argument and the convergence \eqref{belle1}, \eqref{belle2}, it is readily verified for $(p, \, y) \in \mathbb{R} \times S^1_+$ that
\begin{align}\label{rhorho1}
 \bar{V}^m_{\rho,\rho}(p, \, \frac{y}{V^m(p, \, y)}) =  V^m(p, \, y) \,  \bar{V}^m_{\rho,\rho}(p, \, y)\rightarrow (p-p^0) \, V^{\infty}(y) \,    D^2V^{\infty}(y) \, .
\end{align}
Moreover, using that $k^m$ and $V^m(p^0,\cdot)$ are positively homogeneous, the convergence \eqref{belle2} yields
\begin{align}\label{rhorho2}
 D^2\bar{k}^m(\frac{y}{V^m(p, \, y)}) & =  D^2k^m(\frac{y}{V^m(p, \, y)}) + p^0 \, V^m_{\rho,\rho}(p^0, \, \frac{y}{V^m(p, \, y)}) \\
 = & V^m(p, \, y) \, (D^2k^m(y) + p^0 \, V^m_{\rho,\rho}(p^0, \, y)) \rightarrow V^{\infty}(y) \, (D^2k(y) + p^0 \, D^2V^{\infty}(y)) \, .\nonumber
\end{align}
Consider for $\eta \in \mathbb{R}^N$ arbitrary the vectors $\xi^m := \eta - (V^m_{\rho}(p, \, y) \cdot \eta/V^m(p, \, y)) \, y$. Since $V^m_{\rho} \rightarrow V^{\infty}_{\rho}$ uniformly of compact sets of $\mathbb{R} \times S^1_+$ (this follows easily from \eqref{belle1}, \eqref{belle2}), we see that
\begin{align*}
 \xi^m \rightarrow \xi := \eta -\frac{V^{\infty}_{\rho}(y) \cdot \eta}{V^{\infty}(y)} \, y  \, .
\end{align*}
Using that $V^m_{\rho}(p, \, y) \cdot \xi^m = 0$, we check by means of \eqref{rhorho1} and \eqref{rhorho2} that
\begin{align}\label{eigenschaft2}
& H^m(p, \, y) \xi^m \cdot \xi^m =  (D^2\bar{k}^m(\frac{y}{V^m(p, \, y)}) + \bar{V}^m_{\rho,\rho}(p, \, \frac{y}{V^m(p, \, y)})) \xi^m \cdot \xi^m\nonumber\\
&\quad \rightarrow  V^{\infty}(y) \, D^2(k + p \,  V^{\infty})(y) \xi \cdot \xi =  V^{\infty}(y) \, D^2(k + p \,  V^{\infty})(y) \eta \cdot \eta\, .
 \end{align}
If now $D^2V^{\infty}(y)$ would possess a negative eigenvalue, then we would find some $p \in ]0, \, +\infty[$ such that, in view of \eqref{eigenschaft2}, the matrix $H^m(p, \, y)$ already possesses a strictly negative eigenvalue for $m$ sufficiently large. But recall that $H^m(p, \, y) = D^2f^m(y/V^m(p, \, y))$, with $D^2f^m$ strictly positive definite on $\mathbb{R}^N_+$ by assumption. Hence, avoiding a contradiction is possible only if all eigenvalues of $D^2V^{\infty}(y)$ are nonnegative.

Next, for every positive eigenvalue $\ell$ of $D^2V^{\infty}(y)$, we can choose some eigenvector $\eta$ such that $D^2V^{\infty}(y) \eta \cdot \eta = \ell$. With the help of \eqref{eigenschaft2} again, it follows that $ H^m(p, \, y) \xi^m \cdot \xi^m < 0$ if $p$ is sufficiently small, namely if $\ell \, p < - \lambda_{\max}(D^2k^m(y))$ and $m$ is sufficiently large. Since $f^m$ is a $C^2$ strictly convex function for all $m$, all eigenvalues of $D^2V^{\infty}(y)$ can only be trivial, meaning that the limit function $V^{\infty}$ must be linear in $\rho$.

In other words, we have shown that there exists a vector $\bar{\upsilon} \in \mathbb{R}^N$ such that $V^{\infty}(\rho) = \bar{\upsilon} \cdot \rho$ on $\mathbb{R}^N_+$. In order to conclude that $\bar{\upsilon} \in \mathbb{R}^N_+$, it is sufficient to recall that, by assumption, $V^{\infty}(y) > 0$ for all $y \in \overline{S}^1_+$.

We now come to the direct verification of the assumption \eqref{Convex} of Lemma \ref{asymptotics}. For $\eta \in S^2_+$ arbitrary,
\begin{align}\label{darstellungbis}
 H^m(p, \, y) \eta \cdot \eta = V^m(p, \, y) \, (D^2\bar{k}^m(y) + \bar{V}^m_{\rho,\rho}(p, \, y)) \eta \cdot \eta - \frac{V^m(p, \, y)}{V^m_p(p, \, y)} \, |V^m_{\rho}(p, \, y) \cdot \eta|^2 \, .
\end{align}
Consider arbitrary $p_{\inf} < p_1 < p_2 < + \infty$ and $y_0 > 0$. Since $V^{\infty}$ is linear, $ D^2\bar{V}^{\infty}(y) = 0$. For $m\rightarrow +\infty$, the assumption \eqref{belle2} implies that
\begin{align}\label{tozero}
\sup_{p \in [p_1,p_2], \, \inf y \geq y_0} \lambda_{\max}(\bar{V}^m_{\rho,\rho}(p, \, y)) \leq \sup_{p \in [p_1, \, p_2], \, \inf  y \geq y_0} |\bar{V}^m_{\rho,\rho}(p, \, y) - D^2\bar{V}^{\infty}(y)| \rightarrow 0 \, .\end{align}
Similarly, $\sup_{p \in [p_1,p_2], \, \inf y \geq y_0} \lambda_{\max}(D^2V^m(p^0, \, y)) \rightarrow 0$. Using \eqref{belle2} again, we then verify that the matrix $D^2\bar{k}^m(y) + \bar{V}^m_{\rho,\rho}(p, \, y)$ possesses $N-1$ strictly positive eigenvalues for large $m$. In particular, this matrix is strictly positive definite on $\{y\}^{\perp}$. For fixed $y$, we now choose some orthonormal basis $\{\xi^1,\ldots,\xi^{N-1}, \, y/|y|\}$ of $\mathbb{R}^N$. For arbitrary $\eta \in \mathbb{R}^N$, it is then possible to find coefficients $a_1,\ldots,a_{N-1}$ and $b$ to represent $$ \eta = \sum_{i=1}^{N-1} a_i \, (\xi^i - \frac{V^m_{\rho}(p, \, y)\cdot \xi^i}{V^m(p, \, y)} \, y) + b \, y \, .$$
Due to \eqref{belle1}, \eqref{belle2}, $V^m_{\rho} \rightarrow \bar{\upsilon}$ and $V^m \rightarrow V^{\infty} > 0$ uniformly on compact subsets of $\mathbb{R} \times S^1_+$. Hence, for all $y \in S^1_+$ such that $\inf y \geq y_0 > 0$ and all $p \in [p_1, \, p_2]$, there is a constant independent on $m$ such that $|\eta| \leq c \, (|a|+|b|)$. Using \eqref{darstellungbis} it is readily seen with $\Pi = [\xi^1, \ldots,\xi^{N-1}]$ that\begin{align*}
 H^m(p, \, y) \eta \cdot \eta = & V^m(p, \, y) \, \Pi^{\sf T}(D^2\bar{k}^m(y) + \bar{V}^m_{\rho,\rho}(p, \, y)) \Pi \, a \cdot a - \frac{(V^m(p, \, y))^3}{V^m_p(p, \, y)} \, |b|^2 \\
 \geq & V^m(p, \, y) \, [\lambda_{\min}^+(D^2\bar{k}^{m}(y))-\lambda_{\max}(\bar{V}^m_{\rho,\rho}(p, \, y))] \, |a|^2 - \frac{(V^m(p, \, y))^3}{V^m_p(p, \, y)} \, |b|^2\\
 \geq & \frac{1}{c} \, V^m(p, \, y) \, \min\{[\lambda_{\min}^+(D^2\bar{k}^m(y))-\lambda_{\max}(\bar{V}^m_{\rho,\rho}(p, \, y)),\, - \frac{(V^m(p, \, y))^2}{V^m_p(p, \, y)}\} \, |\eta|^2 \, .
 \end{align*}
 Here $\lambda_{\min}^+$ denotes the smallest strictly positive eigenvalue of a symmetric, positive semi--definite matrix. Since for some $c_0 > 0$ independent on $m$, we have $\inf_{p \in [p_1,p_2], \, \inf y \geq y_0} V^m(p, \, y) \geq c_0$, we deduce that
\begin{align*}
& \inf_{p \in [p_1,p_2], \, \inf y \geq y_0} \lambda_{\min}(H^m(p, \, y)) \\
& \qquad \geq \tilde{c} \, \min_{p \in [p_1,p_2], \, \inf y \geq y_0}\Big\{\lambda_{\min}^+(D^2\bar{k}^m(y))-\lambda_{\max}(\bar{V}^m_{\rho,\rho}(p, \, y)),\, \frac{(V^m(p, \, y))^2}{V^m_p(p, \, y)}\Big\} \, .
\end{align*}
Making use of the fact that $D^2\bar{k}^m(y)$ is strictly positive on $\{y\}^{\perp}$ (cf.\ \eqref{belle2}) and of \eqref{tozero}, we infer that $H^m$ satisfies the requirement \eqref{Convex} of Lemma \ref{asymptotics}.

We also notice that the assumptions \eqref{belle1} and \eqref{belle2} imply that $V^m_{\rho} \rightarrow V^{\infty}_{\rho} = \bar{\upsilon}$ uniformly on compact subsets of $\mathbb{R} \times \mathbb{R}^N_+$. Thus, for all $-\infty < p_1 \leq p \leq p_2 < + \infty$ and $y \in S^1_+$, we see that
\begin{align*}
\limsup_{m \rightarrow \infty} | \bar{V}^m_{\rho}(p, \, y)| \leq (|p_1|+|p_2|) \, |V^{\infty}_{\rho}(y)| = (|p_1|+|p_2|) \, |\bar{\upsilon}| \, .
\end{align*}
It therefore also follows that
\begin{align}\label{eigenschaft3}
 \liminf_{m \rightarrow \infty} \inf_{p \in ]p_1, \, p_2[ } |\nabla_{\rho}\bar{k}^m(y) + \bar{V}^m_{\rho}(p, \, y) | \geq |\nabla_{\rho}k(y)| - |p^0| \,|\bar{\upsilon}| -  (|p_1|+|p_2|) \, |\bar{\upsilon}| \, ,
\end{align}
and \ref{asymptotics}, \eqref{compactmore} follows because we assume in \eqref{belle2} that the limit $k$ is essentially smooth on $S^1_+$.

The assumptions of Lemma \ref{asymptotics} are verified, which ends the proof of the main Proposition \ref{MAIN}.

\begin{proof}[Proof of Proposition \ref{MAIN}, main argument]
Since the assumptions of Lemma \ref{asymptotics} are verified, we find a strictly convex co-finite weak solution $f$ to the problem \eqref{probPbulkneu}, \eqref{probPboundaneu} for $V = V^{\infty}(\rho)$. The convex conjugate $g := f^*$ belongs to $C^2(\mathbb{R}^N)$ and is a convex solution to $\bar{\upsilon} \cdot \nabla g = 1$ in $\mathbb{R}^N$. Moreover, due to \ref{asymptotics}, \eqref{claimimage}, the image $\nabla g(\mathbb{R}^N)$ is nothing else than $S = \{\rho \, : \, \bar{\upsilon} \cdot \rho = 1\}$. Thus $f$ is finite only on states satisfying $\bar{\upsilon} \cdot \rho = 1$. Due to \eqref{probPboundaneu}, $f = k$ in this set. 

In order that $\mu \in \partial f^{\infty}(\rho)$, we must have $\nabla g(\mu) = \rho$. This means that $\rho \in \nabla g(\mathbb{R}^N)$, hence $\bar{\upsilon} \cdot \rho = 1$ is necessary. Under this condition, the set $\partial f^{\infty}(\rho)$ is not empty, and it consists of all $\mu \in \mathbb{R}^N$ such that $f^{\infty}(r) \geq k(\rho) + \mu \cdot(r-\rho)$ for all $r \in \mathbb{R}^N_+$. Thus, restricting to $r \in S$, we see that $\mu \in \partial f^{\infty}(\rho)$ implies that $k(r) \geq k(\rho) + \mu \cdot(r-\rho)$ for all $r \in S$. Since $S$ is planar, the tangential part of $\mu$ can only be the tangential gradient of $k$ at $\rho$, and there must exist a $p \in \mathbb{R}$ such that $\mu = \nabla_{\rho} k(\rho) + p \, \bar{\upsilon}$.

Conversely, if there is some $\mu^0$ in $\partial f^{\infty}(\rho)$, then clearly $\mu^0 + p \, \bar{\upsilon} \in \partial f^{\infty}(\rho)$ for all $p \in \mathbb{R}$. This is due to the fact that the $g$ is affine with slope one in the direction of the vector $\bar{\upsilon}$, which belongs to the kernel of the Hessian $D^2g$.
This achieves to show the characterization of $\partial f^{\infty}$, hence the claim.
\end{proof}

\begin{proof}[Proof of Lemma \ref{MUELLER}]

The condition $\partial^2_{T,T} f^m <0$ means that
\begin{align}\label{FEconcTunhere}
0 >  & \partial^2_{T,T}k^m(T, \, \rho) + p^0 \, \partial^2_{T,T} V^m(T, \, p^0, \, \rho) + \int_{p^0}^{p^m(T,\rho)}  \partial^2_{T,T}V^m(T, \, p^{\prime}, \, \rho) \, dp^{\prime} \nonumber\\
& - \frac{(\partial_TV^m(T, \, p(T,\rho), \, \rho))^2}{\partial_pV^m(T, \, p^m(T,\rho), \, \rho)} \, .
\end{align}
Now we fix $(T, \, p, \, x)$ and choose $\rho_i^m = \hat{\rho}^M_i(T, \, p, \, x)$ which converges to $M_i \, x_i/ \hat{\upsilon}^{\infty}(T, \, p^0, \, x) = \bar{\rho}$. Owing to the convergence assumed in the statement of Lemma \ref{MUELLER}, it follows that
\begin{align*}
 0 \geq & \partial^2_{T,T}k(T, \, \bar{\rho}) + p^0 \, \partial^2_{T,T} V^{\infty}(T, \, p^0, \, \bar{\rho}) + \int_{p^0}^{p}  \partial^2_{T,T}V^{\infty}(T, \, p^{\prime}, \, \bar{\rho}) \, dp^{\prime} \nonumber\\
& + \liminf_{m\rightarrow \infty}  \frac{(\partial_TV^m(T, \, p, \, \rho^m))^2}{|\partial_pV^m(T, \, p, \, \rho^m)|} \, .
\end{align*}
Using the definition of $V^m$, we then see that
\begin{align*}
\liminf_{m\rightarrow \infty}  \frac{(\partial_TV^m(T, \, p, \, \rho^m))^2}{|\partial_pV^m(T, \, p, \, \rho^m)|} = \liminf_{m \rightarrow \infty} \frac{1}{\hat{\upsilon}^m(T, \, p, \, x)} \,  \frac{(\partial_T\hat{\upsilon}^m(T,p,x))^2}{|\hat{\upsilon}_p^m(T, \, p, \,x)|} \, ,
\end{align*}
and the latter quantity is positive infinite unless $\partial_T\hat{\upsilon}^m(T,p,x) \rightarrow 0$.
This completes the proof of Lemma \ref{MUELLER}.
\end{proof}

\subsection{Proof of Proposition \ref{MAINNN}}\label{localargu}

In view of the assumption \eqref{belle0pr}, the limit function $V^{\infty}(T) = V^{\infty}(T, \, p, \, \rho)$ is independent of $p$ in the set of all $$(T, \, p, \, \rho) \quad\text{s.\ t.\ } \quad T \in ]T_{\inf}, \, T_{\sup}[, \, \rho \in \mathbb{R}^{N}_+, \, p \in ]a(T,x(\rho)), \, b(T,x(\rho)[ \, ,$$
but it continues to depend on $p$ elsewhere.

Next we proceed to checking the assumptions of Lemma \ref{asymptotics} for this new case. In checking the assumption \eqref{compactum}, we argue like in \eqref{eigenschaft0}, \eqref{eigenschaft} and we now find that
\begin{align*}
& \liminf_{m \rightarrow \infty} \inf_{y \in S^1_+} |\bar{k}^m(y) + \bar{V}^m(p, \, y)|\\
& \quad \geq \liminf_{m \rightarrow \infty}  \inf_{y \in S^1_+} \left| \int_{p^0}^pV^m(t, \, y) \, dt \right| - \limsup_{m \rightarrow \infty} \sup_{y \in S^1_+} (|k^m(y)|+p^0 \, |V^m(p^0, \, y)|)  \nonumber\\
& \quad = \inf_{y \in S^1_+} \left| \int_{p^0}^{p}   V^{\infty}(t, \, y) \, dt \right| - \sup_{y \in S^1_+} (|k(y)|+p^0 \, |V^{\infty}(p^0, \, y)|) \, .
\end{align*}
Next we observe that
\begin{align*}
\int_{p^0}^{p} V^{\infty}(t, \, y) \, dt = \begin{cases}
-\int_p^{a} V^{\infty}(t, \, y) \, dt - (p^0-a) \,  V^{\infty}(p^0, \, y) & \text{ for } p < a \, ,\\
(p-p^0) \, V^{\infty}(p^0, \, y) & \text{ for } a \leq p \leq b \, ,\\
\int_b^{p} V^{\infty}(t, \, y) \, dt + (b-p^0) \,  V^{\infty}(p^0, \, y) & \text{ for } b < p \, .
\end{cases}
\end{align*}
Making use of \eqref{bien5strich}, property \eqref{compactum} of Lemma \ref{asymptotics} is obvious.

We can check the requirement of Lemma \ref{asymptotics}, \eqref{Convex} with the same arguments as in the preceding section. The only difference is that we must not require that the limit function $V^{\infty}$ is necessarily linear. Recall that for all $a\leq p  \leq b$, the matrix $D^2V^{\infty}(p, \, y) = D^2V^{\infty}(p^0, \, y)$ is independent on $p$. The argument following \eqref{eigenschaft2} implies that, for $p \in [a, \, b]$,
\begin{align*}
V^{\infty}(y) \, D^2(k + p \,  V^{\infty})(y) \eta \cdot \eta \geq 0 \quad \text{ for all } \eta \in \mathbb{R}^n, \, y \in S^1_+ \, .
\end{align*}
Hence, in the local case, the requirement of convexity implies that the thresholds $a, \, b$ are restricted by the condition
\begin{align*}
D^2(k + p \,  V^{\infty})(y) \quad \text{ positive definite for all } y \in S^1_+ \text{ and all } p \in [a(T,y), \, b(T,y)] \, .
\end{align*}

Finally, we use \eqref{bien6strich} and we find that $$ \liminf_{m \rightarrow \infty}  \inf_{p \in ]p_1, \, p_2[ } |\nabla_{\rho}\bar{k}^m(y) + \bar{V}^m_{\rho}(p, \, y) | \geq |\nabla_{\rho}k(y)| - C \, (|p^0| + |p_1| + |p_2|) \, ,$$ and the condition \eqref{compactmore} of Lemma \ref{asymptotics} is valid whenever $k$ is essentially smooth on $S^1_+$. Thus, by means of the Lemma \ref{asymptotics}, we can establish the following claims.
\begin{lemma}\label{redonde}
Under the assumptions of Proposition \ref{MAINNN}, the function $\rho \mapsto V^{\infty}(p^0, \, \rho)$ is convex on $\mathbb{R}^N_+$. For all $\rho \in \mathbb{R}^N_+$ such that $V^{\infty}(p^0, \, \rho) < 1$, there is a unique $p(\rho) \in ]p_{\inf}, \, a[$ such that $V^{\infty}(p(\rho), \, \rho) = 1$ and the map $\rho \mapsto p(\rho)$ belongs to $C^1(\{\rho \, : \, V^{\infty}(p^0, \, \rho) < 1\})$. Similarly for all $\rho \in \mathbb{R}^N_+$ such that $V^{\infty}(p^0, \, \rho) > 1$, there is a unique $p(\rho) \in ]b, \, +\infty[$ such that $V^{\infty}(p(\rho), \, \rho) = 1$ and the map $\rho \mapsto p(\rho)$ belongs to $C^1(\{\rho \, : \, V^{\infty}(p^0, \, \rho) > 1\})$. The function
\begin{align}
& f^{\infty}(T,\rho) := \nonumber \\
  & \begin{cases}
                         k(T,\rho) + p^0 \, V^{\infty}(T,p^0,\, \rho) +\bar{V}(T,p(T,\rho), \,\rho) - p(T,\rho) & \text{ for } \rho  \text{ s.\ t.\ } V^{\infty}(T,p^0,\, \rho) \neq 1 \, ,\\
                         k(T,\rho) & \text{ for } \rho  \text{ s.\ t.\ } V^{\infty}(T,p^0,\, \rho) = 1 \, ,
                         \end{cases} \nonumber
\end{align}
is a co-finite strictly convex weak solution to the problem \eqref{probPbulkneu}, \eqref{probPboundaneu} for $V = V^{\infty}$. We have $f^m \rightarrow f^{\infty}$ (epi-convergence). A point $\mu \in \mathbb{R}^N$ belongs to $\partial f^{\infty}(\rho)$ if and only if one of the two following conditions is valid
\begin{itemize}
\item[\text{Either}] $V^{\infty}(p^0, \, \rho) \neq 1$ and $\mu = p^0 \, V^{\infty}_{\rho}(p^0,\, \rho) + \nabla_{\rho}k(\rho) + \bar{V}_{\rho}(p(\rho), \,\rho)$;
\item[\text{Or}] $V^{\infty}(p^0, \, \rho) = 1$ and $\mu = p \, V^{\infty}_{\rho}(p^0, \, \rho) + \nabla_{\rho}k(\rho)$ for a $a \leq p \leq b$.
\end{itemize}
\end{lemma}
Note that, since the threshold functions $a$ and $b$ are finite, the convergence of Lemma \ref{asymptotics} can be improved to uniform convergence. We start from the representation
\begin{align}\label{lesfem}
 f^{m}(\rho) :=  k^m(\rho) + p^0 \, V^{m}(p^0,\, \rho) +\bar{V}^m(p^m(\rho), \,\rho) - p^m(\rho) \quad \text{ for } \rho \in \mathbb{R}^N_+ \, .
\end{align}
We show that for all $\rho$ such that $V^{\infty}(p^0, \, \rho) < 1$ the pressures converge, i.e.\ $p^m(\rho) \rightarrow p(\rho)$, where $p(\rho) < a$. Indeed, if $V^{\infty}(p^0, \, \rho) < 1$, then $V^{m}(p^0, \, \rho) < 1$ for $m$ large. Thus the number $p^m(\rho)$ belongs to the interval $]p_{\inf}^m, \, p^0[$, implying that the sequence $p^m(\rho)$ is bounded. Let $\pi$ be any accumulation point. As $V^m$ converges uniformly, we get $1 = V^{\infty}(\pi, \, \rho)$. If $\pi \geq a(T, \, x(\rho))$, then $V^{\infty}(\pi, \, \rho) \leq V^{\infty}(p^0, \, \rho) < 1$. Thus $\pi < a(T, \, x(\rho))$. But in the interval $]p_{\inf}, \, a[$, $V^{\infty}$ is strictly decreasing in the first variable. Hence $\pi = p(\rho)$ is the unique solution in $]p_{\inf}, \, a[$ to $V^{\infty}(\pi, \, \rho) = 1$.

Similarly, for all $\rho$ subject to $V^{\infty}(p^0, \, \rho) > 1$, we show that $p^m(\rho) \rightarrow p(\rho)$ with $p(\rho) > b$ being the unique solution to $V^{\infty}(\pi,\rho) = 1$. To show that the sequence $p^m(\rho)$ is bounded, we employ the following argument: Suppose that $\limsup p^m(\rho) =+ \infty$. Then for $K$ sufficiently large, we find $m_0$ appropriate such that $V^{m}(K, \, \rho) \geq 1$ for $m \geq m_0$. Passing to the limit, $V^{\infty}(K, \, \rho) \geq 1$, and for $K \rightarrow \infty$ we must violate the assumptions \eqref{bien5strich}.

Hence the sequence $\{f_m(\rho)\}$ defined in \eqref{lesfem} is bounded for every $\rho \in \mathbb{R}^N_+$. A classical result of convex analysis \cite{rockafellar} shows that $f_m$ converges uniformly on compact subsets of its domain to its pointwise limit.  This pointwise limit being finite, it must be identical with the epi--limit $f^{\infty}$ of Lemma \ref{redonde}.

\newcommand{\etalchar}[1]{$^{#1}$}

\appendix

\section{Pieces of thermodynamics}\label{pieces}

\vspace{0.2cm} {\bf Variables.} We consider a fluid mixture consisting of $N$ chemical species $\ce{A}_1,\ldots,\ce{A}_N$. Locally it is characterized by the quantities
\begin{equation}\label{Athermo0}
    T\quad\text{-~(absolute) temperature}\qquad\text{and}\qquad (n_i)_{i=1,2,\ldots,N}\quad\text{-~mole densities}~.
\end{equation}
These $N+1$ quantities constitute the set of basic variables.\footnote{ Following the viewpoint adopted so far, the main variables are $T$ and the partial mass densities $\rho_1, \, \rho_2, \ldots, \rho_N$.

With the simple connection $n_i = \rho_i/M_i$ where $M_i > 0$ is the constant molecular mass, we can associate with every thermodynamic function of the main variables $f(T,\rho_1,\ldots,\rho_N)$ a transformed function of the variables in \eqref{Athermo0} via $\tilde{f}(T, n_1, \ldots,n_N) := f(T, m_1 \, n_1 ,\ldots,m_N \, n_N) \, .$
Working directly with the transformed function $\tilde{f}$ significantly simplifies the calculations resulting from the change of variables. From now and for the remainder of the appendix, we use $T$ and $n_1, \ldots,n_N$ as main variables and we skip the $\tilde{}$ or {\rm mol} superscripts.} In the context of experimental data and microscopic modeling other but equivalent sets of variables will occur.


Total mass density, total mole density and the molar volume of the mixture are defined as in \eqref{thermo1}.
%
%
%
We recall that the mole/number fraction of constituent $\ce{A}_i$ is given by $x_i:= n_i/n=\upsilon n_i$ (cf.\ \eqref{thermo8}).
The \textit{(Helmholtz) free energy density} has the form
\begin{equation}\label{Athermo1}
    \varrho\psi=\varrho\psi(T,n_1,n_2, \ldots ,n_N) \, .
\end{equation}
Comparing with \eqref{thermo8a}, some small readjustments are necessary. We now introduce the (molar--based) chemical potentials via
\begin{equation}\label{Athermo2}
    \mu_i = \mu_i^{\rm mol} :=\frac{\partial \varrho\psi}{\partial~ n_i}\quad\text{for}\quad i=1,2 \ldots,N.
\end{equation}
While the definitions \eqref{thermo8c}$_2$ and \eqref{thermo8c}$_3$ of internal energy and entropy are unaffected, the Gibbs--Duhem equation \eqref{thermo8c}$_1$ assumes the form
\begin{equation}\label{Athermo3}
    p=-\varrho\psi+\sum_{i=1}^N n_i \, \mu_i \, ,
\end{equation}
and the positivity requirement of \eqref{thermo8d}$_1$ is equivalent with
\begin{equation}\label{Athermo4}
   \left\{\frac{\partial \mu_i}{\partial n_j}\right\}_{i,j=1,\ldots,N} \quad \text{is symmetric, positive definite} \, .
\end{equation}
%

\vspace{0.2cm} {\bf Objectives and motivations.} The knowledge on the constitutive functions results either from experiments or from atomistic models. The exploitation of atomistic models relies on the laws of \textit{statistical mechanics}. If statistical mechanics is fully exploited, it directly yields the free energy density as a function of the basic variables in \eqref{Athermo0}.

Then the other constitutive functions are calculated  by means of the thermodynamic relations \eqref{thermo8b} (resp.\ \eqref{Athermo3}) and \eqref{thermo8c}. However, frequently one meets the situation that atomistic models are not fully exploited. Rather, they solely provide relations between some of the derivatives of the free energy function. Usually this restricted knowledge is supplemented by experimental data on other derivatives to achieve the full picture. For example, pressure, specific volume, specific heats and chemical potentials may be constructed from experimental data essentially resulting from different sources. In this case we are confronted with the problem of consistency of all these data.

The main problem we want to solve next answers the question: Which derivatives of the free energy function can be given consistently and which properties must they have in order to construct the free energy function in terms of these derivatives?

Different representation theorems are possible, depending on the provided data. An important aspect concerns the employed variables since, in experiments, the basic variables are not under control.

\vspace{0.2cm} {\bf Change of variables.} Most thermodynamic experiments control at least temperature and pressure or, alternatively, temperature and the specific volume. Particularly in chemical experiments the pressure is controlled and should be included in the list of variables. In the following we treat both cases and to this end we change the variables according to (cf.\ \eqref{thermo10})
\begin{equation}\label{Athermo10}
    (T,n_1,\ldots,n_N)\quad\leftrightarrow\quad(T,\upsilon,x_1, \ldots,x_{N-1})\quad\leftrightarrow\quad(T,p,x_1,\ldots,x_{N-1})~.
\end{equation}
The variable transformation \eqref{Athermo10}$_1$ is simple. It uses $n_i=\upsilon^{-1}x_i$ for $i=1,\ldots,N$ and $x_N=1-\sum_{j=1}^{N-1} x_j$. The change of variables \eqref{Athermo10}$_2$ needs the \textit{thermal equation of state} (see \eqref{thermo10a})
\begin{equation*}
    p=\bar p(T,\upsilon,x)\quad\leftrightarrow\quad \upsilon=\hat\upsilon(T,p,x)~.
\end{equation*}

Next we use a generic function $f(T,n_1,\ldots,n_N)$ to define two new functions,
\begin{equation}\label{Athermo11}
    \bar f(T,\upsilon,x):= f\Big(T,\frac{1}{\upsilon} \, x\Big)\qquad\textrm{and} \qquad \hat f(T,p,x):= \bar{f}(T,\upsilon,x)_{|\upsilon=\hat{\upsilon}(T,p,x)}~.
\end{equation}
The derivatives of the new functions $\bar f$ and $\hat f$ can easily be calculated. We have for \eqref{Athermo11}$_1$
\begin{equation}\label{Athermo12}
\partial_T \bar f  = \partial_T f
,\qquad  \partial_{\upsilon} \bar f =-\frac{1}{\upsilon^2}\sum_{j=1}^N \partial_{n_j} f \, x_j \, ,\qquad \partial_{x_i} \bar f =\frac{1}{\upsilon}\left(\partial_{n_i} f - \partial_{n_N} f \right)~,
\end{equation}
and for \eqref{Athermo11}$_2$ we obtain
\begin{equation}\label{Athermo12a}
\partial_T \hat f = \partial_T \bar f +\partial_{\upsilon} \bar f\, \partial_T \hat \upsilon \,
,\quad \partial_p \hat f =\partial_{\upsilon} \bar f \, \partial_p \hat \upsilon \, ,\quad\partial_{x_i} \hat f = \partial_{x_i} \bar f+\partial_{\upsilon} \bar f \,
\partial_{x_i} \hat \upsilon \, .
\end{equation}

\vspace{0.2cm} {\bf Consequences for the variables $(T,\upsilon,x)$.} We apply the transformation identities \eqref{Athermo12} to the thermodynamic equations \eqref{Athermo1}--\eqref{Athermo3} (see also \eqref{thermo8a}--\eqref{thermo8c}). The simple proofs of the following results are left to the reader.

In the new variables, the mass density obeys $\bar\varrho(\upsilon,x):=\upsilon^{-1}M(x)$. For the specific free energy function we now have $\bar \psi(T,\upsilon,x):=\psi(T,\upsilon^{-1}x)$.
Then, we obtain
\begin{equation}\label{Athermo13}
\partial_T \bar \psi =-\bar s(T,\upsilon,x),\quad\partial_{\upsilon} \bar \psi =-\frac{1}{M(x)}\bar p(T,\upsilon,x),\quad \partial_{x_i} (\bar{\varrho} \bar \psi) =\frac{1}{\upsilon}\left(\bar\mu_i(T,\upsilon,x)-\bar\mu_N(T,\upsilon,x)\right)~.
\end{equation}
The specific internal energy transforms as
$\bar u(T,\upsilon,x):=u(T,\upsilon^{-1}x)$. This representation of the specific internal energy with respect to the variables $(T,\upsilon,x)$ is known as the \textit{caloric equation of state}. Its derivative
\begin{equation}\label{Athermo14}
    \bar c_\upsilon:= \partial_T \bar u(T,\upsilon,x)
\end{equation}
defines the specific heat (at constant volume). Two further derivatives of the caloric equation of state read
\begin{equation}\label{thermo15}
    \partial_{\upsilon} \bar u =T \, \partial_T\bar p(T,\upsilon,x)-\bar p(T,\upsilon,x)\qquad\textrm{implying}\qquad    \partial_{\upsilon} \bar{c}_\upsilon =T\, \partial^2_{T,T}\bar p(T,\upsilon,x)~.
\end{equation}
We conclude that the volume dependence of both the specific internal energy and the specific heat is already given by the thermal equation of state. Thus their measuring is not needed!

The specific entropy transforms as $\bar s(T,\upsilon,x):= s(T,\upsilon^{-1}x)$ and satisfies
\begin{equation}\label{Athermo16}
\partial_T \bar s =\frac{\bar{c}_\upsilon}{T},\quad \partial_{\upsilon} \bar s =\frac{1}{M(x)} \partial_T\bar p\, ,\quad
\bar s=-\frac{\upsilon}{M(x)} \, \partial_T\bar p+\frac{1}{M(x)} \, \sum_{j=1}^N
\partial_T \bar\mu_j \, x_j~.
\end{equation}
In analogy to the caloric equation of state we find that the volume dependence of the specific entropy is likewise given by the thermal equation of state.

In the new variables, the chemical potentials are given by
$\bar \mu_i(T,\upsilon,x):= \mu_i(T,\upsilon^{-1}x)$. Their derivatives with respect to $T$ and $\upsilon$ read
\begin{equation}\label{Athermo17}
\partial_T( \bar \mu_i-\bar \mu_N) =-\upsilon \, \partial_{x_i}(\bar\varrho\bar s)(T,\upsilon,x),\qquad \partial_{\upsilon} (\bar \mu_i-\bar \mu_N) =- \partial_{x_i}\bar p(T,\upsilon,x)~.
\end{equation}
Once again the important role of the thermal equation of state shows up. It determines the volume dependence of the chemical potentials as well.

The equations of this paragraph constitute the basis to derive representation theorems for the Helmholtz free energy as a function of the variables $T,\upsilon,x$. These representation theorems must be consulted to answer the question which data can be used to construct a free energy function consistently.

\vspace{0.2cm} {\bf Consequences for the variables $(T,p,x)$.} Next we apply the transformation identities \eqref{Athermo12a} to the thermodynamic equations \eqref{Athermo13}-\eqref{Athermo17}.
However, with respect to the variables $(T,p,x)$ the most suited potential is not the Helmholtz free energy density but the specific Gibbs free energy $g :=\psi+p~\varrho^{-1}$. The following results are represented in an analogous manner as before.

Consider the Gibbs free energy function $M g=M(x)\hat g(T,p,x):=M(x)\hat\psi(T,p,x)+p~\hat \upsilon(T,p,x)$ and obtain from \eqref{Athermo1}-\eqref{Athermo3}
\begin{align}\begin{split}\label{Athermo13b}
\partial_T (M\hat g) =& -M(x)  \hat s(T,p,x),\qquad\partial_p (M\hat g) =\hat\upsilon(T,p,x) \, ,\\
& \partial_{x_i} (M \hat g) =\hat\mu_i(T,p,x)-\hat\mu_N(T,p,x)~.
\end{split}
\end{align}
Note that here the potential property of the free energy function concerns $Mg$ rather than $\varrho g$ as one would expect at first glance.

In the $(T,p,x)$ setting, the heat density, which is also called enthalpy density, $\varrho h:=\varrho u+p$ is more important than the internal energy density $\varrho u$. We calculate the enthalpy according to $M h=M(x)\hat h(T,p,x):=M(x) u(T,\hat\upsilon(T,p,x),x)+p~\hat\upsilon(T,p,x)$. The derivative of the specific enthalpy with respect to $T$,
\begin{equation}\label{thermo14b}
    \hat c_p:= \partial_T \hat h(T,p,x)
\end{equation}
defines the specific heat (at constant pressure). Its difference to the specific heat at constant volume,\eqref{Athermo14}, is easily calculated by means of \eqref{Athermo12a}$_1$, to the result
\begin{equation}\label{Athermo14c}
    c_p-c_\upsilon=-\frac{T}{M(x)} \, (\partial_T \bar p)^2 \, \partial_p \hat\upsilon \, .
\end{equation}
Corresponding to the equations \eqref{thermo15}, we now have
\begin{equation}\label{Athermo15b}
    \partial_p (M \, \hat h) =\hat \upsilon(T,p,x)- T\, \partial_T\hat \upsilon(T,p,x)\qquad\text{implying}\qquad
    M \, \partial_p \hat c_p=-T\, \partial^2_{T,T}\hat \upsilon(T,p,x)~.
\end{equation}
In the $(T,p,x)$ variables, we observe a corresponding behavior to the $(T,\upsilon,x)$ setting: The pressure dependence of the enthalpy function and the specific heat as well is already given by the thermal equation of state, now represented by $\hat\upsilon(T,p,x)$.

This is also true of the specific entropy, which is given by $s=\hat s(T,p,x):=\bar s(T,\hat\upsilon(T,p,x),x)$. It satisfies
\begin{equation}\label{Athermo16b}
\partial_T \hat s =\frac{\hat c_p(T,p,x)}{T},\quad\partial_p \hat s =-\frac{1}{M(x)} \, \partial_T\hat\upsilon(T,p,x),\quad
\hat s=-\frac{1}{M(x)} \, \sum_{j=1}^N
\partial_T \hat\mu_j \, x_j~.
\end{equation}

Finally, we define the $(T,p,x)$-representation of the chemical potentials, viz.\ $$\mu_i=\hat \mu_i(T,p,x):=\bar \mu_i(T,\hat\upsilon(T,p,x),x) \, .$$ The functions $\hat \mu_i(T,p,x)$ satisfy
\begin{equation}\label{Athermo17b}
\partial_T (\hat \mu_i-\hat \mu_N) =-\partial_{x_i} (M(x)\hat s)(T,p,x),\qquad\partial_p (\hat \mu_i-\hat \mu_N) = \partial_{x_i}\hat \upsilon(T,p,x)~.
\end{equation}

The equations of this paragraph constitute the basis to derive representation theorems for the Gibbs free energy function.

\vspace{0.2cm} {\bf Inequalities.} It is useful to list and prove the relevant inequalities between some of the introduced quantities. From the inequalities \eqref{thermo8d} (cf.\ \eqref{Athermo4}), we may derive further inequalities, viz.
\begin{equation}\label{Athermo17c}
\partial_{\upsilon} \bar p <0~,\qquad\hat c_p>\bar c_\upsilon>0~,\qquad\text{and}\quad
(\partial_T \hat \upsilon)^2\leq - \frac{\hat c_p \, M}{T} \, \partial_p\hat \upsilon \, .
\end{equation}
The proof of \eqref{Athermo17c}$_1$ starts from the Gibbs-Duhem equation \eqref{Athermo3}$_1$ for $p=\bar p$. Differentiation with respect to $\upsilon$ yields, after some simple rearrangements,
\begin{equation}\label{Athermo17d}
\partial_{\upsilon} \bar p =-\frac{1}{\upsilon^3}\sum_{i,j=1}^Nx_i\frac{\partial \mu_j}{\partial n_i} \, x_j<0~,
\end{equation}
where the inequality in \eqref{Athermo17d} is due to \eqref{Athermo4}$_1$.

The proof of $\bar c_\upsilon>0$ relies on the inequality \eqref{thermo8d}$_2$ and chooses $ f= u$ and $\bar f=\bar u$ in the identity \eqref{Athermo12}$_1$. Then, together with \eqref{Athermo17d}, the inequality $\hat c_p>\bar c_\upsilon$ is a direct consequence of \eqref{Athermo14c}.

\vspace{0.2cm} {\bf Free energy representations in general.} Two crucial facts can be read off from the experimental literature:
\begin{itemize}
\item[(i)] Exclusively the functions
\begin{equation}\label{Athermo18}
\upsilon=\hat \upsilon(T,p,x),\qquad c_p=\hat c_p(T,p,x),\qquad\mu_i=\hat \mu_i(T,p,x)~.
\end{equation}
are directly accessible in experiments.
\item[(ii)] The most simple task is the measurement of the thermal equation of state \eqref{Athermo18}$_1$. Calorimetric experiments to determine the specific heats \eqref{Athermo18}$_{1,2}$  are much more involved than pressure-volume-mole fraction measurements. The most complex procedure is needed for the experimental determination of the chemical potentials \eqref{Athermo18}$_3$.
\end{itemize}
Based on these facts we \underline{assume} that the thermal equation of state, i.e. \eqref{Athermo18}$_1$, is completely given. In other words, we know the function \eqref{Athermo18}$_1$ with respect to all variables $(T,p,x)$.

Then from \eqref{Athermo15b}$_2$ we conclude that the specific heat \eqref{Athermo18}$_2$ must be measured with respect to the variables $(T,x)$ for only one pressure value $p=p^0$. Below we see that the chemical potentials must solely be measured with respect to the mole fractions $x$ for a single pair $(T=T^0,p=p^0)$.

\vspace{0.2cm} {\bf Representation theorem for the specific entropies $s=\bar s(T,\upsilon,x)$ and $s=\hat s(T,p,x)$.} The representations of the free energy functions need in advance a corresponding representation of the specific entropies. It reads
\begin{equation}\label{Athermo19}
\bar s(T,\upsilon,x) =\frac{1}{M(x)}\int_{\upsilon_0}^\upsilon \partial_T\bar p(T,V,x) \, dV +
    \int_{T^0}^T\frac{\bar c_\upsilon(\vartheta,\upsilon_0,x)}{\vartheta}d\vartheta+\bar s(T^0,\upsilon_0,x)~.
\end{equation}
We conclude that $\bar s(T,\upsilon,x)$ can be calculated via the complete thermal equation of state $\bar p(T,\upsilon,x)$ and the specific heat $\bar c_\upsilon(T,\upsilon_0,x)$ for one value of the specific volume. There remains to determine the composition dependence of the specific entropy at a single pair $(T^0,\upsilon_0)$.

For $s=\hat s(T,p,x)$ we obtain
\begin{equation}\label{Athermo19a}
\hat s(T,p,x) =-\frac{1}{M(x)} \, \int_{p^0}^p\partial_T\hat\upsilon(T,p^{\prime},x) \, dp^{\prime} +
    \int_{T^0}^T\frac{\hat c_p(\vartheta,p^0,x)}{\vartheta}d\vartheta+\hat s(T^0,p^0,x)~.
\end{equation}

The proof of \eqref{Athermo19a} starts with \eqref{Athermo16b}$_2$, yielding
\begin{equation}\label{Athermo20}
\hat s(T,p,x)=-\frac{1}{M(x)}\int_{p^0}^p\partial_T\hat\upsilon(T,V,x) \, dV+\bar s(T,\upsilon_0,x)~.
\end{equation}
Next we calculate $\bar s(T,\upsilon_0,x)$ by means of \eqref{Athermo16b}$_1$, resulting in
\begin{equation}\label{Athermo21}
\bar s(T,\upsilon_0,x)=\int_{T^0}^T \partial_{\vartheta} \bar c_\upsilon(\vartheta,\upsilon_0,x) \, d \vartheta+\bar s(T^0,\upsilon_0,x)~,
\end{equation}
which is inserted into \eqref{Athermo19a}. In an analogous manner we may prove \eqref{Athermo19}.

These representations for the specific entropies are among the building blocks for the construction of the free energy functions.

\vspace{0.2cm} {\bf Representation theorems for the free energy functions $\varrho\psi=\bar \varrho(\upsilon,x)\bar\psi(T,\upsilon,x)$ and $Mg=M(x)\hat g(T,p,x)$.} We show that the free energy functions can be constructed from the functions
\begin{alignat}{2}\label{Athermo22}
    &\bar p(T,\upsilon,x),\qquad \bar c_\upsilon(T,\upsilon_0,x),\qquad\bar \mu_i(T^0,\upsilon_0,x) & \nonumber\\
    \text{or} \qquad & & \quad \\
    &\hat\upsilon(T,p,x),\qquad \hat c_p(T,p^0,x),\qquad\hat\mu_i(T^0,p^0,x)\quad  \text{respectively.}\nonumber
\end{alignat}
The representation for the free energy function $\varrho\bar\psi(T,\upsilon,x)$ reads
\begin{align}
    \varrho\bar \psi(T,\upsilon,x) =-\frac{1}{\upsilon}\int_{\upsilon_0}^\upsilon\bar p(T,V,x)dV &-\frac{M(x)}{\upsilon}
    \int_{T^0}^Td\vartheta\int_{T^0}^\vartheta d\vartheta'\frac{\bar c_\upsilon(\vartheta',p^0,x)}{\vartheta'}\nonumber\\
      &+\frac{M(x)}{\upsilon}\left(\bar\psi(T^0,\upsilon_0,x)-(T-T^0)\bar s(T^0,\upsilon_0,x)\right)~.\label{Athermo23}
\end{align}
This representation of the free energy function makes essential use of the assumption that the constitutive functions \eqref{Athermo22} are known, either by experiments or from microscopic modelling.

To prove \eqref{Athermo23}, we start with \eqref{Athermo13}$_2$. Integration with respect to $\upsilon$ yields
\begin{equation}\label{Athermo24}
\bar \psi(T,\upsilon,x)=\bar \psi(T,\upsilon_0,x)-\frac{1}{M(x)}\int_{\upsilon_0}^\upsilon\bar p(T,V,x)dV~.
\end{equation}
Then we use \eqref{Athermo13}$_1$ to analogously determine $\bar \psi(T,\upsilon_0,x)$. The result is
\begin{equation}\label{Athermo25}
\bar \psi(T,\upsilon_0,x)=\bar \psi(T^0,\upsilon_0,x)-\int_{T^0}^T\bar s(\vartheta,\upsilon_0,x)d\vartheta~.
\end{equation}
The term $\bar \psi(T^0,\upsilon_0,x)$ is calculated from \eqref{Athermo3}$_1$, and reads
\begin{equation}\label{Athermo26}
\bar \psi(T^0,\upsilon_0,x)=-\frac{\upsilon_0}{M(x)}\bar p(T^0,\upsilon_0,x)+\frac{1}{M(x)}\sum_{j=1}^N\bar\mu_j(T^0,\upsilon_0,x)x_j~.
\end{equation}
These three intermediate results are inserted into each other. The resulting equation is finally multiplied by $\varrho$, which completes the proof of the construction of the free energy function \eqref{Athermo23}.

The representation for the free energy function $M(x)\hat g(T,p,x)$ reads\begin{align}
    M\hat g(T,p,x) =\int_{p^0}^p\hat \upsilon(T,p^{\prime},x)dp^{\prime} &-M(x)
    \int_{T^0}^T d\vartheta'\int_{T^0}^{\vartheta'} d\vartheta\frac{\hat c_p(\vartheta,p^0,x)}{\vartheta}\nonumber\\
      &-(T-T^0)M(x)\hat s(T^0,p^0,x)+M(x)\hat g(T^0,p^0,x).\label{Athermo40}
\end{align}
There are two options to express the second line. We may assume that the specific entropy $s(T^0,p^0,x)$ and the specific enthalpy $h(T^0,p^0,x)$ are given. Then we have
\begin{equation}\label{Athermo41}
-(T-T^0)M(x)\hat s(T^0,p^0,x)+M(x)\hat g(T^0,p^0,x)=-T M(x)\hat s(T^0,p^0,x)+M(x)\hat h(T^0,p^0,x)~.
\end{equation}
Secondly, if we know the chemical potentials $\mu_i(T^0,p^0,x)$ and their derivatives $\partial_T\mu_i(T^0,p^0,x)$, the functions $s(T^0,p^0,x)$ and $g(T^0,p^0,x)$ may be substituted by
\begin{equation}\label{Athermo42}
\hat s(T^0,p^0,x)=-\sum_{i=1}^N \partial_T\mu_i(T^0,p^0,x) \, x_i\qquad\textrm{and}\qquad
\hat g(T^0,p^0,x)=\sum_{i=1}^N\mu_i(T^0,p^0,x) \, x_i~.
\end{equation}
%
To prove the representation \eqref{Athermo40}, at first we integrate \eqref{Athermo13b}$_2$ with respect to the volume. Then \eqref{Athermo13b}$_1$ at $(p^0,\, x)$ is integrated with respect to temperature. This leads to \eqref{Athermo40}. The first option \eqref{Athermo41} simply follows from the decomposition of the Gibbs free energy \eqref{Athermo40} into entropy and enthalpy. The second option relies on the Gibbs-Duhem equation \eqref{Athermo3} specialized to the Gibbs free energy and \eqref{Athermo16b}$_3$.

\section{Ideal mixture}\label{idmix}
Assume that the chemical potentials obey the additive ansatz
\begin{align}\label{celui}
\hat{\mu}_i(T, \, p, \, x) = g_i(T, \, p) + a_i(T, \, x_i)\text{ for } i =1,\ldots,N \, .
\end{align}
Here we consider given regular functions $g_i(T, \cdot): \, ]p_{\inf}, \, +\infty[ \rightarrow \mathbb{R}$ and $ a_i(T,\cdot): \, ]0, \, +\infty[ \rightarrow \mathbb{R}$.

First using that $\mu_i = \partial_{n_i}\varrho\psi$, we obtain via differentiation of \eqref{celui} that
\begin{align}\label{diffcelui}
 \partial^2_{n_i,n_j}\varrho\psi = \partial_p g_i(T, \, p) \, \partial_{n_j} p + a_i^{\prime}(T, \, x_i) \, \partial_{n_j}x_i \, .
\end{align}
We multiply with $n_j \, n_i$ and sum over $i,j=1,\ldots,N$. We have $\sum_j n_j \, \partial_{n_j}x_i = 0$ and, moreover, $\sum_{j} n_j \, \partial_{n_j} p = D^2_{n,n}(\varrho\psi) n \cdot n > 0$, we obtain that
\begin{align*}
 D^2_{n,n}(\varrho\psi) n \cdot n \, (1 - \sum_{i=1}^N \partial_p g_i(T, \, p) \, n_i) = 0 \Longleftrightarrow \sum_{i=1}^N \partial_p g_i(T, \, p) \, n_i = 1 \, .
\end{align*}
Now, we multiply \eqref{diffcelui} with $n_i$ and sum over $i=1,\ldots,N$. Since $\partial_{n_j} p = \sum_i \, n_i \, \partial^2_{n_i,n_j}\varrho\psi$, it follows that
\begin{align*}
 \sum_{i=1}^N n_i \, a_i^{\prime}(T, \, x_i) \, \partial_{n_j}x_i = 0 \text{ for all } j = 1,\ldots,N \, ,
\end{align*}
which we can rephrase as
\begin{align*}
x_j \, a_j^{\prime}(T, \, x_j) = \sum_{i=1}^N x_i^2 \, a_i^{\prime}(T, \, x_i) \text{ for all } j = 1,\ldots,N \, .
\end{align*}
From the latter relation we readily see that, necessarily,
\begin{align*}
& a_j^{\prime}(T, \, x_j) = \frac{c(T)}{x_j} \, \text{ for } j = 1,\ldots,N \, \\
\text{implying that } \quad & a_j(T, \, x_j) = c(T) \, \ln x_j + C_j(T) \, .
\end{align*}
With the molar based chemical potentials, we have the general relationship $\sum_i \hat{\mu}_i \, x_i = M(x) \, \hat{\psi} + p \, \hat{\upsilon}$. It implies that $\hat{\upsilon} = \sum_i x_i \, \partial_p \hat{\mu}_i$. For an ideal mixture, it thus follows that
\begin{align*}
 \sum_{i=1}^N \partial_pg_i(T, \, p) \, x_i = \hat{\upsilon}(T, \, p, \, x) \, .
\end{align*}
Hence we must require that
\begin{align*}
 \sum_{i = 1}^N \partial^2_{p,p} g_i(T, \, p) \, x_i < 0 \text{ for all } x \, ,
\end{align*}
which is possible only if $p \mapsto g_i(T, \, p)$ is strictly concave for all $i$.

 If $\sum_{i=1}^N x_i \, \partial^2_{p,p}g_{i}(T, \, p) < 0$ for all $x$, we introduce $p(T,n_1, \ldots,n_N)$ as the root of the equation $\sum_{i=1}^N n_i \, \partial_pg_{i}(T, \, p(T, \, n)) = 1$, and $\varrho\psi$ is, up to a function solely of temperature, defined via the formula
\begin{align*}
\varrho\psi = n \, \left(\sum_{i=1}^N x_i \,  g_i(T, \, p(T, \, n_1,\ldots,n_N)) + c(T) \, x \cdot \ln x +  C(T) \cdot x \right) - p(T, \, n_1,\ldots,n_N) \, ,
\end{align*}
in which $x_i = x_i(n_1,\ldots,n_N) = n_i/n$.

To verify that $\varrho\psi$ is concave in $T$, we compute the second derivatives, given as
\begin{align*}
\partial^2_{T,T}\varrho\psi = &\sum_{i=1}^N n_i \,  (c^{\prime\prime}(T) \, \ln x_i + C_i^{\prime\prime}(T)) \\
& + \sum_{i=1}^N n_i \, \partial^2_{T,T} g_i(T, \, p(T,n)) - \frac{(\sum_{i=1}^N n_i \, \partial^2_{T,p} g_i(T, \, p(T,n)))^2 }{\sum_{i=1}^N n_i \, \partial^2_{p,p} g_i(T, \, p(T,n)) } \ .
\end{align*}
Sufficient conditions for $\partial^2_{T,T} \varrho\psi < 0$ are, therefore, $c^{\prime\prime}(T) \geq 0$ and
\begin{align*}
 T \mapsto C_i(T), \quad (T, \, p) \mapsto g_i(T, \, p) \text{ concave for all } i = 1,\ldots,N \, .
\end{align*}
Up to the functions $c$ and $C$ of temperature, the free energy is completely determined from the assumption \eqref{celui}.

\section{Proof of the Propositions \ref{formula} and \ref{exi}} \label{detrop}

Defining $f$ as in the statement of Prop.\ \ref{formula}, the assumptions yield directly $f \in C^1(\mathbb{R}^N_+)$. By direct computation, use of $\partial_p\bar{V}(p(\rho), \, \rho) = V(p(\rho), \, \rho) = 1$ implies that
\begin{align}\label{gradh}
 \partial_{\rho_i} f(\rho) = & p^0 \, V_{\rho_i}(p^0, \, \rho) + k_{\rho_i}(\rho) + \bar{V}_{\rho_i}(p(\rho), \, \rho) + (\bar{V}_p(p(\rho), \, \rho) - 1)  \,  p_{\rho_i}(\rho) \nonumber\\
 = & p^0 \, V_{\rho_i}(p^0, \, \rho) +  k_{\rho_i}(\rho) + \bar{V}_{\rho_i}(p(\rho), \, \rho) \, .
\end{align}
Hence, both $k$ and $\bar{V}$ being positively homogeneous in $\rho$, we see that $$\rho \cdot \nabla_{\rho} f(\rho) = \rho \cdot \big(p^0 \, V_{\rho}(p^0, \, \rho) + k_{\rho}(\rho) + \bar{V}_{\rho}(p(\rho), \, \rho\big) = p^0 \, V(p^0, \, \rho) + k(\rho) + \bar{V}(p(\rho), \, \rho)\, .$$ Thus, $\rho \cdot \nabla_{\rho} f- f = p$, which is \eqref{probPbulk}. For $\rho$ satisfying $p(\rho) = p^0$, the definition of the pressure implies that $V(p^0, \,\rho) = 1$. Since $\bar{V}(p^0, \, \rho) = 0$ by definition, it holds that $f(\rho) = p^0 \, V(p^0, \, \rho) + k(\rho)  - p^0 = k(\rho)$. This verifies \eqref{probPbounda}.
In order to show that the solution is unique, assume that $f^1, \, f^2$ both solve \eqref{probPbulk}, \eqref{probPbounda}. We see that $f^1-f^2 =: d$ is positively homogeneous, and we must have $d(\rho) = 0$ whenever $p(\rho) = p^0$. For each fixed $\rho$, we have $V(p^0, \, \lambda \, \rho) = 1$ for $\lambda = 1/V(p^0, \, \rho)$. Thus, $p^0 = p(\lambda \, \rho)$, and $\lambda \, d(\rho) = d(\lambda \, \rho) = 0$ implies that $d \equiv 0$. This shows that $f^1 = f^2$.

Next we prove the additional properties of the solution $f$.

Since $V$ and $k$ are both of class $C^2(\mathbb{R}^N_+)$, the formula \eqref{gradh} shows that $\nabla_{\rho} f$ is differentiable, and that
\begin{align}\label{hessh}
&  \partial^2_{\rho_j,\rho_i} f = p^0 \, V_{\rho_j,\rho_i}(p^0, \, \rho) + k_{\rho_j,\rho_i}(\rho) + \bar{V}_{\rho_j,\rho_i}(p(\rho), \, \rho) + V_{\rho_i}(p(\rho), \, \rho)  \,  p_{\rho_j}(\rho)\\
& \qquad = p^0 \, V_{\rho_j,\rho_i}(p^0, \, \rho) + k_{\rho_j,\rho_i}(\rho) + \bar{V}_{\rho_j,\rho_i}(p(\rho), \, \rho) - \frac{1}{V_p(p(\rho), \,\rho)} \, V_{\rho_i}(p(\rho), \, \rho)  \,   V_{\rho_j}(p(\rho), \, \rho) \, .\nonumber
\end{align}
Here, we also used the equation $V(p(\rho), \, \rho) = 1$ to compute $\partial_{\rho}p = -V_{\rho}/V_p$. Thus, if the Hessians $ D^2_{\rho,\rho}(p^0 \, V(p^0) + k + \bar{V}(p))(\rho)$ are positive semi-definite for all $p \in ]p_{\inf}, \, p_{\sup}[$, it follows that $D^2 f$ is positive semi-definite.

Since $\rho \mapsto p^0 \, V(p^0, \, \rho) + k(\rho) + \bar{V}(p, \, \rho)$ is positively homogeneous, these Hessians possess necessarily a kernel containing $\{\rho\}$. We can verify for arbitrary $\eta \in S^2_+$ and $p = p(\rho)$ that
\begin{align}\label{darstellung}
  D^2_{\rho,\rho} f(\rho) \eta \cdot \eta = (p^0 \, V_{\rho,\rho} (p^0, \, \rho) + k_{\rho,\rho} (\rho) + \bar{V}_{\rho,\rho} (p, \, \rho)) \eta \cdot \eta - \frac{1}{V_p(p, \, \rho)} \, |V_{\rho}(p, \, \rho) \cdot \eta|^2 \, .
\end{align}
Choose now an orthonormal basis $\{\xi^1,\ldots,\xi^{N-1}, \, \rho/|\rho|\}$. For $\rho \in \mathbb{R}^N_+$ and $p > p_{\inf}$, we can verify that the $N-1$ vectors $\xi^i - (V_{\rho}(p, \, \rho)\cdot \xi^i/V(p, \, \rho)) \, \rho$, together with $\rho$, form again a basis of $\mathbb{R}^N$. For any $\eta \in \mathbb{R}^N$, we thus find coefficients $a_1, \ldots, a_{N-1}$ and $b$ to represent $$ \eta = \sum_{i=1}^{N-1} a_i \, (\xi^i - \frac{V_{\rho}(p, \, \rho)\cdot \xi^i}{V(p, \, \rho)} \, \rho) + b \, \rho\,.$$
Hence, we find a number $c = c(\rho, \, p)$ such that $|\eta| \leq c \, (|a|+|b|)$. Define the matrix $\Pi \in \mathbb{R}^{N\times (N-1)}$ with columns given by the vectors $\xi^1,\ldots,\xi^{N-1}$. Using \eqref{darstellung}, it is readily seen that
\begin{align}\label{plusloin}
&  D^2_{\rho,\rho} f(\rho) \eta \cdot \eta =  \Pi^{\sf T}(p^0 \, V_{\rho,\rho}(p^0, \, \rho)  + k_{\rho,\rho}(\rho) + \bar{V}_{\rho,\rho}(p, \, \rho)) \Pi \, a \cdot a - \frac{(V(p, \, \rho))^2}{V_p(p, \, \rho)} \, |b|^2 \nonumber \\
& \qquad \geq  \lambda_{\min}(\Pi^{\sf T}(p^0 \, V_{\rho,\rho}(p^0, \, \rho)  + k_{\rho,\rho}(\rho) + \bar{V}_{\rho,\rho}(p, \, \rho)) \Pi) \, |a|^2 - \frac{(V(p, \, \rho))^2}{V_p(p, \, \rho)} \, |b|^2 \, .
 \end{align}
If $p^0 \, V_{\rho,\rho}(p^0, \, \rho)  + k_{\rho,\rho}(\rho) + \bar{V}_{\rho,\rho}(p, \, \rho)$ possesses $N-1$ strictly positive eigenvalues, it is positive definite on $\{\rho\}^{\perp}$. Hence, the choice of the $\xi^i$ implies that $\lambda_{\min}(\Pi^{\sf T} D^2(p^0 \, V(p^0, \, \rho)  + k(\rho) + \bar{V}(p, \, \rho)) \Pi) > 0$. It follows from \eqref{plusloin} that $ D^2f(\rho) \eta \cdot \eta \geq c \, |\eta|^2$, showing that $D^2f$ is strictly positive definite a every point of $\mathbb{R}^N_+$.
%

In particular, $f$ is strictly convex. To show that $f$ is of Legendre--type, it remains to prove that $f$ is essentially smooth. Consider an arbitrary sequence $\{\rho^m\} \subset \mathbb{R}^N_+$ such that $\rho^m \rightarrow \rho \in \partial \mathbb{R}^N_+$. The {\bf first case} is $\rho = 0$. In this case the equations $V(p(\rho^m), \,\rho^m) = 1$ imply for the fractions $y^m := \rho^m/|\rho^m|_1$ that $V(p(\rho^m), \,y^m) = 1/|\rho^m|_1$. In turn, this yields $\limsup_{m \rightarrow \infty} p(\rho^m) = p_{\inf}$. Otherwise, we would find a subsequence, a number $p_1 > p_{\inf}$, and $y \in \overline{S}^1_+$, such that $V(p_1, \, y) = \lim_{k\rightarrow \infty} V(p(\rho^{m_k}), \, y^{m_k}) = + \infty$ in contradiction to the fact that $V \in C(]p_{\inf}, \, p_{\sup}[ \times \overline{S}^1_+)$. Next, using the homogeneity of degree zero of $V_{\rho}$ and $k_{\rho}$, we compute
\begin{align*}
 y^m \cdot \nabla_{\rho} f(\rho^m) = & y^m \cdot (p^0 \, V_{\rho}(p^0, \, \rho^m) + k_{\rho}(\rho^m) + \bar{V}_{\rho}(p(\rho^m), \, \rho^m))\\
 = & p^0 \, V(p^0, \, y^m) + k(y^m) + \bar{V}(p(\rho^m), \, y^m) \\
 \leq &  \sup_{y \in S^1_+} (p^0 \, V(p^0, \, y) + k(y) + \bar{V}(p(\rho^m), \, y)) \, .
\end{align*}
The assumption \eqref{suprem} now implies that $\limsup_{m\rightarrow \infty} y^m \cdot \nabla_{\rho} f(\rho^m)  = - \infty$.
Clearly, it follows that $|\nabla_{\rho} f(\rho^m)| \rightarrow + \infty$. The {\bf second case} is $\rho \neq 0$, so that the sequence $\{\rho^m\}$ is uniformly bounded and bounded away from zero. The equations $V(p(\rho^m), \, \rho^m) = 1$ imply that $p(\rho^m)$ remains in a compact set $[p_1, \, p_2] \subset ]p_{\inf}, \, p_{\sup}[$. Moreover, since $|\rho^m|_1 \geq c_0 > 0$, we see that the fractions $y^m = \rho^m/|\rho^m|_1$ must tend to a boundary point of $S^1_+$. Hence
\begin{align*}
 |\nabla_{\rho} f(\rho^m)| \geq \inf_{p \in ]p_1, \, p_2[} |p^0 \, V_{\rho}(p^0, \, y^m) +  k_{\rho}(y^m) + \bar{V}_{\rho}(p, \, y^m)| \, ,
\end{align*}
and in view of the assumption \eqref{condition4prime}, it follows that $\liminf_{m\rightarrow \infty} |\nabla_{\rho} f(\rho^m)| = + \infty$. This concludes the proof that $f$ is essentially smooth.

We now discuss the conditions in order that $f$ is co-finite. Recall that $p_{\sup} = + \infty$ in the statement of Prop.\ \ref{exi}. First we show that $|\nabla_{\rho} f(\rho^m)| \rightarrow + \infty$ for $|\rho^m|\rightarrow +\infty$. Since $V(p(\rho^m), \,y^m) = 1/|\rho^m|_1$, we find that $\liminf_{m \rightarrow \infty} p(\rho^m) = +\infty$. Otherwise $V(p_1, \, y) = \lim_{k\rightarrow \infty} V(p(\rho^{m_k}), \, y^{m_k}) = 0$ for some subsequence, finite $p_1$ and $y \in \overline{S}^1_+$. Thus
\begin{align*}
 y^m \cdot \nabla_{\rho} f(\rho^m) = &  p^0 \, V(p^0, \, y^m) + k(y^m) + \bar{V}(p(\rho^m), \, y^m)\\
 \geq & \inf_{y \in S^1_+} (p^0 \, V(p^0, \, y) + k(y) + \bar{V}(p(\rho^m), \, y)) \, .
\end{align*}
In view of the assumption \eqref{inf2}, we infer $\liminf_{m\rightarrow \infty} |\nabla_{\rho} f(\rho^m)| = + \infty$.

Consider the set $\nabla_{\rho} f(\mathbb{R}^N_+) := \{\mu \in \mathbb{R}^N \, : \, \mu = \nabla_{\rho} f(\rho) \text{ for some } \rho \in \mathbb{R}^N_+\}$. This set is closed. Indeed, if $\{\mu^m\} \subset \nabla f(\mathbb{R}^N_+)$ converges to some $\mu \in \mathbb{R}^N$, then any sequence $\{\rho^m\} \subset \mathbb{R}^N_+$ such that $\mu^m = \nabla_{\rho} f(\rho^m)$ must remain in a compact of $\mathbb{R}^N_+$. Otherwise, we can find a subsequence such that either $\{\rho^{m_k}\}$ tends to a boundary point of $\mathbb{R}^N_+$ or to $+\infty$. Then the previous arguments show that $|\mu^{m_k}| = |\nabla_{\rho} f(\rho^{m_k})|$ tends to infinity, in contradiction to the fact $\mu^{m_k} \rightarrow \mu$. Thus $\nabla_{\rho} f(\mathbb{R}^N_+)$ is closed. Hence we infer that $\nabla_{\rho} f(\mathbb{R}^N_+)$ does not possess any boundary points. For if $\mu$ would be such a boundary point, then $\mu \in  \nabla_{\rho} f(\rho)$ for some $\rho \in \mathbb{R}^N_+$, and the inverse mapping theorem implies that there are neighborhoods of $\rho$ in $\mathbb{R}^N_+$ and of $\mu \in \mathbb{R}^N$ such that $\nabla_{\rho} f$ is a bijection therein.
\end{document}